%% file: main.tex
\documentclass[opre,unblindrev]{informs3}
\RequirePackage{endnotes}
\usepackage{endnotes}

\OneAndAHalfSpacedXI

\usepackage{natbib}
\bibpunct[, ]{(}{)}{,}{a}{}{,}%
\def\bibfont{\small}%
\usepackage{endnotes}
\usepackage{enumerate}
\usepackage{rotating}
\usepackage{subfigure}
\usepackage{subcaption}
\usepackage{float}
\usepackage{mathtools}
\usepackage{booktabs}
\usepackage{multirow}
\usepackage{multibib}
\newcites{ec}{References}
\usepackage[svgnames]{xcolor}
\usepackage{caption}
\usepackage{amsmath}
\usepackage{graphicx}
\usepackage{enumitem}
\usepackage{multirow}
\usepackage{makecell}
\usepackage{comment}
\usepackage{threeparttable}
\usepackage{hhline}
\usepackage{algpseudocode}
\usepackage{mathrsfs}  
\usepackage{physics}
\usepackage{bbm}
\usepackage{svg}
\usepackage{bm}
\usepackage{mathdots}
\usepackage{bbding}

\graphicspath{ {./Fig/} }
\allowdisplaybreaks
\newcommand{\qedsymbol}{$\square$}
\definecolor{DarkBlue}{rgb}{0,0.08,0.45}
\usepackage[colorlinks=true,breaklinks=true,bookmarks=true,urlcolor=DarkBlue,citecolor=DarkBlue,linkcolor=DarkBlue,bookmarksopen=false,draft=false]{hyperref}  
\usepackage[noabbrev, nameinlink, capitalize]{cleveref}
\crefname{assumption}{Assumption}{Assumptions}
\crefname{lemma}{Lemma}{Lemmata}
\crefname{theorem}{Theorem}{Theorems}
\crefname{corollary}{Corollary}{Corollaries}
\crefname{proposition}{Proposition}{Propositions}
\crefname{claim}{Claim}{Claims}
\crefname{algorithm}{Algorithm}{Algorithms}
\crefname{figure}{Figure}{Figures}
\crefname{remark}{Remark}{Remarks}
\crefname{section}{Section}{Sections}
\crefname{procedure}{Procedure}{Procedures}
\crefname{definition}{Definition}{Definitions}
\crefname{example}{Example}{Examples}
\crefname{table}{Table}{Tables}
\crefname{equation}{}{}

\TheoremsNumberedThrough  

\ECRepeatTheorems
\EquationsNumberedThrough 

\begin{document}
\RUNAUTHOR{Fu, Hu}
\RUNTITLE{Capacitated Spatiotemporal Matching}

\TITLE{Capacitated Spatiotemporal Matching}
 \ARTICLEAUTHORS{
 \AUTHOR{Mingyang Fu}
 \AFF{Rotman School of Management, University of Toronto 
 , \EMAIL{mingyang.fu@utoronto.ca}}
 \AUTHOR{Ming Hu}
 \AFF{Rotman School of Management, University of Toronto
 , \EMAIL{ming.hu@rotman.utoronto.ca}}
 }

\ABSTRACT{We study a spatiotemporal service matching problem in which demand, heterogeneous in location and time sensitivity/preference, is to be assigned to service stations. The planner seeks to maximize social welfare, defined as total service reward minus spatial and temporal costs, by optimally scheduling demand to stations and service time under processing capacity constraints. We formulate the problem as an optimal transport (OT) model that allows for both demand–capacity imbalance and endogenously unserved demand when service costs exceed rewards. Leveraging a barycenter-style decomposition, we reformulate the problem as a finite-dimensional convex optimization problem that generalizes semi-discrete OT and enables scalable computation. We characterize the geometry of optimal assignments, showing that spatial partitions correspond to generalized Laguerre cells. Temporally, we show that the structure of the optimal schedule depends on demand heterogeneity: when demand differs only in temporal cost sensitivity, higher-sensitivity demand is assigned service times closer to the common ideal time; when demand differs only in preferred times, the assignment is order-preserving with respect to preferred times. We further propose an envy-free, individually rational implementation of the optimal schedule using time-dependent pricing and a finite-slot mechanism with explicit bounds depending on the number of required slots. To illustrate the framework, we extend the classic Hotelling linear-city model on a line segment by incorporating a continuum of waiting-cost sensitivities, demonstrating how optimal spatial partitions vary with changes in sensitivity heterogeneity and reward. Finally, we conduct a numerical study of a vaccination-planning problem using publicly available 2021 data from the City of Toronto. The results show that joint spatiotemporal matching reduces total social cost by at least 3.24\% relative to other policies that separate location assignment and scheduling.
}


\maketitle
\section{Introduction}\label{sec:intro}
The efficient coordination of supply and demand across space and time is a central challenge in modern operations management. Examples range from e-commerce fulfillment to service assignment. In each case, planners must jointly account for two sources of loss: spatial inefficiency arising from matching demand to distant supply, and temporal misalignment arising from service scheduling under processing capacity constraints.

These challenges are vividly illustrated in the context of e-commerce fulfillment centers during peak periods. A large e-commerce platform may operate Inbound Cross-Docking (IXD) facilities that serve as the initial, high-speed sorting hub for seller inventory before it reaches fulfillment centers (FCs). Before peak events (e.g., Black Friday), IXDs face unloading‐dock bottlenecks as inbound shipments surge. Uncoordinated arrivals lead to congestion, stockouts, overtime, and seller complaints. A possible remedy is a time-slot allocation mechanism: the platform assigns appointments subject to facility capacity, prioritizing more delay-sensitive shipments into more desirable windows while delaying or pulling earlier, lower-priority shipments. By trading off each shipment’s value and time sensitivity/preference against facility capacity, the platform smooths shipment arrivals, eliminates peak-time congestion, and maximizes seller satisfaction.

A related challenge arises in the same retail environment once products are stocked: the platform must decide not only where to place units across fulfillment centers, but also which delivery promises to offer and which centers to use for each order. The spatial cost is borne by the retailer in shipping from a particular warehouse. The temporal cost, however, can be twofold. On the one hand, capacity constraints at a facility can lead to worker overload, creating costs for the firm such as overtime pay. On the other hand, processing bottlenecks lead to service delays, imposing a waiting cost on customers. This creates a difficult trade-off: for example, concentrating order fulfillment at a popular downtown center during the holiday season might minimize spatial costs, but it risks triggering high temporal costs by overloading the local workforce and simultaneously failing to meet delivery promises. Assigning order fulfillment to suburban warehouses might relieve this temporal pressure, but at the expense of higher transportation costs. This problem thus illustrates the same fundamental interplay of spatial and temporal frictions under capacity constraints.

Another example is charging planning for a centralized electric-vehicle (EV) fleet, such as that of an autonomous-vehicle company like Waymo. In this context, demands are EVs in the fleet that require charging, and service stations correspond to available charging stations. Each EV exhibits degrees of sensitivity and varying time preferences depending on its battery status. Due to this dual heterogeneity in sensitivities and preferences, EVs originating from the same location may be optimally assigned to different charging stations at different times, rather than simply being assigned to the nearest station.

Similar challenges arise in large-scale public service operations. In mass vaccination campaigns, each clinic can serve only a fixed number of people per hour, yet participants differ in the urgency of their need for service. Those at higher risk, such as the elderly or those with pre-existing conditions, may require earlier appointments, while others are more flexible. Unscheduled arrivals lead to long queues and under- or over-utilized sites, whereas coordinated assignment of individuals to specific time–location slots ensures timely service, balanced utilization, and reduced waiting costs. A comparable situation arises in large-scale elections: voting centers have limited numbers of booths and staff, and voters differ in their preferred times to cast ballots. Some seek early-morning convenience; others arrive after work. Without coordination, simultaneous arrivals generate long lines and discourage participation. By carefully balancing distance and time preferences, planners can improve voter scheduling across locations and times. These public-sector examples highlight that demand is not only spatially distributed but also shaped by heterogeneous time preferences, making coordinated spatiotemporal matching essential for efficient and equitable service delivery.

Addressing these challenges requires a framework that integrates spatial, temporal, and capacity constraints within a unified design of service systems. In this paper, we develop a theoretical framework for spatiotemporal service matching under capacity constraints, motivated by the applications outlined above. We consider a centralized planner that seeks to maximize social welfare (i.e., the total service rewards of matched services minus transportation and temporal deviation costs) by optimally assigning each demand to a service station and time slot, or, if doing so is too costly, deciding not to serve some demands. Each service station has a time-varying processing capacity, which limits the number of units it can serve in each time interval. Demands are heterogeneous in both spatial location and type, such as time preference and time sensitivity (a measure of urgency); specifically, each demand incurs an increasing penalty as its service time deviates from that demand’s ideal or due date, with the magnitude of the penalty increasing with time sensitivity. Our model captures these features in a continuum formulation, leveraging optimal transport (OT) theory to find the welfare-maximizing assignment of demand mass to service capacity. 

In contrast to classical partial OT models that allow for endogenously unserved demand, where transportation costs are fixed and exogenous, our framework introduces an endogenous temporal cost arising from capacity constraints. In the classical uncapacitated setting, the optimal matching follows a simple rule: each demand is either assigned to its nearest station if the reward exceeds the transportation cost, or remains unmatched. As the homogeneous service reward increases, the set of served customers expands monotonically, but no demand point ever switches to a farther station. Our model departs fundamentally from this structure. Here, capacity limits induce congestion, resulting in temporal costs that depend on the endogenous load at each station. As the service reward increases, more customers enter the system, amplifying congestion at popular stations and reshaping the total cost. Consequently, some customers may find it optimal to switch from a nearby but congested station to a more distant, less busy one. 

Moreover, our paper treats temporal cost as a consequence of scheduling decisions, rather than attributing it solely to stochastic queuing. This enables us to accommodate heterogeneous time preferences and tailor the matching policy accordingly. This approach aligns more closely with real-world observations, such as those in e-commerce fulfillment centers or clinic appointments. In these settings, spatial assignment is typically integrated with appointment-style scheduling, which stands in contrast to analyses based on First-Come, First-Served (FCFS) performance. As a result, our framework offers implementable prescriptions (e.g., time-based pricing or time-slot allocations) and managerial insights into capacity configuration.
We make the following contributions in this paper. 

{\it Methodology.} We formulate spatiotemporal matching as an OT problem allowing partial service subject to capacity constraints. We further propose a barycenter-style decomposition that clarifies the geometry of optimal plans and induces type-dependent spatial partitions. Most importantly, when the set of time-sensitivity types is finite, we reduce the problem to a finite-dimensional convex program \eqref{eq:STB_D}. This reformulation extends the classic semi-discrete OT to a spatiotemporal setting and significantly enhances tractability.
   
{\it Implementation and design.} 
Based on our finite-dimensional reduction \eqref{eq:STB_D}, we propose an implementation via time-dependent pricing that is envy-free and individually rational. We also design a finite-slot mechanism with bounds on the number of offered slots. Moreover, we study budget-constrained capacity allocation across stations and provide structural guidance for how to deploy limited capacity. We show that this joint matching and allocation problem can be reformulated as a tractable convex program.

{\it Managerial insights.} In \Cref{sec:opt_structure}, we characterize optimal matching in space and time. Spatially, service regions are weighted Laguerre cells with a no-match threshold. Specifically, the matched region is partitioned into station-specific cells, in which demands at each location are assigned to the station with the smallest weight-adjusted transportation cost, and demands outside the matched region are left unmatched. Temporally, when demands share a preferred time but deviation sensitivities differ, schedules are sensitivity-prioritized, and the served set shrinks as sensitivity increases; when sensitivities are the same but preferred times differ, schedules are order-preserving, in the sense that demands with earlier preferred times are assigned (weakly) earlier service times. In \Cref{Sec:Hotelling}, we extend the classic Hotelling model \cite{harold1929stability} to a spatiotemporal setting with heterogeneous time sensitivities, showing how optimal matching plans shift as reward or capacity varies.

Finally, we complement the theoretical analysis with a numerical study motivated by vaccination planning in Toronto during the COVID-19 rollout, involving two types of vaccination sites across 21 locations and 1.65 million individuals divided into 4 risk groups. We compare the optimal spatiotemporal matching policy with simpler benchmark policies that separate spatial assignment from temporal scheduling. The results show that joint optimization delivers the lowest total social cost, primarily by reducing temporal cost. In particular, the policy prioritizes elderly individuals, assigning them earlier time slots and directing them to more efficient sites. The computed solution also reflects our structural results: higher-risk populations are systematically scheduled earlier, and spatial partitions vary across risk groups rather than remaining fixed across groups. Specifically, high-capacity sites serve broader regions for higher-risk groups than for low-risk groups.

{\bf Organization of the paper.} The remainder of the paper is organized as follows. At the end of this section, we list the mathematical notations used in this paper. In \Cref{sec:lit_review}, we review the relevant literature streams on optimal transport theory, spatiotemporal matching, and resource allocation. In \Cref{sec:model}, we formally introduce the problem setting. We formulate the capacitated spatiotemporal matching problem as an OT model, present its dual formulation, and propose an alternative barycenter-style reformulation. We also characterize several structural properties of the optimal spatial and temporal assignments. \Cref{Sec:finite_type} analyzes the important case of finite demand types, reducing the problem to a finite-dimensional convex program. Based on this, we develop an envy-free implementation via time-dependent pricing and a finite-slot mechanism, and we also analyze the problem of capacity allocation. In \Cref{Sec:Hotelling}, we apply our framework to extend the classic Hotelling model, illustrating how optimal matching plans shift in response to changes in reward and capacity limits. In \Cref{sec:numerical}, we present a numerical study of a vaccination planning problem and compare the proposed policy with several benchmark rules. Finally, \Cref{sec:conclusion} concludes the paper. All proofs and additional analyses are deferred to the E-Companion.

{\bf Notation.} Throughout the paper, let $[n] = \{1, 2, \dots, n\}$ denote the set of integers from $1$ to $n$. 
We use $\mathbb{R}$ to denote the set of all real numbers, with $\mathbb{R}_{+}$ 
representing its subsets of non-negative numbers. 
For any Borel set $\mathcal{B}\subset\mathbb{R}^{n}$, the set of all positive integrable functions on $\mathcal{B}$ is denoted by $\mathcal{L}^+_1(\mathcal{B})$. Specifically, we equip each continuous component of $\mathcal{B}$ with the Lebesgue measure, each discrete (finite) component with the counting measure, and, in mixed cases, with the corresponding product of these measures. The maximum of two values $a$ and $b$ is denoted by $a \vee b = \max\{a, b\}$, and their minimum by $a \wedge b = \min\{a, b\}$. The indicator function is denoted by $\mathbbm{1}(\cdot)$, where $\mathbbm{1}(A) = 1$ if $A$ is true and $\mathbbm{1}(A) = 0$ otherwise. The non-negative part of a real number $x$ is defined as $(x)^+ = \max\{x, 0\}$. 

\section{Literature Review}\label{sec:lit_review}
Our research is related to several streams of literature: optimal transport theory, spatiotemporal matching, and resource allocation.

The theory of optimal transport (OT), originated by \cite{monge1781memoire} and later formalized by \cite{kantorovich1942translocation}, provides a powerful framework for comparing probability distributions by computing the minimum cost of transforming one into the other. 
As a classical mathematical tool (see modern expositions such as \citealt{villani2008optimal} and \citealt{santambrogio2015optimal}), it has recently emerged as a valuable tool in operations and supply chain management. A key contribution in this area is the application of OT to logistics and geographical partitioning problems (see, e.g., a recent review \citealt{ryzhov2024introduction}). Specifically, our research is closely related to unbalanced OT and semi-discrete OT. Unbalanced OT \citep{liero2018optimal,chizat2018scaling} relaxes mass conservation to capture unmatched demand. When continuous demand is assigned to a discrete set of facilities or time slots, semi-discrete OT provides a natural formulation \citep{bourne2018semi,tacskesen2023semi}. Classical constructs such as power diagrams and weighted Voronoi partitions \citep{aurenhammer1987power,okabe2009spatial} have been widely used to partition territories into regions associated with discrete points. These extensions closely parallel our admission and scheduling decisions.

On the operations management side, our research is related to the literature on the design of spatial and temporal service systems. This stream, dating back to Hotelling's seminal work (\citealt{harold1929stability}), examines how firms or facilities partition a market according to customers' geographic distribution. Much of this literature focuses on facility location and service partitioning \citep{eiselt1993competitive,haugland2007designing,carlsson2013dividing}. In public-service and healthcare applications, researchers have incorporated accessibility and endogenous capacity decisions \citep{castillo2009social,zhang2010bilevel,aboolian2016maximal,aboolian2022probabilistic,aboolian2023location,he2025spatial}. Recently, researchers have integrated congestion effects, recognizing that customer choice depends on both travel distance and expected waiting time. Many of these papers model the temporal cost using a queuing system, as in \cite{baron2008facility,dan2019competitive,besbes2022spatial,gunnar2024demand,abouee2025spatial}. While the spatial queuing literature typically treats waiting time as an stochastic outcome of a first-come, first-served queueing process, we model time scheduling as an explicit decision within a welfare-maximizing framework. This connects our work conceptually to the extensive literature on appointment scheduling (see, e.g., \citealt{wang2024appointment}), but with the crucial addition of the spatial dimension. This shift from performance analysis to optimal design allows us to explicitly prioritize service based on demand characteristics, such as time sensitivity, which is not available under a classic spatial queuing framework.

Another parallel, highly active stream studies spatiotemporal problems from an online perspective, with a particular focus on ride-hailing. Foundational work in online matching provides competitive-ratio benchmarks for sequential, irrevocable decisions under uncertainty \citep{karp1990optimal,mehta2007adwords}. Building on these ideas, the ride-hailing literature has developed real-time matching and pooling \citep{alonso2017demand,lowalekar2018online,barrientos2025online, bo2026spatialstaffing}. 
Broad surveys synthesize matching and dynamic-pricing controls \citep{yan2020dynamic,wang2019ridesourcing}. 
By contrast, our framework addresses a different, complementary set of questions. We adopt an offline, deterministic perspective not as a simplification, but as a deliberate choice to uncover the fundamental structure of an optimal plan under capacity constraints.

\section{Matching as Optimal Transport}\label{sec:model}
Optimal transport (OT) provides a natural framework for addressing the challenge of spatiotemporal matching. By treating the joint distribution of customer locations and types as a source measure and each service station’s capacity throughout the time horizon as a sink, we can frame the assignment of demand to stations as a transportation problem that balances travel costs against temporal deviation penalties. This OT-based reformulation integrates spatial and temporal trade-offs into a single convex optimization problem. 

Next, we first formulate our model from the perspective of optimal transport and present its primal and dual formulations. Additionally, we introduce an alternative formulation that involves selecting an intermediate measure to balance multiple OT objectives. This approach is analogous to the concept of a Wasserstein barycenter. Finally, we highlight the structural properties of the optimal matching that facilitate tractable analysis.

\subsection{Problem Setting}
Building on the motivating examples, we now formalize the key elements of our spatiotemporal matching problem. Our framework consists of two primary components: a heterogeneous demand population and a set of capacity-constrained service stations. 

Let $\mathcal X \subset \mathbb{R}^d$ with $d\ge 2$ denote the spatial domain, a compact set in $d$-dimensional space that contains all demand locations. We assume the time horizon of interest $\mathcal T \subset \mathbb{R}$ is a bounded interval (e.g., a fixed operating period). Within this region $\mathcal X$ and time horizon $\mathcal T$, a collection of service requests (demands) is to be served by a finite set of stations. Each demand is heterogeneous in location $\bm x \in \mathcal X$ and type $\alpha \in \mathcal A$, where $\mathcal{A}$ denotes the type space of demand, endowed with a $\sigma$-finite measure. The type space $\mathcal{A}$ can be either finite (with the counting measure) or continuous (a compact subset of $\mathbb{R}^{d_a}$ with Lebesgue measure). The parameter $\alpha$ captures the urgency or temporal tolerance of the demand. For instance, when $\mathcal A\subseteq\mathbb{R}$, a larger $\alpha$ may indicate that the request’s service value decays quickly if not served promptly.

We write the joint demand density on $\mathcal{X}\times\mathcal{A}$ as $\mu(\bm{x},\alpha)=\varrho(\alpha)\rho_\alpha(x)$. Specifically, $\varrho(\alpha)$ denotes the mass function of demand types on $\mathcal{A}$. We assume $\varrho(\alpha)\ge 0$ for all $\alpha\in\mathcal{A}$ when $\mathcal{A}$ is finite or $\varrho \in \mathcal{L}_+^1(\mathcal{A})$ when $\mathcal{A}$ is continuous.
Given demand type $\alpha$, we assume that the density function of demand with respect to location is $\rho_\alpha(\bm{x})\in\mathcal{L}_1^{+}\left(\mathcal{X}\right)$ with $\int_{\mathcal{X}}\rho_\alpha(\bm{x})\mathrm{d}\bm x =1$.

There are $n$ service stations, indexed by $i \in [n]$ and located at fixed positions $\{\bm{y}_i\}_{i=1}^n \subset \mathcal{X}$. The service operations occur over a bounded time horizon $\mathcal{T} \subset \mathbb{R}$. Each station $i$ possesses a time-varying processing capacity, captured by a function $c_i(t) > 0$ for $t \in \mathcal{T}$, representing the maximum volume of demand that station $i$ can process per unit of time at time $t$. This capacity may fluctuate over time due to factors such as staffing schedules, equipment availability, or operational costs. The total service supply of station $i$ over a time interval is the integral of $c_i(t)$ over that interval.


The value generated by serving a unit of demand arises from a homogeneous base reward for successful service and is reduced by spatial and temporal costs. The planner's decision is to design a matching plan that specifies, for each demand type $\alpha$ at location $\bm{x}$, how much demand each station $i$ will serve and at what time $t$, subject to each station's capacity limits. 

The objective is to maximize the total social welfare, i.e, the aggregate matching value of all served demand. This objective is common in service operations and matching models (see, e.g.,\citealt{besbes2022spatial,gunnar2024demand}). Although often motivated by public-service settings, it is equally natural in commercial contexts. In our motivating examples, an e-commerce platform such as Amazon must account for the welfare of all sellers to sustain a healthy marketplace. Moreover, in the EV-charging setting, the vehicles are owned by a single centralized company seeking to maximize overall system utility.  
We next provide a formal OT representation of this planning problem.

\subsection{Optimal Transport Model}
We model the matching problem as a classical OT problem. Assume that each unit type-$\alpha$ demand at location $\bm x$ gets served at station $i$ at time $t$ has the matching value
$$
r - \delta(\bm{x},\bm{y}_i) - \ell(t,\alpha),
$$
where $r>0$ is a standard matching reward homogeneous across demands, and $\delta(\bm{x},\bm{y}_i)\in\mathcal{L}^+_1(\mathcal{X}\times\mathcal{X})$ is the transportation cost per unit demand transferred from location $\bm{x}$ to station $i$ (with location $\bm{y}_i$). We assume that $\delta(\bm{x},\bm{y}_i) = \hat\delta(\|\bm{x}-\bm{y}_i\|_2)$, where $\|\bm{x}-\bm{y}_i\|_2$ is the Euclidean distance between $\bm{x}$ and $\bm{y}_i$ and $\hat\delta$ is a continuous and differentiable function on $\mathbb{R}^+$ with $\hat\delta(0) = 0$ and $\hat\delta^\prime(\cdot)>0$. Furthermore, we assume $\hat\delta$ is convex. $\ell(t,\alpha)\in \mathcal{L}^+_1 (\mathcal{T}\times\mathcal{A})$ is a temporal cost of type-$\alpha$ demand assigned with scheduled time $t$, which is assumed to be continuous in $t$. 

A matching plan is characterized by its density function $\pi\in\mathcal{L}_1^{+}\left(\mathcal{X}\times\mathcal{A}\times\mathcal{T}\times[n]\right)$, which represents the amount of type-$\alpha$ demand from location $\bm{x}$ assigned to be served at station $i$ at time $t$. Since the set of stations $[n]$ is discrete, we can simplify the notation by representing the plan as a vector of $n$ functions, $\{\pi_i\}_{i\in[n]}$, where $\pi(\bm{x},\alpha,t,i) = \pi_{i}(\bm{x},\alpha,t)$ for short. The set of all feasible plans, denoted as $\Gamma\left(\mu,\bm{c}\right)$ , is then defined as:
\begin{equation}
\Gamma\left(\mu,\bm{c}\right) \coloneqq  
\left\{\
\pi\in\mathcal{L}_1^+\left(\mathcal{X}\times\mathcal{A}\times\mathcal{T}\times[n]\right)
\left|\quad
\begin{aligned}
&\int_{\mathcal{X}\times\mathcal{A}}\pi_{i}(\bm{x},\alpha,t)\mathrm{d}(\bm{x},\alpha)\leq c_i(t), \quad  \forall t\in\mathcal{T}, \forall i\in[n],\\
&\sum_{i=1}^{n}\int_{\mathcal{T}}\pi_{i}(\bm{x},\alpha, t)\mathrm{d}t\leq \mu(\bm{x},\alpha),\quad \forall(\bm{x},\alpha)\in\mathcal{X}\times\mathcal{A}.\\
\end{aligned}
\right.
\right\}
\end{equation}
The feasible set is defined by two key constraints. The first inequality is the service capacity constraint, ensuring that the aggregate flow of all demand types assigned to station $i$ at time $t$ does not exceed its processing capacity $c_i(t)$. The second enforces demand capacity: no extra demand exceeding the available density at location $\bm x$ may be ``created," while some demand may remain unfulfilled. The planner's objective is to find the feasible plan that maximizes the total social welfare:
\begin{equation}\tag{STM}\label{eq:STM}
\begin{aligned}
\sup_{\pi\in\Gamma\left(\mu,\bm{c}\right) }&\quad  \sum_{i=1}^{n}\int_{\mathcal{X}\times\mathcal{A}\times\mathcal{T}}
\left(r- \delta(\bm{x },\bm{y}_i) - \ell(t,\alpha)\right)\pi_{i}(\bm{x},\alpha,t)\mathrm{d}(\bm{x},\alpha,t).\\
\end{aligned}
\end{equation}

Applying the standard Lagrangian duality to \eqref{eq:STM}, we write its dual problem as:  
\begin{equation}\tag{STM-D}\label{eq:STM_D}
\begin{aligned}
\inf_{\phi,\varphi_i}&\quad
\int_{\mathcal{X}\times\mathcal{A}} \phi(\bm{x},\alpha)\mu(\bm{x},\alpha)\mathrm{d}(\bm{x},\alpha) + \sum_{i=1}^{n}\int_{\mathcal{T}} \varphi_i(t)c_i(t) \mathrm{d}t  \\
\text{s.t.}&\quad 
\phi(\bm{x},\alpha) + \varphi_i(t) \geq r-\delta(\bm{x},\bm{y}_i) - \ell(t,\alpha),\quad\forall (\bm{x},\alpha,t)\in\mathcal{X}\times\mathcal{A}\times\mathcal{T}, \\
&\quad 
\phi(\bm{x},\alpha), \varphi_i(t)\geq 0,\quad\forall (\bm{x},\alpha,t)\in\mathcal{X}\times\mathcal{A}\times\mathcal{T}.
\end{aligned}
\end{equation}

Unlike the classical OT formulation, our problem neither enforces mass balance between supply and demand nor requires that every demand or supply be matched. Nevertheless, with appropriate adjustments to the primal and dual programs, one can still prove that strong duality holds.

\begin{proposition}\label{prop:general_Dual}
 The strong duality holds for the primal problem \eqref{eq:STM} and the dual problem \eqref{eq:STM_D}.
\end{proposition}

This result can be obtained by investigating an extended version of \eqref{eq:STM} as a balanced OT problem (supply equals demand) that includes additional dummy points (cf. \citealt{ekeland2010existence}, \citealt{caffarelli2010free}). 

The dual problem \eqref{eq:STM_D} is linear in terms of $\phi$ and $\varphi$, which play the role of generalized Kantorovich potentials. The function $\varphi_i(t)$ plays the role of the marginal value of capacity at station $i$ and time $t$, while $\phi(\bm x,\alpha)$ is the marginal value of admitting a unit of type $\alpha$ demand at $\bm x$. The constraint ensures that no assignment can create value beyond the sum of these marginal values (i.e., shadow prices of constraints).

\subsection{A Barycenter Perspective}
Note that in \eqref{eq:STM_D}, the Kantorovich potentials remain continuous functions of $\bm{x}$ and $t$, even when the demand types are finite. In this section, we present an alternative perspective on the spatial and temporal matching problem by reformulating it as the selection of an intermediate measure that jointly balances reward, and spatial and temporal costs. As demonstrated in \Cref{Sec:finite_type}, this reformulation enables a reduction to a finite-dimensional convex program when the demand types are finite. To this end, we introduce two optimal transport subproblems, $OT_{\mathrm{time}}$ and $OT_{\mathrm{space}}$. The first transports the measure of service capacity on $\mathcal{T}\times[n]$, and the second transports the demand measure on $\mathcal{X}\times\mathcal{A}$, both to an intermediate measure $\bm{q}$ defined on $[n]\times\mathcal{A}$.
\begin{equation*}
\begin{aligned}
&OT_{\mathrm{time}}(\bm{c},\bm{q}) =\\
\inf_{\nu}\quad
& \sum_{i=1}^{n}\int_{\mathcal{T}\times\mathcal{A}}
\ell(t,\alpha)\nu_{i}(t,\alpha)\,\mathrm{d}(t,\alpha), \\
\text{s.t.}\quad
& \int_\mathcal{A}\nu_{i}(t,\alpha)\le c_i(t),\quad \forall i,t, \\
& \int_{\mathcal{T}}\nu_{i}(t,\alpha)\,\mathrm{d}t
= q_i(\alpha),\quad \forall i,\alpha;
\end{aligned}
\qquad
\begin{aligned}
&OT_{\mathrm{space}}(\bm{q},\mu) =\\
\inf_{\theta}\quad
& \sum_{i=1}^{n}\int_{\mathcal{X}\times\mathcal{A}}
\delta(\bm{x},\bm{y}_i)\theta_{i}(\bm{x},\alpha)\,\mathrm{d}(\bm{x},\alpha), \\
\text{s.t.}\quad
& \sum_{i=1}^{n}\theta_{i}(\bm{x},\alpha)\le \mu(\bm{x},\alpha),\quad \forall \alpha,\bm{x}, \\
& \int_{\mathcal{X}}\theta_{i}(\bm{x},\alpha)\,\mathrm{d}\bm{x}
= q_i(\alpha),\quad \forall \alpha,i.
\end{aligned}
\end{equation*}

The problem can be viewed as the production and allocation of virtual goods, where each good $(i,\alpha)$ represents serving a unit of type-$\alpha$ demand at station $i$. The production corresponds to temporal scheduling: “producing” a portfolio $\{ q_i(\alpha)\}_{\alpha\in\mathcal A}$ at station $i$ incurs the minimum aggregate temporal cost from optimally matching demand type to the service time under station’s processing capacity $c_i(t)$. The allocation then assigns these goods to end-users, where demand $(\bm x,\alpha)$ “purchases” a good $(i,\alpha)$ by bearing spatial cost $\delta(\bm x,\bm y_i)$. The planner’s objective is to find market-clearing quantities that jointly optimize production and allocation, thus maximizing total social welfare. This decomposition is not merely illustrative but is also helpful to our reformulation. The following theorem formalizes this intuition: the original problem \eqref{eq:STM} is equivalent to optimizing over $\bm q$, the balancing measure linking these two transport problems.

\begin{theorem}\label{Thm:STB}
The problem \eqref{eq:STM} can be reduced to the following problem with respect to $q_i \in \mathcal{L}_1^+\left(\mathcal{A}\right), i\in[n]$:
\begin{equation}\tag{STB}\label{eq:STB}
\begin{aligned}
\sup_{\bm{q}}\quad& \sum_{i=1}^{n}\int_{\mathcal{A}}rq_i(\alpha)\mathrm{d}\alpha-OT_{\mathrm{time}}(\bm{c},\bm{q})-OT_{\mathrm{space}}(\bm{q},\mu),\\
\text{\textup{s.t.}}\quad& \sum_{i=1}^{n}q_i(\alpha)\leq \varrho(\alpha) \quad\forall\alpha\in\mathcal{A},\\
& \int_\mathcal{A} q_i(\alpha)\mathrm{d}\alpha \leq \int_{\mathcal{T}}c_{i}(t)\mathrm{d}t\quad\forall i\in[n].
\end{aligned}
\end{equation}
Moreover, the objective is concave in $\bm q$.
\end{theorem}

In fact, the reward $r$ can always be absorbed into either the temporal or spatial cost. For instance, define $\hat{\delta}(\bm{x},\bm{y}_i) = \delta(\bm{x},\bm{y}_i) -r$ and denote the corresponding subproblem as $\widehat{OT}_{space}(\bm{q},\mu)$. This adjustment does not alter the procedure for optimizing $\theta$ given $\bm q$ since it only adds a constant in the objective. Alternatively, one may interpret the first term as a special OT problem in which the transportation cost (i.e., the negative reward) is $-r$. Consequently, problem \eqref{eq:STB} decomposes (up to a sign reversal) naturally into two or three OT-type components. This decomposition is conceptually akin to the construction of a Wasserstein barycenter (see, e.g., \citealt{agueh2011barycenters, borgwardt2020improved}). In that framework, a barycenter is a measure that minimizes a weighted sum of transport costs (e.g., the Wasserstein distance) with respect to a set of given measures. Here, our decision $\bm q$ plays a similar role to a mediating plan that balances the sum of spatial and temporal costs. Accordingly, we refer to this reformulation as the spatiotemporal barycenter (STB).

The formulation \eqref{eq:STB} provides an alternative perspective on the problem. The decision variables are now $n$ functions $\{q_i\}_{i=1}^{n}$ on $\mathcal{A}$. By contrast, \eqref{eq:STM} optimizes over $n$ matching functions $\{\pi_i\}_{i=1}^{n}$ on $\mathcal{X}\times\mathcal{A}\times\mathcal{T}$, and \eqref{eq:STM_D} uses $n$ temporal potentials $\{\varphi_i\}_{i=1}^{n}$ on $\mathcal{T}$ plus a single potential $\phi$ on $\mathcal{X}\times\mathcal{A}$. This represents a significant reduction in the domain of the core decision variable. The trade-off is that the objective in \eqref{eq:STB} becomes more intricate, involving two subproblems $OT_{\mathrm{time}}$ and $OT_{\mathrm{space}}$. However, this alternative formulation is advantageous. As we will show in \Cref{Sec:finite_type}, when the set of demand types is finite, this barycenter structure is precisely what allows the problem to be reduced to a tractable, finite-dimensional convex program.

\subsection{Structure of Optimal Matching}\label{sec:opt_structure}

While formulations such as \eqref{eq:STM_D} and \eqref{eq:STB} are essential to our framework, they do not, by themselves, reveal the qualitative nature of the optimal allocation. To generate actionable managerial insights, we examine the underlying geometry and priority rules that characterize the optimal policy $\pi$. Specifically, we introduce several structural properties of the matching and show that, without loss of optimality, attention may be restricted to policies that satisfy these properties. 

Our first structural result characterizes the spatial property of the matching. It is based on generalized Laguerre cells, a weighted extension of the classical Laguerre cells (also known as Voronoi diagram). In this framework, each site is assigned a real-valued weight that modifies its effective influence on nearby points, so that the allocation depends not only on transportation cost but also on the site-specific priority measured by the weight. Our construction further extends this concept by incorporating a threshold parameter, which defines a residual ``no-match" region. 

\begin{definition}[Generalized Laguerre Cells]\label{def:GLC}
Given a transportation cost $\delta$ and a set of stations located at $\{\bm{y}_i\}_{i=1}^n\subset \mathcal{X}$. The generalized Laguerre cells are defined by a weight vector $\bm{\omega}\subset\mathbb{R}^{n}$ and a threshold $\omega_0$. Concretely, the \emph{generalized Laguerre cell} of site $i$ is the region
$
\mathcal{C}^{\omega_0}_i(\bm{\omega}) \coloneqq \left\{\bm{x}\in\mathcal{X}:\delta(\bm{x},\bm{y}_i) + \omega_i \leq \omega_0, \delta(\bm{x},\bm{y}_i) + \omega_i \leq\delta(\bm{x},\bm{y}_{\tilde{i}}) + \omega_{\tilde{i}}, \forall\,\tilde{i}\neq i\right\}
$. 
Moreover,
$
\mathcal{C}^{\omega_0}_0(\bm{\omega}) \coloneqq \left\{\bm{x}\in\mathcal{X}:\omega_0 \leq\delta(\bm{x},\bm{y}_{\tilde{i}}) + \omega_{\tilde{i}},\forall\,\tilde{i}\in[n]\right\}=\mathcal{X}\setminus\left(\cup_{i=1}^{n}\mathcal{C}_i(\bm{\omega})\right)
$
is the residual of $\mathcal{X}$.
\end{definition}

When $\omega_0 \rightarrow +\infty$ and $\bm\omega = \bm{0}$, the generalized Laguerre cells reduce to the classic unweighted Laguerre diagram. Because $\delta$ is strictly increasing in the Euclidean distance, $\mathcal{C}^{+\infty}_i(\bm{0})$ corresponds to the set of points $\bm{x}$ that are closer to station $i$ than to any other station. Conversely, a non-zero weight vector $\bm\omega$ introduces station-specific penalties, meaning that a point might be assigned to a more distant station if its associated weight is sufficiently lower to offset the increased spatial cost. 

A standard argument further shows that the boundaries of cells $\mathcal{C}^{\omega_0}_i(\bm{\omega})$ have Lebesgue measure zero.

\begin{lemma}\label{lemma:null_space}
For any $i,j\in\{0,1,\dots,n\}$ with $i\neq j$, $\mathcal{C}^{\omega_0}_i(\bm{\omega})\cap\mathcal{C}^{\omega_0}_j(\bm{\omega})$ has Lebesgue measure $0$ for all $\bm\omega$ and $\omega_0$.
\end{lemma}

Since the demand is continuous on $\mathcal{X}$ with a density, the total demand on these boundaries has measure zero and is therefore negligible. In what follows, we therefore ignore all geometric boundary effects without loss of generality.

Based on the generalized Laguerre cell, we show the following property of optimal matching.

\begin{proposition}\label{prop:Laguerre_cell}
Without loss of optimality, there exists an optimal matching for \eqref{eq:STB} whose spatial assignment is characterized by type-based generalized Laguerre cells. That is, there exists a function $\bm{\omega}(\alpha):\mathcal{A}\mapsto\mathbb{R}^{n}$ for each demand type~$\alpha$, such that for station $i$, any type $\alpha$ demand will be transported to station $i$ if and only if it is in the region $\mathcal{C}^{r}_i(\bm{\omega}(\alpha))$.
\end{proposition}

\Cref{prop:Laguerre_cell} shows that the service regions are not necessarily fixed partitions, but rather overlap across different demand types. Unlike classical OT or Voronoi partitions, in which every location has a unique serving facility, here the effective territory of a station depends on the demand type. This implies that service boundaries are fluid, adjusting in response to demand heterogeneity.

To analyze optimal time scheduling, we identify structural properties of optimal matching under two settings: (i) homogeneous preferences with heterogeneous sensitivities, and (ii) homogeneous sensitivities with heterogeneous preferences. We begin with the first case, formally described in the following assumption.

\begin{assumption}
[{\sc Homogeneous Preference with Heterogeneous Sensitivities}]\label{assm:Hom_Pref_Het_Sens}
Let the preferred service time be normalized to $t=0$ and demand type $\alpha\in\mathcal{A}\subseteq\mathbb{R}^+$ denote the time sensitivity. The temporal cost function $\ell\colon\mathcal{T}\times \mathcal{A}\to[0,\infty]$ satisfies the following conditions:
\begin{itemize}
    \item$\ell(t,\alpha)\ge0, \forall t,\alpha$ and $\ell(0,\alpha)=0$;
    \item for each $\alpha$, $t\mapsto\ell(t,\alpha)$ is continuous, differentiable, and strictly increasing (decreasing) on $\mathcal{T}^+$ ($\mathcal{T}^-$); 
    \item if $\alpha_1>\alpha_2$ then $|\partial_t\ell(t,\alpha_1)|>|\partial_t\ell(t,\alpha_2)|$ for all $t\in\mathcal{T}$.
\end{itemize}
Where $\mathcal{T}^+ = \mathcal{T}\cap\mathbb{R}^+$ and $\mathcal{T}^- = \mathcal{T}\setminus \mathcal{T}^+$ represent the positive and negative time portions of the horizon, respectively.
\end{assumption}

This assumption posits that all demand types share the same preferred service time, normalized to the reference time $t=0$, but differ in the degree of pain associated with being served earlier or later. The function $\ell(t,\alpha)$ is a temporal cost (or delay/earliness cost): it is zero exactly at the reference time and increases smoothly as service deviates from that time, increasing on the “late” side $\mathcal{T}^+=\{t\geq0\}$ and decreasing on the “early” side $\mathcal{T}^-=\{t<0\}$. The third requirement implies that higher-$\alpha$ types have uniformly steeper marginal temporal cost, i.e., they are more time-sensitive everywhere. Economically, $\alpha$ indexes impatience or urgency; a larger $\alpha$ means the customer’s utility drops more quickly if you miss the preferred time. This structure is typical when everyone wants the same release or appointment window, but some are more tolerant of deviations. 
As highlighted in our motivating examples, during mass vaccination campaigns, all participants may prefer early appointments; however, those at higher risk, such as the elderly, are especially delay-sensitive. Similarly, for inbound e-commerce shipments, while all sellers prefer immediate inbound, those facing low inventory levels and imminent stockouts typically exhibit greater urgency. 
Under this assumption, we present the following proposition.

\begin{proposition}\label{prop:sensitive_priority}
Under \Cref{assm:Hom_Pref_Het_Sens}, without loss of optimality, the decision maker can restrict attention to \textbf{sensitivity-prioritized schedules}. 
That is, if two demand types with sensitivities $\alpha_1>\alpha_2$ are matched to a common station~$i$ on the same side of the ideal time $0$, then the more sensitive demand $\alpha_1$ is scheduled closer to~$0$ than the less sensitive demand $\alpha_2$.
\end{proposition}

From this proposition, we observe that optimal time scheduling prioritizes demand with higher time sensitivity, if admitted, assigning it a closer time to the ideal time $0$ than another admitted demand with low time sensitivity. However, the following proposition shows that the service inclusion actually favors demand with lower time sensitivity. 

\begin{proposition}\label{prop:exclusion}
Under \Cref{assm:Hom_Pref_Het_Sens}, without loss of optimality, the decision maker can restrict attention to \textbf{sensitivity-monotone coverage}. That is, denote the service coverage of type $\alpha$ demand as $S(\alpha)\subseteq\mathcal{X}$, then for any two distinct sensitivity $\alpha_1<\alpha_2$, $S(\alpha_2)\subseteq S(\alpha_1)$.
\end{proposition}

Taken together, Propositions~\ref{prop:sensitive_priority} and~\ref{prop:exclusion} highlight a subtle tension in optimal time scheduling. On the one hand, when high-sensitivity customers are admitted, they will be served closer to the ideal time, ahead of less-sensitive customers, ensuring that urgency translates into temporal priority. On the other hand, whether such high-sensitivity customers are admitted depends on the service inclusion rule: the system favors accommodating lower-sensitivity types, who are easier to schedule without incurring high costs, while possibly excluding excessively time-sensitive ones, whose service leads to negative utility. This combination implies that optimal scheduling is priority-based within the served set, but can be selective across sensitivities, potentially balancing fairness in scheduling with efficiency in inclusion. Nevertheless, in applications such as vaccination scheduling, where the service reward is sufficiently large, the optimal policy serves all demand, so that the exclusion effect is inactive and only the priority effect remains relevant, with the elderly prioritized for vaccination.

A complementary scenario arises when all demand shares the same sensitivity to pull-in early or delay, but different demand types have different ideal service times. This case is formalized as follows. 

\begin{assumption}[{\sc Homogeneous Sensitivity with Heterogeneous Preferences}]\label{assm:Hom_Sens_Het_Pref}
Let $\tau_\alpha : \mathcal{A}\mapsto\mathcal{T}$ denote the time preference of demand type $\alpha$
. The function $\ell\colon\mathcal{T}\times \mathcal{A}\to[0,\infty]$ is defined by $\ell(t,\alpha) = \hat{\ell}(t - \tau_\alpha)$, where $\hat{\ell}:\mathbb{R} \to \mathbb{R}^+$ is a convex function satisfying $\hat{\ell}(t) \ge 0$ for all $t \in \mathbb{R}$, and $\hat{\ell}(t) = 0$ if and only if $t = 0$.
\end{assumption}

In essence, this assumption models temporal cost as a common penalty curve, horizontally shifted to align with individual preferences. All demand types share a common temporal cost function, $\hat{\ell}$, which penalizes deviations from their personal ideal time and captures how quickly dissatisfaction increases as the assigned time deviates from this ideal. However, their ideal time $\tau_\alpha$ may differ. Convexity of $\hat{\ell}$ means marginal pain increases as the deviation gets large. Economically, this model's populations have dispersed preferred times but similar tolerance around those times. For example, dental patients with different appointment preferences throughout the day, and package recipients who prefer distinct delivery windows due to work-related constraints. 

\begin{proposition}\label{prop:order_preserve}
Under \Cref{assm:Hom_Sens_Het_Pref}, without loss of optimality, the decision maker can restrict attention to \textbf{order-preserving schedules}. That is, if two demand types with preferred times $\tau<\tau'$ are matched to a common station~$i$, their realized service times $t,t'$ satisfy $t \leq t'$. 
\end{proposition}

This proposition establishes a strong structural property for the case of heterogeneous time preferences but homogeneous sensitivity. The optimal schedule is preference-order preserving, meaning that the sequence of service times must respect the chronological order of the demands' preferred times. This result establishes a ``non-crossing" principle: a demand with a later preference time cannot be scheduled ahead of one with an earlier preferred time. This result is analogous to monotone measure-preserving maps for one-dimensional OT problems (see, e.g., \citealt{mccann1995existence}). 

\section{Finite Demand Types}\label{Sec:finite_type}
In this section, we assume that demand is classified into $m$ types, each of which may have a distinct temporal cost according to preferences. This assumption is widely used in practice because demand is often segmented into discrete categories (e.g., urgency levels or preferred time windows). We index demand types by $j\in\{1,\dots,m\}=\mathcal{A}$. The spatial distribution of type-$j$ demand is given by the following density function $\mu_j(\bm{x}) = \mu(\bm{x},j)\in\mathcal{L}^1_+(\mathcal X)$. One approach is to work with the formulation in \eqref{eq:STM_D}, which offers a clean representation of both the objective and the constraints. However, even though $\mathcal{A}$ is discrete, \eqref{eq:STM_D} involves $m + n$ continuous potentials: those related to stations are defined over the temporal domain $\mathcal{T}$, while those corresponding to demand types are defined over the spatial domain $\mathcal{X}$. To address this complexity, the reformulation in \eqref{eq:STB} provides a cleaner angle. In the following theorem, we further simplify the problem based on \eqref{eq:STB}.

\begin{theorem}\label{Thm:finite_type}
When the demand types are finite, the problem \eqref{eq:STB} coincides with
\begin{equation}\tag{STB-D}\label{eq:STB_D}
\begin{aligned}
\inf_{\bm\eta\in\mathbb{R}^{n\times m}}\quad &
   \sum_{i=1}^{n}\int_{\mathcal T} 
      \Big(\ \max_{j\in[m]}\big[\eta_{i,j}-\ell_j(t)\big]\ \Big)^+c_i(t) \mathrm{d}t \\
      &+\sum_{j=1}^m \int_{\mathcal X} 
      \Big(\ \max_{i\in[n]}\big[r-\eta_{i,j}-\delta(\bm{x},\bm{y}_i)\big]\ \Big)^+\mu_j(\bm{x}) \mathrm{d}\bm{x},\\
\end{aligned}
\end{equation}
which is a finite-dimensional convex optimization problem. 
\end{theorem}

This reformulation offers two significant advantages. First, it transforms the original infinite-dimensional linear program \eqref{eq:STM} over matching plans into a $nm$-dimensional convex optimization problem. For instance, consider a city-scale discretization: coarse spatial grid ($1000\times1000$), $20$ demand types, $50$ sub-hourly time slots, and $100$ stations. The primal \eqref{eq:STM} would involve about $10^{11}$ variables; the dual \eqref{eq:STM_D} still carries tens of millions of unknowns and a billion constraints. By contrast, the formulation \eqref{eq:STB_D}, although nonlinear, compresses the decision space to roughly a thousand variables. This several orders-of-magnitude reduction makes the formulation computationally practical.

Another crucial advantage relative to \eqref{eq:STB} is that the objective function in \eqref{eq:STB_D} admits a computable subgradient, enabling the use of gradient-based methods. This reformulation also facilitates the design of approximation algorithms. For example, the operators $\max$ and $(\cdot)^+$ can be replaced with smooth approximations (a combination of softmax and softplus):
\(
(\max_j a_j)^{+} \approx \varepsilon \log\left(1 + \sum_j e^{\frac{a_j}{\varepsilon}}\right)
\),
where $\varepsilon > 0$ is a small smoothing parameter. This substitution yields a differentiable surrogate objective and eliminates the need for geometric partitioning, thereby facilitating large-scale computation.

Reformulation \eqref{eq:STB_D} closely resembles the classical semi-discrete optimal transport problem (see, e.g., \citealp[]{bourne2018semi}), with $\eta_{i,j}$ acting as partitioning weights. However, unlike existing formulations, these weights now interact with both spaces $\mathcal{T}$ and $\mathcal{X}$. Suppose temporal costs vanish, and hence, demand becomes homogeneous in sensitivity while remaining spatially differentiated. If total demand equals total stations' capacity and $r$ is sufficiently large to match all mass, then \eqref{eq:STB_D} reduces to the classical semi-discrete OT with capacity constraints. (c.f., formulations (5.17) in \citealp[]{galichon2016optimal} and (2.10) in \citealp[]{bourne2018semi}). Thus, our reformulation can be viewed as an extension of semi-discrete OT to the double-marginal and partial-matching settings.

\subsection{Implementation of Optimal Matching}
We now discuss how to implement the optimal matching in a simple and efficient manner. In practice, a key challenge in implementing optimal matching is ensuring envy-freeness, a condition critical for effective and stable allocations \citep{foley1966resource}. In this section, we assume that monetary transfers are available to the social planner. For instance, the service provider might adopt a time-dependent pricing policy $p_i(t)$ for each station $i$. In public service contexts, this policy may also take the form of subsidies\footnote{For example, in public services, the planner can reduce all payments by the same amount so that none of them is positive. Therefore, the pricing policy can be equivalently interpreted as a subsidy scheme. In practice, such a case arises when a specific demand (e.g., a remote or urgent case) cannot be feasibly met by public services. The subsidy enables individuals to obtain alternative private services.}. Transfers are purely redistributive and therefore do not change overall social welfare, so our focus is on implementation: how the provider can realize the optimal matching in an envy-free manner. From now on, we also use ``agent" to denote a unit of demand to reflect individual preferences.

\subsubsection{Spatiotemporal partition.}\label{sec:SP_partition}
Before proceeding, we first analyze the structure of the desired matching. Denote the optimal $\bm\eta$ in \eqref{eq:STB_D} as $\bm\eta^*$. Let $\mathcal{C}_{i,j}(\bm\eta^*) = \mathcal{C}^r_{i}(\bm\eta^*_{\cdot,j}) $ be the generalized Laguerre cells. $\mathcal{T}_{i,j}$ are defined analogously, that is
\begin{equation*}
\begin{aligned}
\mathcal{T}_{i,j}(\bm\eta^*) =&  \left\{t\in\mathcal{T}:\ell_j(t) -\eta^{*}_{i,j} \leq 0,\space \ell_j(t) -\eta^{*}_{i,j}\leq \ell_{\tilde{j}}(t) -\eta^{*}_{i,\tilde{j}},\quad\forall\tilde{j}\in[m]\right\},\\
\mathcal{T}_{i,0}(\bm\eta^*) =& \left\{t\in\mathcal{T}:\ell_j(t) \geq \eta^{*}_{i,j}, \quad\forall j\in[m]\right\}.
\end{aligned}
\end{equation*}
In other words, $\mathcal{T}_{i,j}(\bm\eta^*)$ can be viewed as the generalized Laguerre cells on $\mathcal{T}$ induced by the temporal cost function $\ell$, the weights $\bm\eta^*$, and threshold $0$. For simplicity, we further assume that for each $i$, the sets $\{t: \eta^*_{i,j} - \ell_j(t) = 0\}$ and $\{t: \eta^*_{i,j} - \ell_j(t) = \eta^*_{i,j'} - \ell_{j'}(t)\}$ for all $j, j' \in [m]$ are Lebesgue null sets. This assumption mirrors our earlier discussion of the boundaries of $\mathcal{C}^{\omega_0}_i(\bm{\omega})$ in \Cref{lemma:null_space}. Under \Cref{assm:Hom_Pref_Het_Sens}, the condition holds since the boundaries are finite points. Under \Cref{assm:Hom_Sens_Het_Pref}, this condition is satisfied whenever $\ell$ is strictly convex. We formally state this null boundary assumption, encompassing both $\mathcal{T}_{i,j}(\bm\eta^*)$ and $\mathcal{C}_{i,j}(\bm\eta^*)$, as follows. 

\begin{assumption}[{\sc Null Boundary}]\label{assm:null_boundary}
For each station $i\in[n]$, and any demand type $j\in[m]$, the intersections 
$\mathcal{C}_{i,j}(\bm\eta^*)\cap\mathcal{C}_{i,j'}(\bm\eta^*)$ and $\mathcal{T}_{i,j}(\bm\eta^*)\cap\mathcal{T}_{i,j'}(\bm\eta^*)$ have Lebesgue measure zero for all $j'\in\{0\}\cup[m]$ with $j'\neq j$.
\end{assumption}

Under \Cref{assm:null_boundary}, the first-order condition of \eqref{eq:STB_D} implies that
\begin{equation}\label{eq:first_order_condition}
\int_{\mathcal{C}_{i,j}(\bm\eta^*)}\mu_{j}(\bm{x})\mathrm{d}\bm{x} = \int_{\mathcal{T}_{i,j}(\bm\eta^*)}c_{i}(t)\mathrm{d}t.
\end{equation}
Then, the optimal matching assigns all type-$j$ demand within $\mathcal{C}_{i,j}$ to station $i$, scheduling it over the corresponding time set $\mathcal{T}_{i,j}$. The first-order condition \eqref{eq:first_order_condition} guarantees that the assigned demand volume matches the station’s capacity on $\mathcal{T}_{i,j}$. This matching yields the following social welfare:
\begin{equation*}
\sum_{i=1}^{n}\sum_{j=1}^{m}\left[\int_{\mathcal{C}_{i,j}(\bm\eta^*)}\left(r - \delta(\bm{x}, \bm{y}_i)\right)\mu_{j}(\bm{x})\mathrm{d}\bm{x} - \int_{\mathcal{T}_{i,j}(\bm\eta^*)}\ell_j(t)c_{i}(t)\mathrm{d}t \right],
\end{equation*}
which corresponds to the objective in \eqref{eq:STB} when $\bm\eta = \bm\eta^*$ due to our null boundary assumption, thereby establishing optimality. Therefore, for the remainder of this section, we may denote the optimal matching plan using the spatial and temporal partitions $\{\mathcal{C}_{i,j}(\bm\eta^*), \mathcal{T}_{i,j}(\bm\eta^*)\}_{i \in [n], j \in [m]}$ in place of $\pi$. \Cref{fig:STMatching} in \Cref{sec:Fig_STMatching} illustrates an example of spatial and temporal partitions with two stations and three demand types. A practical application is presented in \Cref{sec:numerical}.




\subsubsection{Envy-free pricing.}
We assume that the planner announces prices that vary by station and by time. Given posted pricing policy $\bm p (t)=\{p_i(t)\}_{i}$, a type-$j$ agent at $\bm x$ who matched to station $i$ at time $t$ receives net utility
\[
U_j(\bm{x},i,t;\bm{p}) \;=\; r - \delta(\bm{x},\bm{y}_i) - \ell_j(t) - p_i(t).
\]
Specifically, we aim to determine the prices and the corresponding optimal matching that together ensure envy-freeness and individual rationality, as defined below.
\begin{definition}
We say the price policy $\bm p$ and the matching plan $\{\mathcal{C}_{i,j}, \mathcal{T}_{i,j}\}_{i \in [n], j \in [m]}$ satisfies envy-freeness and individual rationality if for any $(i,j)$:
\begin{enumerate}
    \item\emph{Envy-Freeness:} $U_j(\bm{x},i,t;\bm{p}) \geq U_j(\bm{x},i',t';\bm{p})$ for all  $\bm{x}\in\mathcal{C}_{i,j}$, $t\in\mathcal{T}_{i,j}$, $i'\in[n]$ and $t'\in\mathcal{T}$; $ 0 \geq U_j(\bm{x},i',t';\bm{p})$ for all $\bm{x}\in\mathcal{X}\setminus \cup_{i=1}^{n}\mathcal{C}_{i,j}$, $i'\in[n]$ and $t'\in\mathcal{T}$.
    \item\emph{Individual Rationality:} $U_j(\bm{x},i,t;\bm{p})\geq 0$ for all $\bm{x}\in\mathcal{C}_{i,j}$ and $t\in\mathcal{T}_{i,j}$.
\end{enumerate}
\end{definition}

Envy-freeness ensures that any type-\(j\) agent assigned to \((i,t)\) or left unmatched weakly prefers the assigned outcome over any alternative \((i',t')\) under prices \(\bm p\). Individual rationality requires her net utility at \((i,t)\) after paying $p_i(t)$ to be nonnegative. To ensure the optimal matching satisfies both envy-freeness and individual rationality, we introduce the following pricing policy:
\begin{equation}\label{eq:incentive_price}
p^{*}_i(t) = \left(\max_{j\in[m]}\left\{\eta^{*}_{i,j} -\ell_j(t)\right\}\right)^+.
\end{equation}

By the definition of $p^*$ and $\mathcal{T}_{i,j}$, we have the following lemma.
\begin{lemma}\label{lemma:pricing_inequality}
For all $i\in[n],\, j\in[m]$ and $t\in\mathcal T$,
$
\ell_j(t)+p_i^*(t)\ \geq \eta_{i,j}^*
$.
Under \Cref{assm:null_boundary}, the equality holds if and only if $t\in\mathcal T_{i,j}(\bm\eta^*)$ (up to a null set).
\end{lemma}

Intuitively, $p_i^*$ irons out temporal cost so that a type-$j$ agent’s best utility at station $i$ becomes $r-\delta(\bm x,\bm y_i)-\eta_{i,j}^*$ and is simultaneously achieved on $\mathcal T_{i,j}(\bm\eta^*)$. This property, as formalized in \Cref{lemma:pricing_inequality}, then leads to the following proposition of envy-freeness and individual rationality.

\begin{proposition}\label{prop:EF_implementation}
Consider finite demand types and let $\bm{\eta}^*$ be an optimal solution of \Cref{eq:STB_D}. Under \Cref{assm:null_boundary}, the posted prices $\bm p^*$ in \eqref{eq:incentive_price} along with the optimal matching $\{\mathcal{C}_{i,j}(\bm\eta^*), \mathcal{T}_{i,j}(\bm\eta^*)\}_{i \in [n], j \in [m]}$ are envy-free and individually rational.
\end{proposition}

\subsubsection{Finite time-slots.}
In practice, a continuous pricing policy may be impractical. Therefore, we demonstrate that the service provider can still achieve optimal matching by offering a discrete set of time slots (interval) at varying prices. Agents can select from these time slots at a fixed price, and the service time will be randomly assigned according to a specified distribution. Without loss of generality, we assume that when multiple time slots yield the same maximum utility for agents of a given type, they distribute themselves proportionally among these indifferent options, filling available capacities. To characterize the size of time slots, we require the following definition.

Our result is closely connected to the length of the longest Davenport–Schinzel sequence, denoted by $\lambda_s(n)$. This function has been fully characterized for small values of $s$: specifically, $\lambda_0(n)=1$, $\lambda_1(n)=n$, and $\lambda_2(n)=2n-1$. A more detailed discussion of Davenport–Schinzel sequences, the definition of $\lambda_s(n)$, and its upper bounds for larger $s$ is provided in \Cref{sec:D_S_sequence}. Drawing on the function $\lambda_s(n)$, we next formulate the following proposition, which extends this framework to our setting.

\begin{proposition}\label{prop:finite_slots} Consider finite demand types and let $\bm{\eta}^*$ be an optimal solution of \Cref{eq:STB_D}. 
Under \Cref{assm:null_boundary}, suppose each $\ell_j$ is quasi-convex, and for any $j\neq j'$ the functions $\eta^*_{i,j}-\ell_j(t)$ and $\eta^*_{i,j'}-\ell_{j'}(t)$ intersect at most $s$ times on $\mathcal T$. Then there exist disjoint open intervals $\{I_{i,k}\}_{k=1}^{K_i}\subset \mathcal T$ with $K_i\le \lambda_s(m)$ such that, on each $I_{i,k}$, the maximizer $j(i,k)=\arg\max_{j\in[m]} \{\eta^*_{i,j}-\ell_j(t)\}$ is unique and constant. The planner can post the slot price
\[
p^*_{i,k} = \frac{\int_{I_{i,k}} \left(\eta^{*}_{i,j(i,k)} -\ell_{j(i,k)}(t)\right)c_i(t)\mathrm{d}t}{\int_{I_{i,k}}c_i(t)\mathrm{d}t}.
\]
and allocates service time within $I_{i,k}$ according to the probability density $c_i(t)\big/\int_{I_{i,k}} c_i(u) \mathrm{d}u$. This pricing scheme implements the optimal matching.
\end{proposition}

Intuitively, the price of each slot is set to iron out the slot’s expected temporal cost for its target demand type such that the expected net utility for that type remains unchanged and optimal. As a result, the discrete slot mechanism preserves the same allocation and incentive properties as the continuous-pricing scheme \eqref{eq:incentive_price}. Under \Cref{assm:Hom_Pref_Het_Sens}, where $s = 2$, each station needs to offer at most $\lambda_{2}(m) = 2m - 1$ time slots. In contrast, under \Cref{assm:Hom_Sens_Het_Pref}, if $\ell$ is strictly convex, then $s = 1$, and each station needs to offer only $\lambda_{1}(m) = m$ time slots, each tailored to a specific demand type.

\subsection{Capacity Allocation}\label{sec:capacity_allocation}
Having reduced the service matching problem to the finite convex program \eqref{eq:STB_D}, we now turn to the question of how best to allocate total capacity across the $n$ stations. Let $\bar c_i(t)$ denote the normalized capacity profile of station $i$, which captures the temporal shape of its service capability, and let $a_i\ge0$ be the capacity scale decision so that $
c_i(t)=a_i\bar c_i(t), t\in\mathcal T,\ i\in[n]
$. We assume that allocating capacity incurs linear cost with weights $\xi_i>0$ under a total budget $B>0$:
$
\sum_{i=1}^n \xi_i a_i \le B
$. To illustrate, consider a service system with multiple stations. The scale decision $a_i$ represents the level of resource input at station $i$, such as labor or equipment. The baseline profile $\bar c_i(t)$ describes the service capacity generated by one unit of resource input at station $i$; differences in $\bar c_i(t)$ may reflect different operating schedules (e.g., daytime operation or 24-hour service). Finally, $\xi_i$ denotes the local cost of one unit of resource input at station $i$, which may vary across stations.

Given \(\bm a=(a_1,\dots,a_n)\), the induced welfare (equivalently, the optimal objective of \eqref{eq:STB_D}) is
$$
\min_{\bm\eta\in\mathbb{R}^{n\times m}}
\sum_{i=1}^{n} a_i\underbrace{\int_{\mathcal T}
\Big(\max_{j\in[m]}[\eta_{i,j}-\ell_j(t)]\Big)^+\bar c_i(t)\mathrm{d}t}_{\displaystyle Z_i(\bm\eta)}
+
\underbrace{\sum_{j=1}^m \int_{\mathcal X}
\Big(\max_{i\in[n]}[r-\eta_{i,j}-\delta(\bm x,\bm y_i)]\Big)^+\mu_j(\bm x)\mathrm{d}\bm x}_{ \displaystyle\Psi(\bm\eta)}.
$$
Here $Z_i(\bm\eta)$ and $\Psi(\bm\eta)$ are convex functions of $\bm{\eta}$. The capacity allocation problem chooses \(\bm a\) to maximize welfare after optimal matching:
\begin{equation}\label{eq:EA}
\max_{\bm a\ge0}\min_{\bm\eta\in\mathbb{R}^{n\times m}}
\left\{
\sum_{i=1}^n a_i Z_i(\bm\eta) + \Psi(\bm\eta)
:\ \sum_i \xi_i a_i \le B
\right\}.
\end{equation}
Note that this formulation not only determines the optimal capacity allocation $\bm{a}$ but also jointly identifies the optimal matching through $\bm{\eta}$. Thanks to \eqref{eq:STB_D}, problem \eqref{eq:EA} can be reformulated as a convex minimization problem with respect to $\bm{\eta}$ by a standard minimax argument (Sion's minimax theorem).
\begin{equation}\tag{CA}\label{eq:EA_reduced}
\min_{\bm\eta\in\mathbb{R}^{n\times m}}\ \ \Psi(\bm\eta)+B\max_{i\in[n]}\frac{Z_i(\bm\eta)}{\xi_i}.
\end{equation}

The budget \(B\) controls the largest normalized \(Z_i(\bm\eta)\) across stations. Since \(Z_i(\bm\eta)\)'s are convex functions for all $i$, the term $\max_{i\in[n]}\frac{Z_i(\bm\eta)}{\xi_i}$ is also convex and hence the problem \eqref{eq:EA_reduced} is still a convex optimization of $\bm\eta$.

In addition to identifying optimal capacity allocation through continuous decision variables, our reformulation also provides a practical framework for comparing system configurations. This capability is particularly valuable in practical settings, where fine-tuning may be infeasible and a pre-existing configuration is already in place. For instance, one might ask whether eliminating a station and reallocating its capacity to another could enhance overall performance, a scenario that exemplifies the classic trade-off between dispersion and concentration plans. We illustrate this with two simple examples in \Cref{sec:D_or_C}.

\subsection{Homogeneous Preference with Linear Cost}
In this section, we introduce a simplified yet practically relevant structure for temporal costs and service capacity. Specifically, we consider the finite-type demand framework with the assumption of homogeneous time preference (see \Cref{assm:Hom_Pref_Het_Sens}). In this setting, we further assume that the service capacity is constant and the temporal cost incurred by agents varies linearly with their deviation from a common reference time. This framework is not only analytically convenient but also grounded in practical considerations. The assumption of constant capacity is often appropriate in systems with rigid processing capabilities or consistent staffing levels. Linear temporal costs, meanwhile, capture asymmetries in time sensitivity, such as penalties for early scheduling or service delays. Formally, we make the following assumption.

\begin{assumption}[{\sc Slack Capacity and Two-Piece Linear Temporal Costs}]\label{assm:long-horizon-linear-cost} ~\\
\leavevmode
    {1. (Constant capacity)} Each station \(i \in [n]\) operates at a constant service rate over time: \(c_i(t) \equiv c_i>0\).\\
    {2. (Sufficiently long horizon)} The scheduling horizon is \(\mathcal{T}=[-T,T]\) and is long enough that any single station could, in principle, process all demand:
    \(
      T \min_{i\in[n]} c_i \;\ge\; \sum_{j=1}^{m} \int_{\mathcal{X}} \mu_j(\bm{x})\mathrm{d}\bm{x}
    \).\\
    {3. (Two-piece linear costs)} A type-\(j\) agent incurs temporal cost
    \(
      \ell(t,j) \coloneqq s_j\big(b(-t)^{+} + h t^{+}\big)
    \),
    where \(b,h\ge 0\) are the cost slopes for early/late and \(s_1>s_2>\cdots>s_m>0\) are the sensitivities of demand type $1,2,\dots,m$.
\end{assumption}

Under these structural assumptions, the temporal component of the allocation problem becomes analytically tractable. As a result, the problem \eqref{eq:STB_D} admits a simpler reduction:

\begin{proposition}\label{prop:finite_linear}
In the finite demand type model, suppose \Cref{assm:long-horizon-linear-cost} holds. Then the problem \eqref{eq:STB_D} reduces to
\begin{equation}\label{eq:finite_linear_Dual}
\min_{\bm\eta\in\mathbb{R}^{n\times m}}\quad
\sum_{i=1}^n\frac{c_i}{\beta}\bm\eta_i^\top \bm{A}^{-1}\bm\eta_i+
\sum_{j=1}^m
\int_{\mathcal{X}}\left(\max_{i\in[n]}\left\{r-\delta(\bm{x},\bm{y}_i)-\eta_{i,j}\right\}\right)^+\mu_j(\bm{x})\mathrm{d}\bm{x}.
\end{equation}
where 
\(
\beta = \frac{2bh}{b+h}\) and
\(A_{ij} = s_{\max\{i,j\}}\), for \(i,j=1,\dots,m
\).

\end{proposition}
The parameter $\beta$, which is the harmonic mean of $b$ and $h$, scales the temporal cost in this system. We note that when $b$ or $h$ is $+\infty$, $\beta$ is defined as $2h$ or $2b$, respectively, which captures the cases in which only one side (early or delay) matching is available. 

Problem \eqref{eq:finite_linear_Dual} forms a finite size ($n\times m$) convex problem which is computationally tractable. In addition, since $\bm{A}$ is a positive matrix and $\bm{A}^{-1}$ is a Z-matrix with nonpositive off-diagonals (see \Cref{Sec:inverse_A}), according to \cite{plemmons1977m}, $\bm{A}^{-1}$ is an M-matrix which is positive definite, so the objective is strictly convex. 

Under \Cref{assm:long-horizon-linear-cost}, the envy-free pricing $p^*$ in \eqref{eq:incentive_price} is piecewise linear, reaching its maximum at time $t=0$ and decreasing as $t$ moves away from zero in either direction. Moreover, as established in \Cref{prop:finite_slots}, it is sufficient to offer $2m-1$ time slots with different prices. Both the continuous and slot-based prices for a certain station are illustrated in \Cref{fig:Pricing policy}.

\begin{figure}
\caption{Illustration of Envy-Free Pricing Policy.}\label{fig:Pricing policy}
\centering
\includegraphics[width=0.5\linewidth]{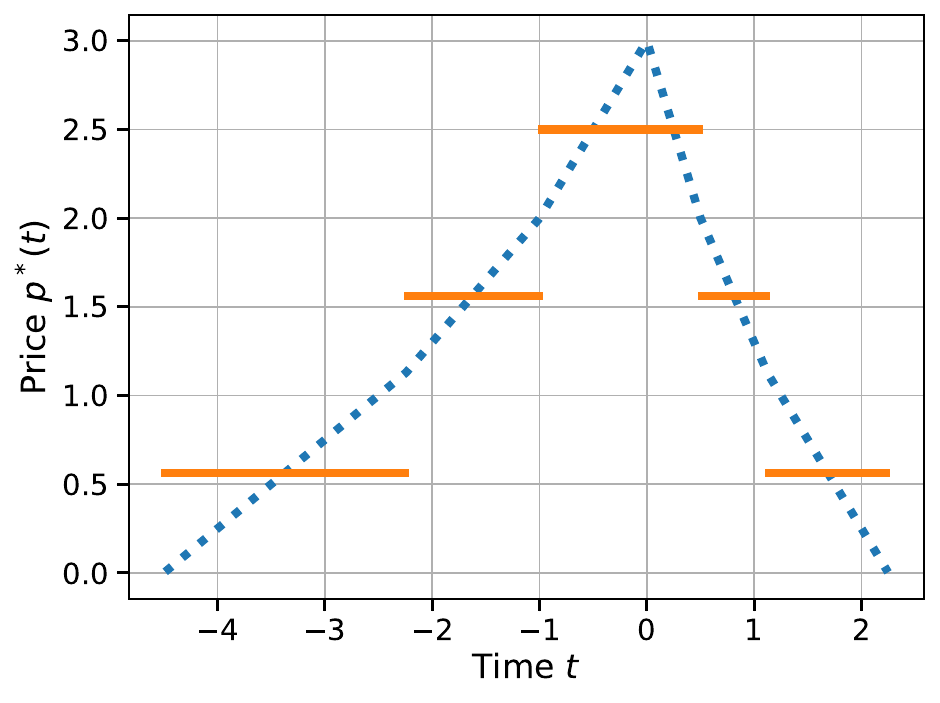}
\par\vspace{0.5ex}
\begin{minipage}{0.95\linewidth}
\footnotesize\emph{Note.} In this example, we set $b=2$ and $h=1$. There are three demand types with sensitivities $1$, $0.7$, and $0.5$. The station has a capacity of $8$ and is assigned volumes of $12$, $15$, and $19$ for the three types, respectively. The blue dotted line represents the continuous envy-free price $p^*$. The orange segments represent the equivalent implementation using a finite-slot mechanism, which requires five slots to coordinate three demand types.
\end{minipage}
\end{figure}

\subsubsection{Capacity allocation with a single demand type.}
When time sensitivity is also homogeneous (i.e., there is only one type of demand), we further find that the capacity allocation \eqref{eq:EA} is closely related to the traditional uncapacitated OT problem (i.e., without temporal costs and processing capacities). Next, we focus on the “no-exclusion” regime by taking $r$ sufficiently large to ensure that all demand is satisfied. This allows us to isolate the effects of capacity variation on social welfare. In particular, we aim to capture the impact of capacity on spatiotemporal costs while abstracting away reward considerations. This setting is also practically relevant. For example, many public service systems operate under a regime of universal access, where no demand shall be excluded. Moreover, in various commercial contexts, services are sold in advance at a given price, and providers commit to fulfilling orders before customers reveal detailed information. In such cases, the planner seeks to minimize the realized transportation and scheduling frictions.

\begin{proposition}\label{Prop:Opt_c}
Under \Cref{assm:long-horizon-linear-cost}, assume that there is only one demand type and the reward $r$ is sufficiently large such that it is optimal to serve all demand. Further assume identical unit capacity costs across stations, so that the budget constraint reduces to $\sum_{i=1}^{n}c_i\leq B$. Then the optimal $c_i$ (denoted as $c^*_i$) is given by
\begin{equation*}
c^*_i = \frac{\int_{\mathcal{C}^{+\infty}_i(\bm{0})}\mu(\bm{x})\mathrm{d}\bm{x}}{\int_{\mathcal{X}}\mu(\bm{x})\mathrm{d}\bm{x}}B.
\end{equation*}
\end{proposition}

\Cref{Prop:Opt_c} provides clear guidelines for allocating capacities across service stations within a fixed budget. Recall that in \Cref{def:GLC}, $\bm\omega=(\omega_i)_{i\in[n]}$ denotes the vector of cell weights associated with the stations. By setting these weights to zero and letting $r\rightarrow+\infty$, the partition $\mathcal{C}^{+\infty}_i(\bm{0}), i\in[n]$ in \Cref{def:GLC} is exactly the unweighted Laguerre cells without residual (also known as Voronoi diagrams). Since a lower spatial cost corresponds to a shorter distance in our setting, the optimal allocation strategy in \Cref{Prop:Opt_c} is straightforward: first, partition the service area of each station solely based on distance, and then allocate capacities $c^*_i$ to each station in proportion to the demand within its coverage area.

\section{Spatiotemporal Hotelling Model}\label{Sec:Hotelling}
While the previous sections focused on the methodology and solvability of the general model, this section turns our attention to the qualitative behavior of the optimal matching. By analyzing more stylized models, we aim to uncover how optimal service regions and schedules may evolve with changes in reward and capacity, highlighting phenomena unique to the spatiotemporal context.

Specifically, we consider an extension of the classical Hotelling model, also known as the “Linear City.” We assume that both demand and stations are located on the interval $[0,1]$. For simplicity, we adopt the cost and capacity assumptions specified in \Cref{assm:long-horizon-linear-cost}. Furthermore, we focus on a one-sided temporal cost structure by setting $l(t,\alpha)=\alpha w t$, where $w$ represents the waiting cost and the demand type $\alpha$ directly captures time sensitivity.

\subsection{Hotelling Model with Homogeneous Demand}
We first consider the setting in which all demand has homogeneous time sensitivity, $\alpha=1$.

\begin{example}[Hotelling Model with Homogeneous Sensitivity]\label{Example:Hotelling_Homogeneous}
The demand follows a uniform distribution on the interval $[0,1]$ with density $1$. Two service stations, indexed by $i=1$ and $2$, are located at $0$ and $1$, respectively, and operate with constant capacity levels $c_1$ and $c_2$. 
\end{example}

In this example, the optimal partition can be determined exactly. Since all demands share the same level of time sensitivity, the scheduling of service times within a given station is arbitrary, as the total waiting cost at that station remains unchanged regardless of the schedule. We analyze the optimal partition in two scenarios: partial service, in which some demand is excluded, and complete service, in which all demand is met.

\begin{proposition}\label{prop:Hotelling_Homogeneous}
In \Cref{Example:Hotelling_Homogeneous}, the optimal partition of the demand interval can be expressed as a disjoint union
$
[0,1]=\mathcal{E}_0 \cup \mathcal{E}_1 \cup \mathcal{E}_2,
$ where demand in $\mathcal{E}_1$ (resp., $\mathcal{E}_2$) is served by station $1$ (resp., $2$), and demand in $\mathcal{E}_0$ is not served. The sets depend on $r$ as follows.

\textbf{Partial service:} If $r<\frac{(w+c_1)(w+c_2)}{2 c_1 c_2 + (c_1 + c_2) w}$, the service regions are
$
\mathcal{E}_1 = \left[0,\frac{rc_1}{c_1+w}\right]$,\quad
$\mathcal{E}_2 = \left[1-\frac{rc_2}{c_2+w},1\right]$, and $\mathcal{E}_0 = \left(\frac{rc_1}{c_1+w},1-\frac{rc_2}{c_2+w}\right)$.

\textbf{Complete service:} Otherwise, the service regions are
$\mathcal{E}_1 
=\left[0,\frac{1}{2}+\frac{(c_1-c_2)w}{2 c_1 c_2 + (c_1 + c_2) w}\right]$, $
\mathcal{E}_2 = [0,1]\setminus \mathcal{E}_1$ and $
\mathcal{E}_0 = \emptyset$.

\end{proposition}

The optimal total social welfare under the partial service and complete service regimes is given by
$
\frac{c_1r^2}{2(c_1+w)} + \frac{c_2r^2}{2(c_2+w)}$ and $r - \frac{(c_2+w)(c_1+w)}{2(2 c_1 c_2 + (c_1 + c_2) w)}
$
respectively.

Under optimal planning, the more efficient station, with higher capacity, is allocated a larger share of the service region and always completes its assigned demand earlier, with the ratio of its service durations given by $\frac{c_2+w}{c_1+w}$. With the capacity of station 2 fixed, the service coverage of station 1, $\min\{\frac{rc_1}{c_1+w},\frac{c_1c_2+c_1w}{2 c_1 c_2 + (c_1 + c_2) w}\}$, concavely increases as $c_1$ rises, implying diminishing marginal improvements in service coverage from capacity expansion. Furthermore, by taking the derivative of the optimal social welfare with respect to the capacity, we find that the marginal gain from capacity expansion is larger at the less efficient station (i.e., with lower processing capacity) than at the more efficient one. This suggests that balancing capacity across stations yields a higher return than further expanding capacity at the station that already has more capacity.

\subsection{Hotelling Model with Uniform Sensitivity}
While the homogeneous‐demand setting offers analytical tractability and clear insight into the spatial allocation of service regions, it assumes that all demand is equally sensitive to waiting time. In practice, however, customers may vary considerably in their willingness to wait: some may tolerate longer delays for lower travel costs, while others may strongly prefer faster service even at a greater distance. To capture this heterogeneity, we extend the model to a setting in which time sensitivity is uniformly distributed across the population. This richer framework allows us to investigate how differences in patience influence the optimal spatial partition and the utilization of station capacities.

\begin{example}[Hotelling Model with Uniform Time Sensitivity]\label{Example:Hotelling_Uniform}
Assume that the demand is uniformly distributed on $\mathcal{X}\times\mathcal{A}=[0,1]\times[0,1]$ with density $1$. Two service stations, indexed by $i=1,2$, are situated at positions $0$ and $1$, and operate with constant capacity levels $c_1$ and $c_2$, respectively.
\end{example}

We present the optimal matching across three scenarios: partial service (disjoint), partial service (adjacent), and complete service. In the first two cases, a portion of the demand is excluded from service, while in the third case, all demand is eventually assigned to a station. In the partial service (disjoint) case, the two stations independently determine their service areas. In contrast, the partial service (adjacent) case introduces a boundary between service regions, particularly for low-sensitivity demand.

\begin{proposition}\label{prop:Hotelling_Uniform}
In \Cref{Example:Hotelling_Uniform}, the optimal scheduling can be expressed as a disjoint union $[0,1]^2 = \mathcal{E}_1\cup \mathcal{E}_2\cup \mathcal{E}_0$ where demand in $\mathcal{E}_1$ (resp., $\mathcal{E}_2$) is served by station $1$ (resp., $2$), with service scheduled in order of time sensitivity, while demand in $\mathcal{E}_0$ is not served. The sets are determined by two boundary functions $f_1$ and $f_2$:
\begin{equation*}
\begin{aligned}
\mathcal{E}_1 = \left\{(x,\alpha): 0 \leq x \leq  f_1(\alpha) 
\right\},\quad
\mathcal{E}_2 = \left\{(x,\alpha): 1- f_2(\alpha)  \leq x \leq 1
\right\},\quad
\mathcal{E}_0 = [0,1]^2\setminus (\mathcal{E}_1\cup \mathcal{E}_2).
\end{aligned}
\end{equation*}
The functions $f_i$ depend on $r$ as follows.\\
\textbf{Case 1: Partial service (disjoint).}
If $r\leq 0.5$, the boundaries are given by
\begin{equation*}
f_i(\alpha) = r f^*(\alpha,c_i,w), \quad f^*(\alpha,c,w) =\frac{\cosh\left((1-\alpha)\sqrt{\frac{w}{c}}\right)}{\cosh\left(\sqrt{\frac{w}{c}}\right)}. 
\end{equation*}

\textbf{Case 2: Partial service (adjacent).} 
Let $ \lambda = \sqrt{\frac{w(c_1+c_2)}{2c_1c_2}}$. If $$0.5<r<
\frac{c_1^2+c_2^2}{(c_1+c_2)^2} +\frac{w}{2(c_1+c_2)} -\frac{(c_1-c_2)^2}{2(c_1+c_2)^2\cosh(\lambda)}\coloneqq r_{\mathrm{c}}(c_1,c_2,w),$$ 
there exists a threshold $\hat{\alpha}$ such that: the two stations’ service regions touch ($f_1(\alpha)+f_2(\alpha)=1$) when $\alpha\leq\hat{\alpha}$, but remain separated ($f_1(\alpha)+f_2(\alpha)<1$) when $\alpha>\hat{\alpha}$. The details of $\hat{\alpha}$ and $f_i$ are provided in \Cref{sec:proof_Hotelling_Uniform}.

\textbf{Case 3: Complete service.} 
If $r\geq r_{\mathrm{c}}(c_1,c_2,w)$, the expressions of $f_i$, $i=1,2,$ are
\begin{equation*}
f_1(\alpha) = \frac{c_2-c_1}{2(c_1 + c_2)}\frac{\cosh(\lambda(1-\alpha))}{\cosh(\lambda)} + \frac{c_1}{c_1 + c_2}
,\quad f_2(\alpha)=1-f_1(\alpha),\quad \lambda = \sqrt{\frac{w(c_1+c_2)}{2c_1c_2}}.
\end{equation*}
\end{proposition}

\Cref{fig:hotelling_all} displays the optimal demand partitions corresponding to the three cases analyzed in \Cref{Example:Hotelling_Uniform}, illustrating how the service regions change with the service reward~$r$. In the low-reward regime, as shown in (a), the stations’ service regions are disjoint: each serves only the demand that provides the highest contribution to social welfare, leaving a central set of demand types unassigned. As the reward increases to an intermediate level, as shown in (b), the potential service regions expand and begin to touch, sharing a boundary of measure zero. This necessitates a partition boundary, which forms for demand with lower time sensitivity. In this “partial service (adjacent)” regime, some highly sensitive but spatially intermediate demand may still be excluded. Finally, in the high-reward regime, as shown in (c), a complete service policy is optimal: all demand is served, and the entire location-sensitivity space is partitioned by a continuous boundary separating the two stations’ regions.

\begin{figure}[htbp!]
\centering
\caption{Illustration of the Optimal Demand Partitioning.}
\label{fig:hotelling_all}
\subfigure[Partial Service (Disjoint)]{
\begin{minipage}[b]{0.35\textwidth}\centering
    \includegraphics[width=1\textwidth]{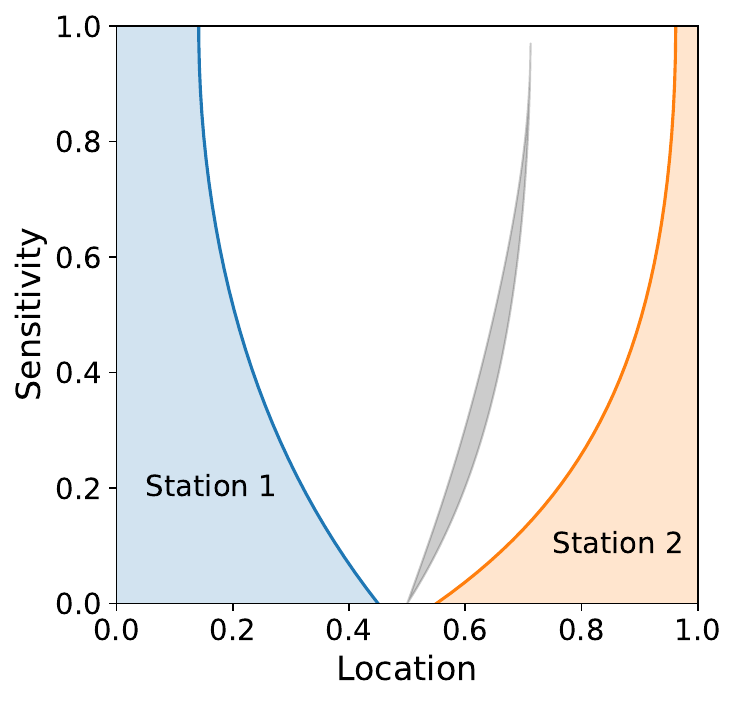}
    \end{minipage}
    }
    \subfigure[Partial Service (Adjacent)]{
\begin{minipage}[b]{0.35\textwidth}\centering
    \includegraphics[width=1\textwidth]{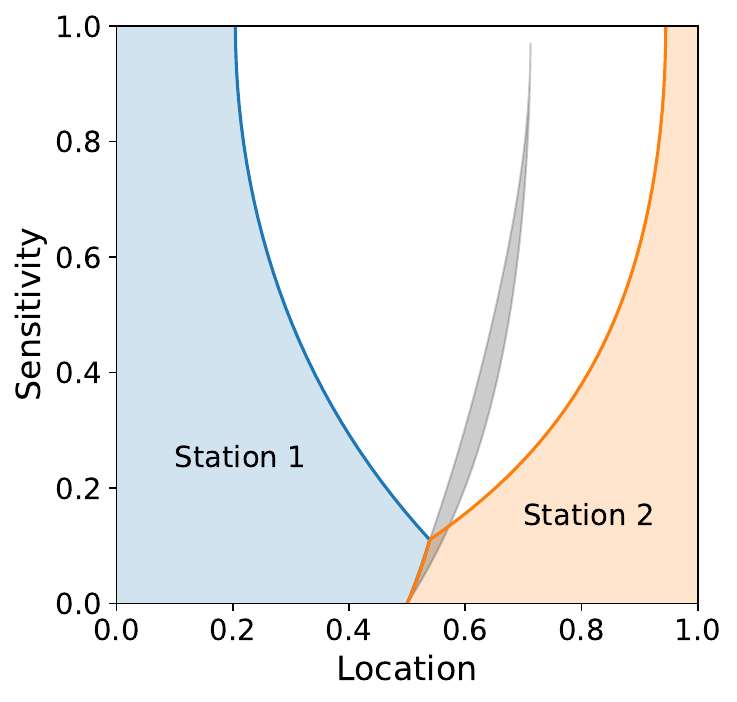}
    \end{minipage}
    }\\
\subfigure[Complete Service]{
\begin{minipage}[b]{0.35\textwidth}\centering
    \includegraphics[width=1\textwidth]{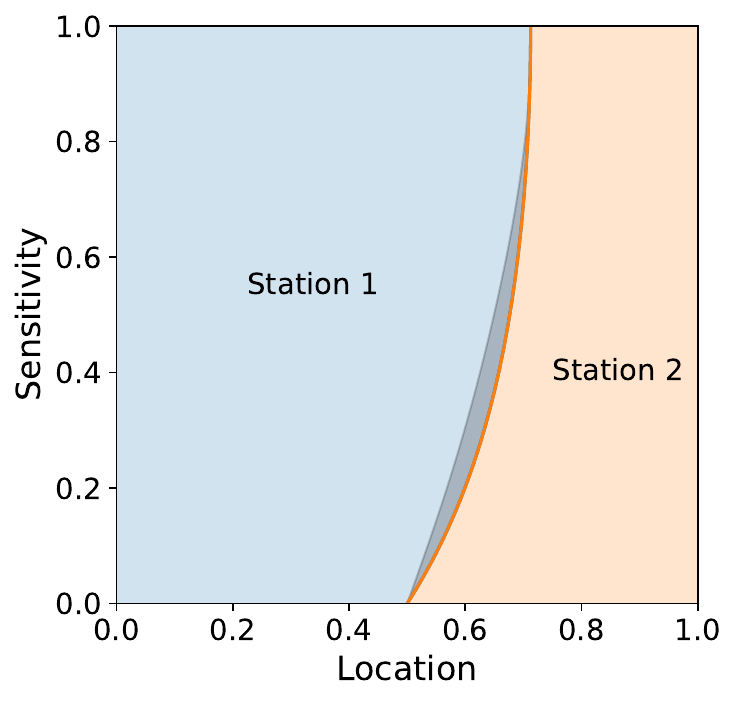}
    \end{minipage}
    }
\subfigure[Complete Service with Varying $w$]{
\begin{minipage}[b]{0.35\textwidth}\centering
    \includegraphics[width=1\textwidth]{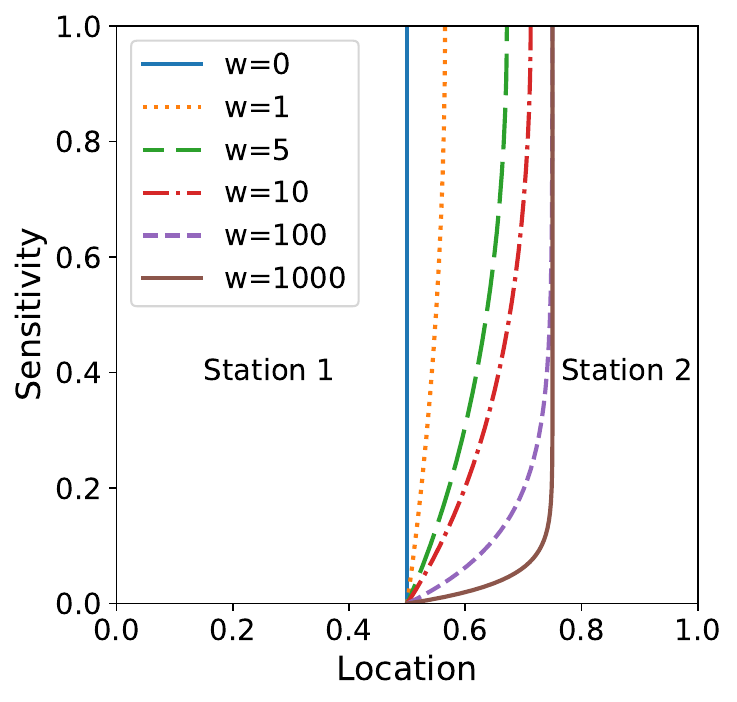}
    \end{minipage}
    }
\par\vspace{0.5ex}
\begin{minipage}{0.95\linewidth}
\footnotesize\emph{Note.}  
In all three panels, $(c_1,c_2)$ are fixed at $(3,1)$. In panels (a)-(c), the waiting cost $w$ is set to be $10$ while the reward $r$ takes values $0.45$, $0.65$ and $2$, respectively. In panel (d), the reward $r$ is taken a sufficiently large value and waiting cost $w$ varies from $0$ to $1000$. 
Panels (a)-(c) show the transition from disjoint partial service to complete service as reward $r$ increases.  The gray region indicates those demand types whose optimal assigned station varies with $r$. Panel (d) illustrates how the boundary shifts to balance station workloads as the waiting cost $w$ becomes dominant. 
\end{minipage}
\end{figure}

One notable phenomenon in \cref{Example:Hotelling_Uniform} distinct from \cref{Example:Hotelling_Homogeneous} is the allocation switching. Under the optimal matching, as the reward $r$ increases from $0$ to infinity, each demand type in \cref{Example:Hotelling_Homogeneous} either remains not served or becomes deterministically allocated to a specific station once $r$ exceeds a certain threshold. In contrast, in \cref{Example:Hotelling_Uniform}, where time sensitivity is heterogeneous, there exist demand types whose assigned stations vary with $r$. The Grey area in \Cref{fig:hotelling_all} illustrates such a phenomenon. When the service reward is low, a customer located closer to station 2 lacks the incentive to incur travel or waiting costs for service. As the service reward reaches an intermediate level, demand is assigned to the nearest station (station 2). However, as the service reward increases further, more high-time-sensitivity customers residing near station 2 are allocated to that station. Since the optimal matching prioritizes customers with higher time sensitivity, the allocation will switch to the more efficient station (station 1) to avoid longer waiting times. The gray region in \Cref{fig:hotelling_all} highlights these swing demand types. In a setting without capacity constraints, once demand at location $\mathbf{x}$ is matched to station 1, increasing or decreasing the reward does not change the relative cost difference between stations. Hence, there is no incentive for the planner to reassign it to station 2, which was initially more costly to travel to. In contrast, in our framework, the temporal cost is not fixed but endogenously shaped by the matching plan. As the global reward parameter changes, the optimal allocation shifts, thereby altering temporal costs and potentially making previously dominated destinations attractive.

The impact of waiting costs on the optimal matching is most clearly illustrated in the complete-service regime. When $r$ is sufficiently large such that any exclusion is suboptimal, Case 3 in \Cref{prop:Hotelling_Uniform} shows that the boundary relies on $w$ for given $c_i$, $i=1,2$. Specifically, when the scale of waiting cost $w$ increases from $0$ to $\infty$ (here $r$ will also go to infinity to ensure no exclusion), the boundary monotonically shifts from $x=1/2$ and approaches $x=c_1/(c_1+c_2)$. In the extreme case where waiting costs are negligible, the temporal dimension of the problem vanishes. The optimal partition is determined solely by minimizing transportation costs, yielding a boundary at the midpoint. Conversely, when waiting costs are dominant, they outweigh transportation considerations. In this case, the planner seeks to minimize total system waiting time by allocating demand in proportion to each station's service capacity to balance workloads.

\section{Numerical Study}\label{sec:numerical}
In the numerical study, we consider a vaccination-planning problem motivated by the Toronto COVID-19 rollout. Demand is distributed across the city, and individuals are partitioned into four age-based urgency classes. A set of capacitated vaccination sites, including city-operated clinics and hospital clinics, provides service over time. The planner jointly assigns each demand unit to a service location and a service time to minimize the total social cost, defined as the sum of travel and temporal costs. We assume that total capacity is sufficient to serve the entire population and that the vaccination reward is sufficiently large so that all individuals are scheduled for vaccination. Under this assumption, maximizing social welfare is equivalent to minimizing the total social cost, defined as the sum of travel and temporal costs.

Temporal cost captures the social burden of delayed protection. A later appointment leaves an individual unvaccinated for a longer period, thereby increasing the risk of adverse health outcomes. We model this cost in reduced form as increasing linearly with waiting time, with the slope depending on the individual’s urgency class. In this sense, temporal cost reflects the idea that delaying vaccination is more harmful for higher-risk populations.

Spatial cost captures the burden of accessing the assigned vaccination site. In addition to standard transportation inconvenience, longer trips may increase exposure during travel and may also reduce the effective vaccination rate by making attendance more difficult. We measure travel cost by Euclidean distance (in km) between the individual’s location and the assigned vaccination site.

To construct demand, we use 2021 Canadian Census data \citep{statcan_marital_status_9810012901} to approximate the spatial population distribution in Toronto. The population is partitioned into four age-based risk groups: highest risk (65+), higher risk (40--64), lower risk (20--39), and lowest risk (0--19). We use these groups as proxies for urgency classes. To calibrate the temporal-cost weights, we adopt the age-specific COVID-19 mortality estimates reported in \cite{levin2020assessing}. Specifically, for the four groups, we evaluate the fitted mortality at ages 75, 52, 30, and 10, respectively. The temporal cost is defined as $
l_j(t) = 0.1\left(1 + 2\times\mathrm{IFR}_j\right)t
$, where $\mathrm{IFR}_j\%$ is the infection fatality rate (mortality) for each group $j$, and the constant captures the additional general health expenditure associated with unvaccinated populations. Note that we intentionally assign a relatively larger weight to temporal cost than to spatial cost so that the model places greater emphasis on earlier protection, especially for high-risk populations.  
On the supply side, we consider two categories of vaccination sites, city-operated immunization clinics and hospital immunization clinics, with locations obtained from \cite{toronto_covid_immunization_clinics}. This dataset is not fully comprehensive and does not include several vaccination channels used in practice, such as pharmacies, mobile clinics, and drive-through clinics. Our analysis should therefore be interpreted as a stylized comparison based on the major fixed-site channels captured in the available location data.

We assume that all sites of the same type share the same capacity profile, and that total capacity within each site type is allocated equally across locations of that type. We focus on first-dose scheduling. The total demand assigned to the modeled system is 1.65 million individuals (with the remaining population assumed to be vaccinated through other channels), with an age composition proportional to the census distribution. Aggregate capacity over the planning horizon is set to 2.19 million doses over 300 days, providing some slack relative to total demand. The time profile of capacity is calibrated from Toronto vaccination statistics\footnote{https://public.tableau.com/app/profile/tphseu/viz/COVID-19VaccinationsinToronto/COVID-19VaccinationsinToronto-Public} using a smoothed fit.

The resulting first-dose capacity is hump-shaped over time. Early in the rollout, capacity rises as sites expand staffing, logistics, and operating capability. After roughly two months, however, second-dose appointments begin to consume a growing share of site resources, reducing the capacity available for first-dose service. This generates a natural increase-then-decrease pattern in first-dose capacity over the planning horizon.

To evaluate the benefit of joint spatial and temporal coordination, we compare the proposed spatiotemporal matching with several benchmark rules that impose simpler operational structures. As benchmarks, we consider simpler rules that decouple the spatial and temporal decisions. On the spatial side, we consider two assignment rules. Under the \emph{Capacity} rule, individuals are assigned across sites using semi-definite optimal transport to minimize spatial costs, subject to each site's total capacity over the planning horizon, without considering individual urgency. Under the \emph{Distance} rule, individuals are assigned to their nearest site without considering capacity, which induces a Voronoi-type spatial partition of the city. On the temporal side, we consider two scheduling rules for the assigned population. Under \emph{Random} scheduling, individuals assigned to a site are scheduled uniformly at random over the available service opportunities. Under \emph{Priority} scheduling, individuals are scheduled in decreasing order of urgency, so that higher-risk groups receive earlier service. For each combination, we report the average spatial cost, average temporal cost, and average total social cost among vaccinated individuals, together with the uncovered rate. 
\begin{table}[t]
\centering
\caption{Comparison of per-capita social cost under different assignment rules.}
\label{tab:percapita_cost_comparison}
\begin{tabular}{lcccr}
\toprule
Rule 
& Spacial Cost
& Temporal Cost
& Total Cost
& Uncovered Rate \\
\midrule
Spatiotemporal 
& 2.99
& \textbf{21.47}
& \textbf{24.46}
& 0.00\% \\
Capacity\_Priority
& 2.79
& 22.49
& 25.28
& 0.00\% \\
Capacity\_Random 
& 2.79
& 31.48
& 34.27
& 0.00\% \\
Distance\_Priority 
& \textbf{2.46}
& 25.32
& 27.78
& 16.66\% \, (8.16,\,8.00,\,0.50,\,0.00) \\
Distance\_Random
& \textbf{2.46}
& 30.81
& 33.27
& 16.66\% \, (3.15,\,5.28,\,5.49,\,2.74) \\
\bottomrule
\end{tabular}

\vspace{0.5ex}
\parbox{0.96\linewidth}{\footnotesize
Notes: The costs are the average cost among vaccinated individuals. For the uncovered rate, the numbers in parentheses indicate the breakdown of the total by risk group, ordered from highest to lowest risk. }
\end{table}

\Cref{tab:percapita_cost_comparison} summarizes the results and yields three main insights. First, the spatiotemporal matching achieves the lowest total social cost, outperforming the second-best result by 3.24\%. Its gain comes mainly from reducing temporal cost, while spatial cost remains close to that under the capacity-based rules. This shows the value of jointly optimizing location and service time. To illustrate the magnitude, a 3.24\% improvement can translate into meaningful gains at the scale of social services and healthcare systems. 
Second, among the full-coverage benchmarks, Capacity\_Priority performs close to the spatiotemporal optimum, whereas Capacity\_Random and Distance\_Random perform substantially worse. Under our calibration, delay is relatively costly and capacity slack is limited, so the temporal schedule remains a first-order determinant of performance. Third, the distance-based rules attain the lowest spatial cost but perform poorly overall. They incur much higher temporal costs and leave 16.66\% of the population uncovered, indicating that assigning individuals solely by distance can overload nearby sites and fail to use system capacity efficiently.

\begin{figure}[p]
\centering
\caption{Illustration of Vaccination Scheduling.} \label{fig:vac}
\subfigure[Spatial partition for lowest risk group]{
\begin{minipage}[b]{0.45\textwidth}
    \includegraphics[width=1\textwidth]{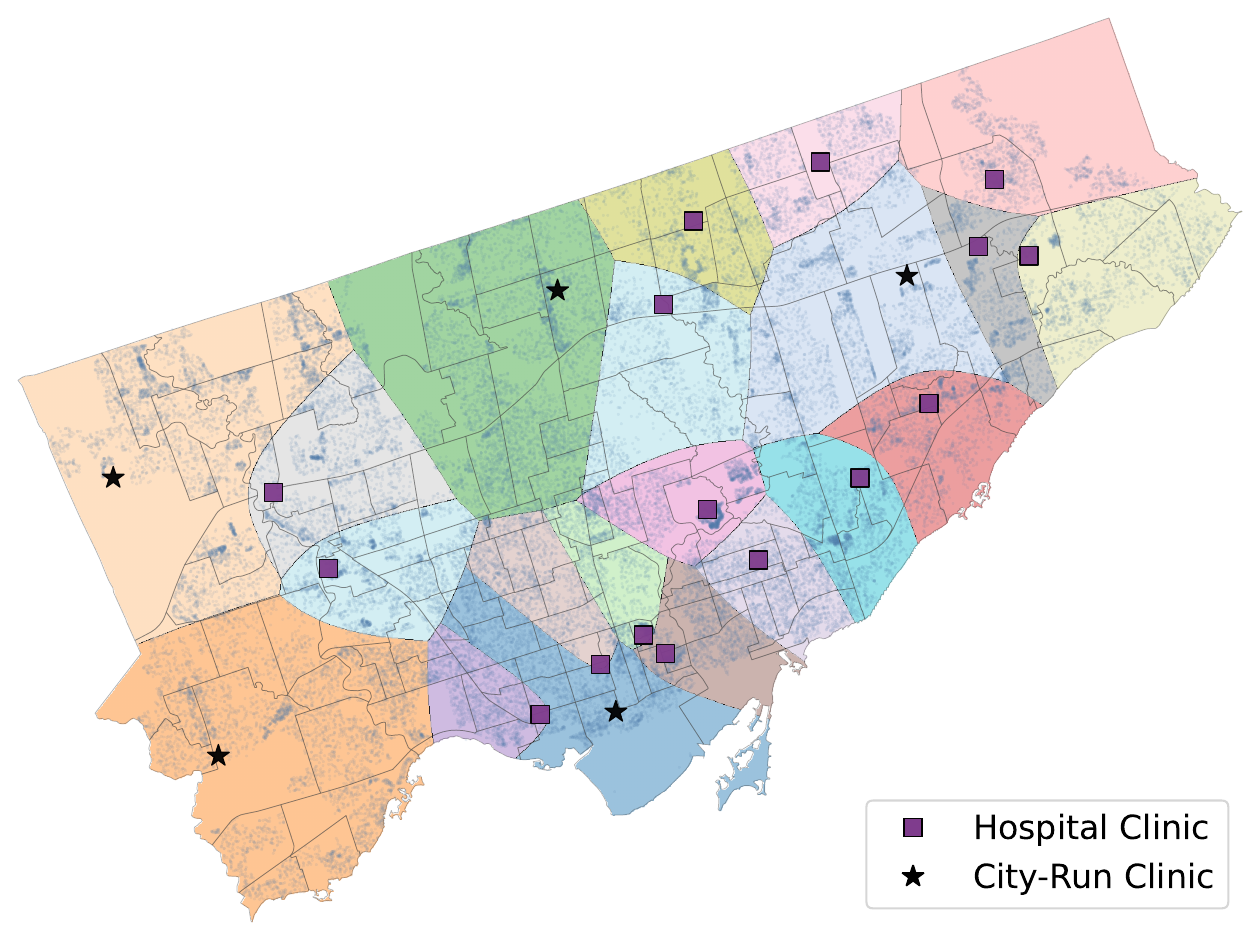}
\end{minipage}
    }    
\subfigure[Spatial partition for lowest risk group]{
\begin{minipage}[b]{0.45\textwidth}
    \includegraphics[width=1\textwidth]{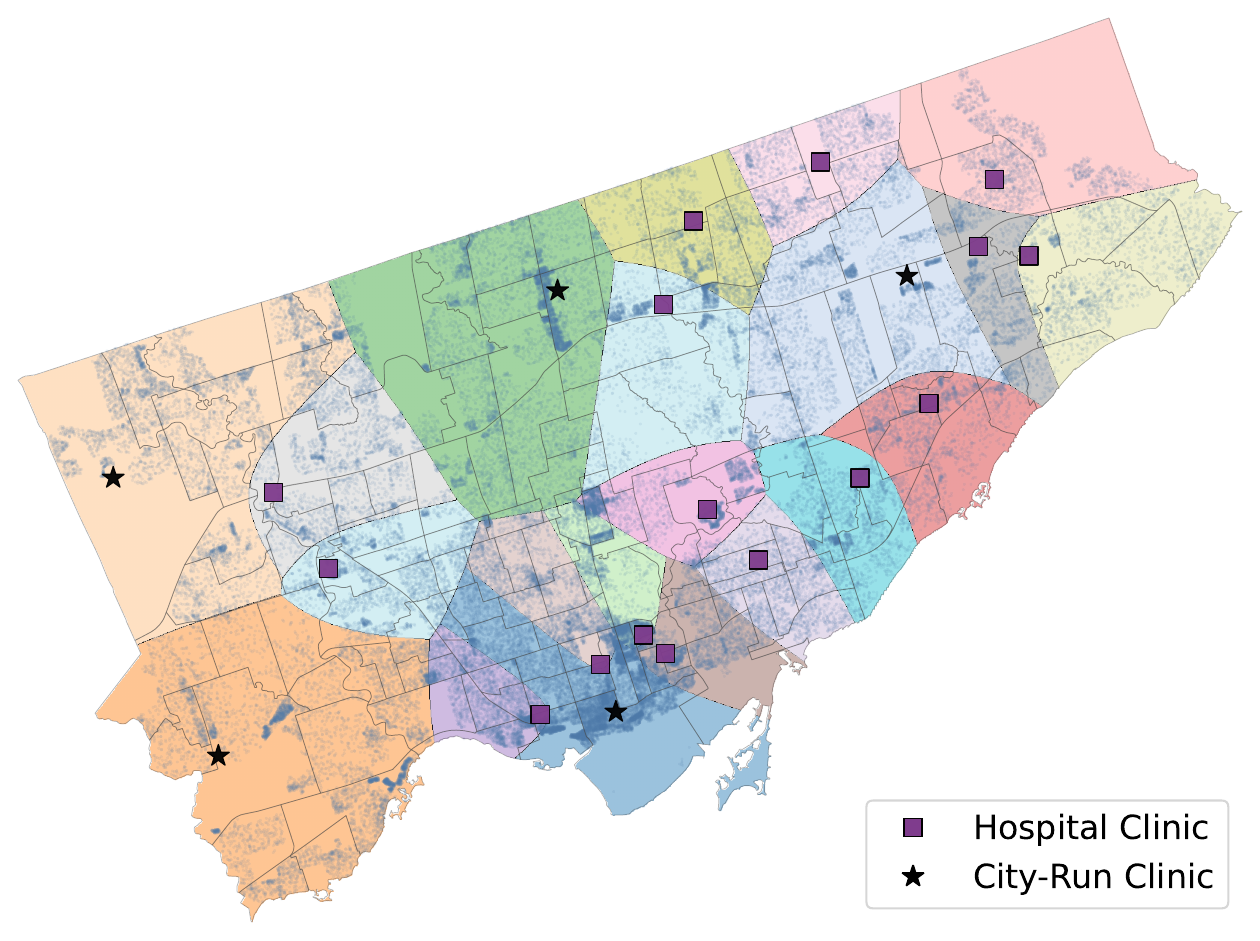}
\end{minipage}
    }  \\
\subfigure[Spatial partition for higher risk group]{
\begin{minipage}[b]{0.45\textwidth}
    \includegraphics[width=1\textwidth]{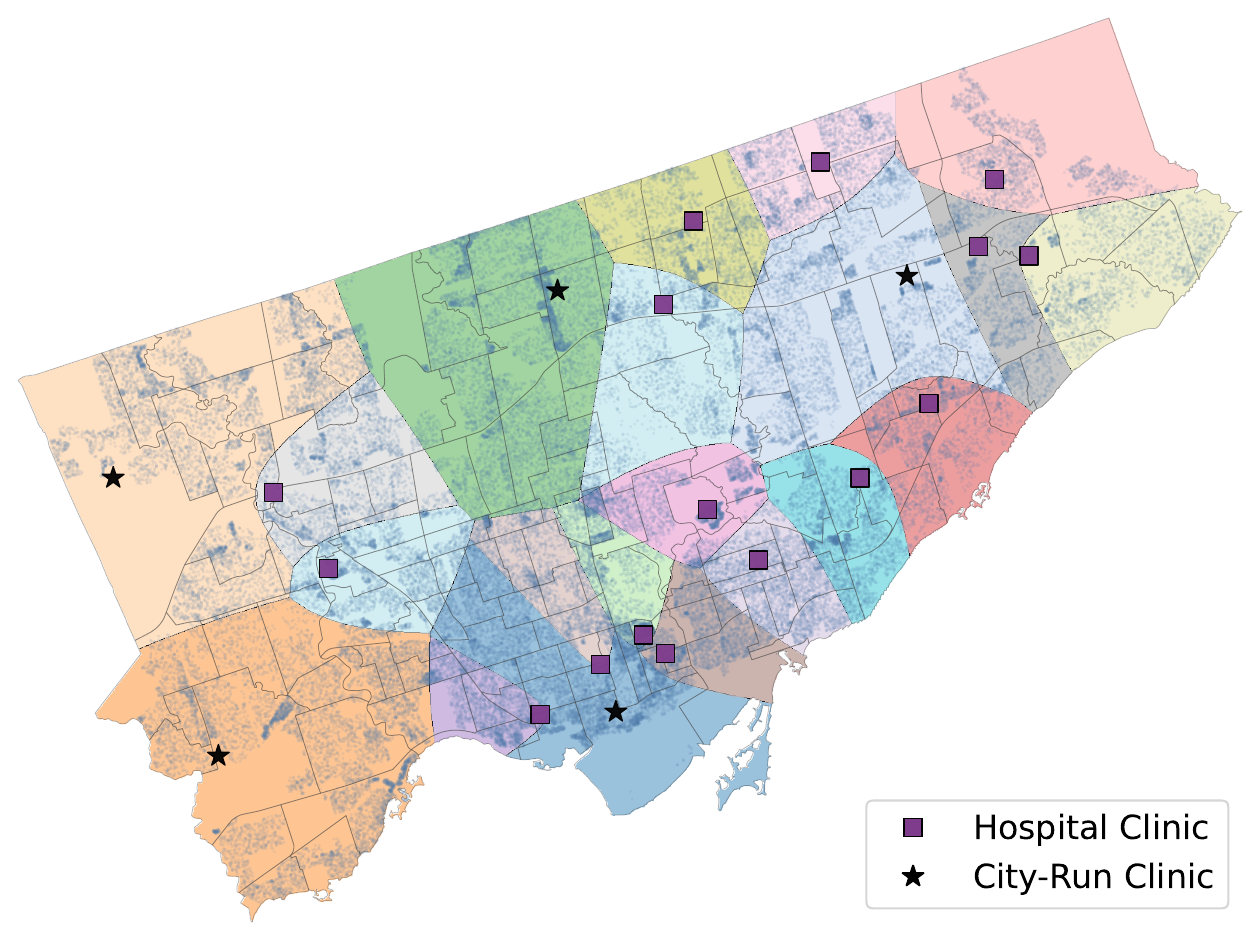}
\end{minipage}
    }    
\subfigure[Spatial partition for highest risk group]{
\begin{minipage}[b]{0.45\textwidth}
    \includegraphics[width=1.05\textwidth]{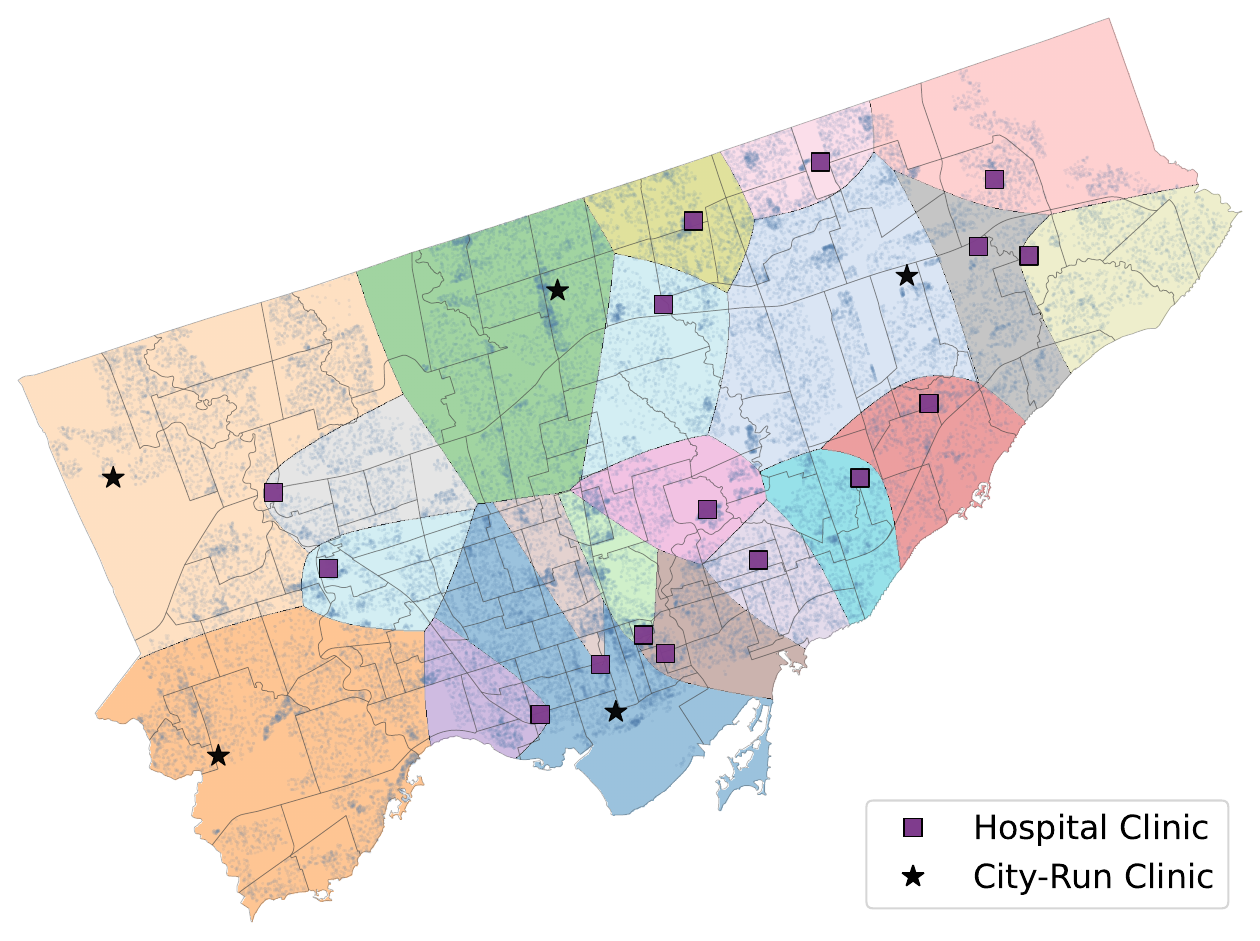}
\end{minipage}
    } \\
\subfigure[Metro Toronto Convention Centre]{
\begin{minipage}[b]{0.45\textwidth}
    \includegraphics[width=1.05\textwidth]{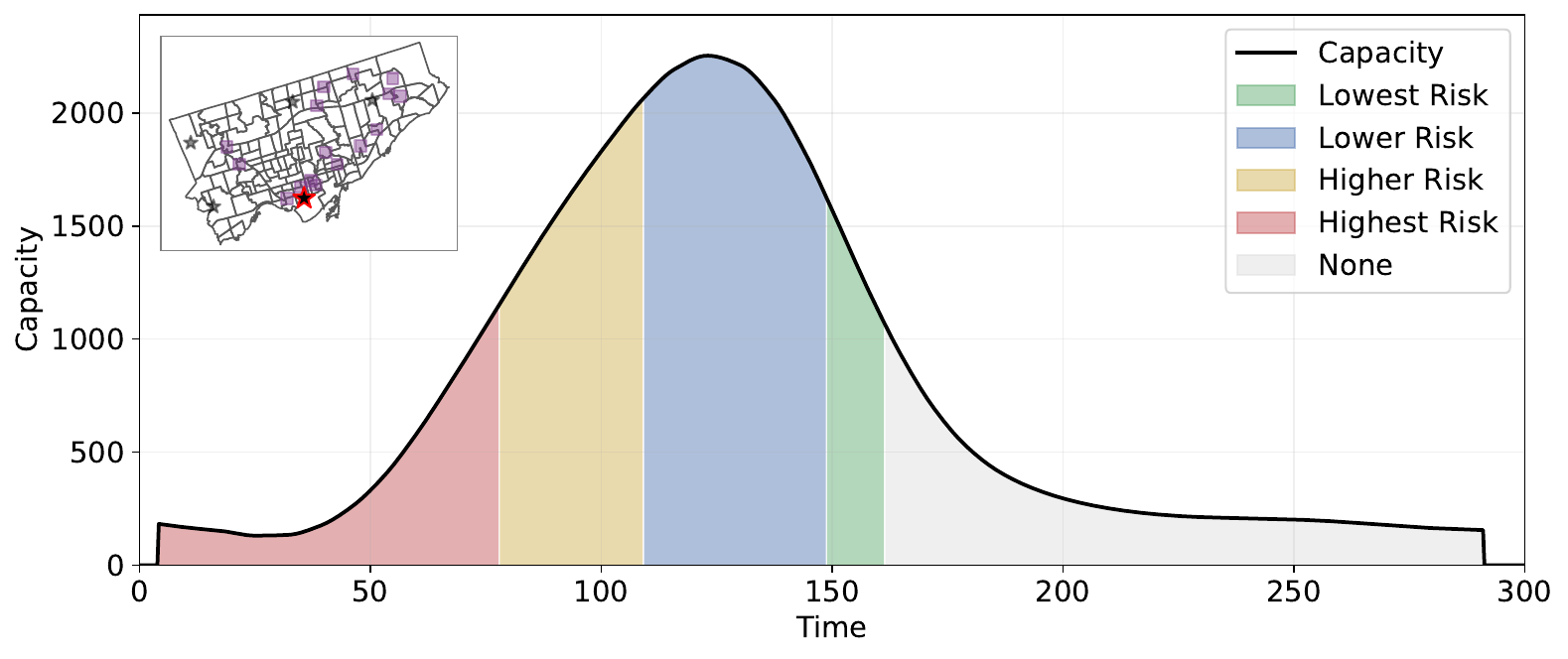}
\end{minipage}
    }    
\subfigure[Michael Garron Hospital]{
\begin{minipage}[b]{0.45\textwidth}
    \includegraphics[width=1.05\textwidth]{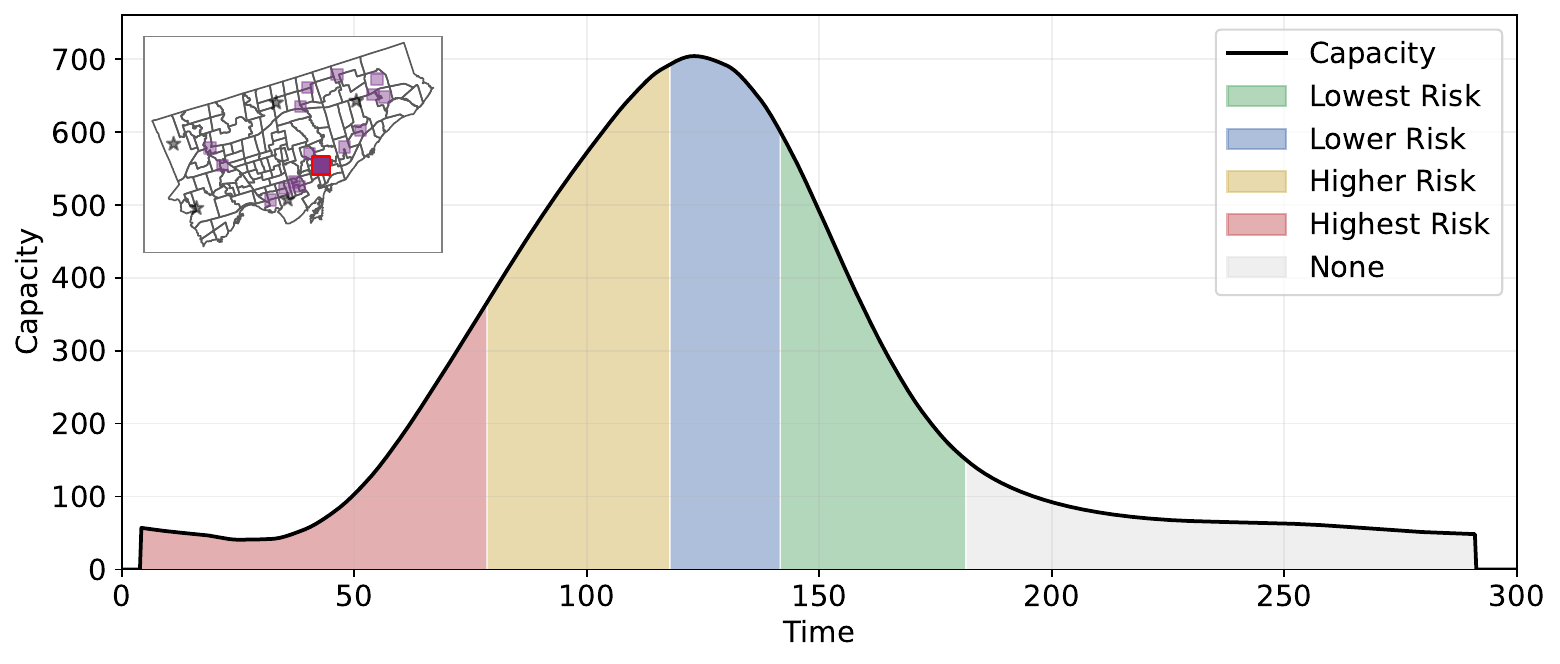}
\end{minipage}
    } \\
\subfigure[Cloverdale Mall]{
\begin{minipage}[b]{0.45\textwidth}
    \includegraphics[width=1.05\textwidth]{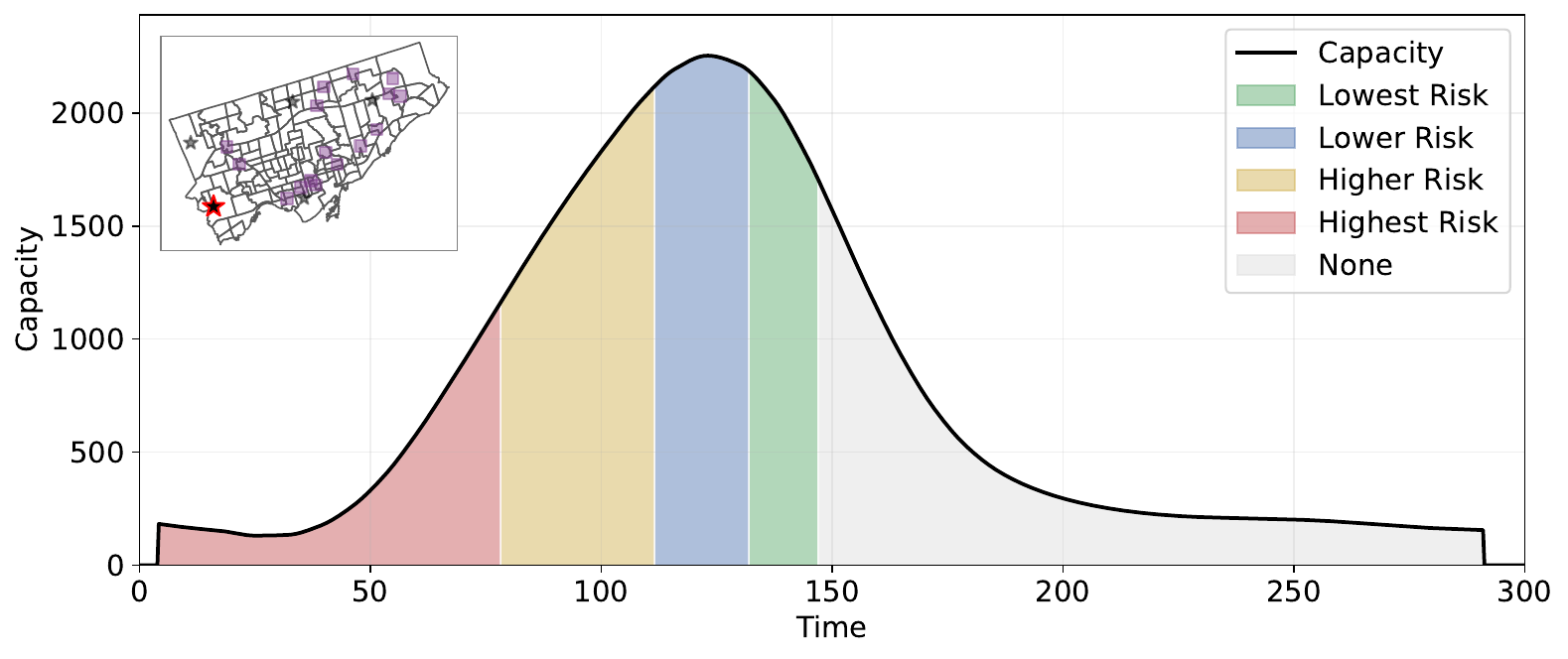}
\end{minipage}
    }    
\subfigure[SHN Centenary Clinic]{
\begin{minipage}[b]{0.45\textwidth}
    \includegraphics[width=1.05\textwidth]{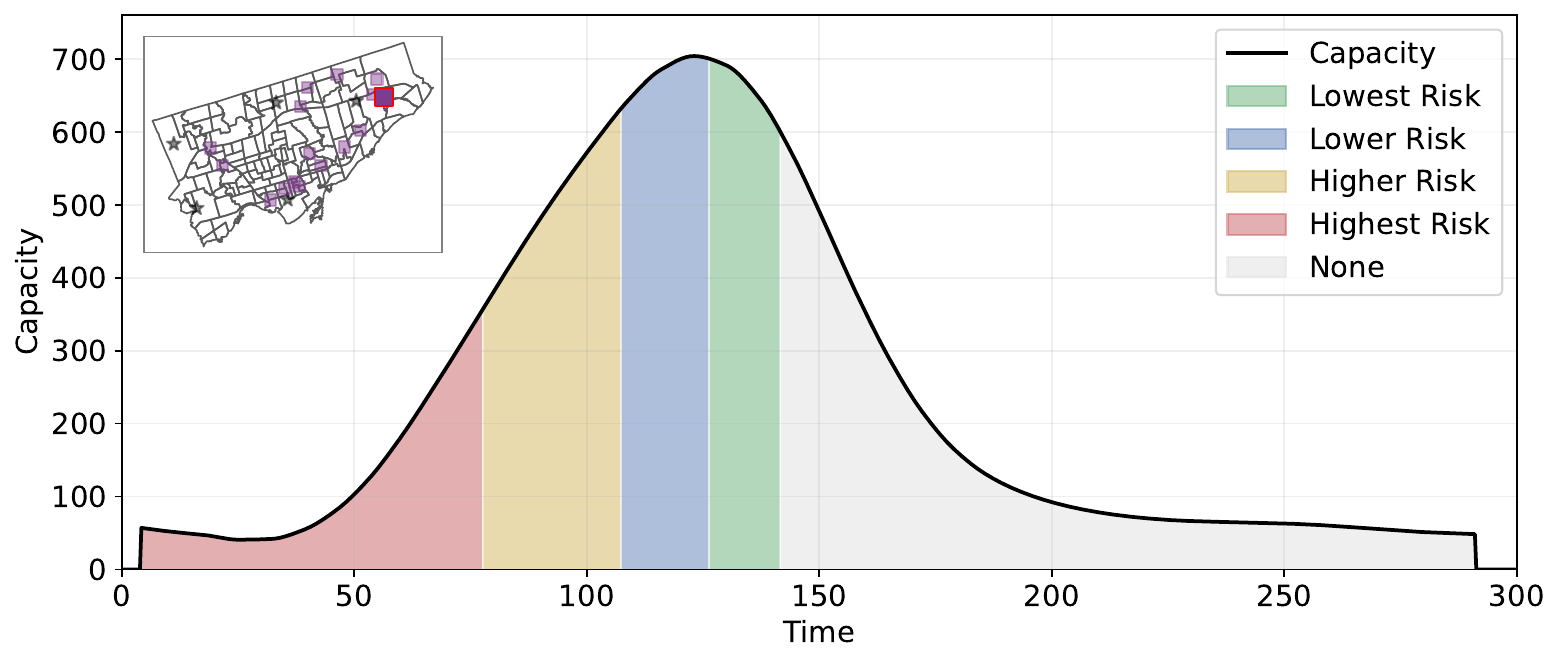}
\end{minipage}
    } 
\begin{minipage}{0.95\linewidth}
\footnotesize \emph{Note.} Panels (a)–(d) display the optimal spatial partitions for the four risk groups. Panels (e)–(h) display the corresponding temporal schedules at four representative sites. The solid curve represents the site’s capacity profile, colored areas represent allocations to different risk groups, and the gray area represents unused capacity. 
\end{minipage}
\end{figure}

\Cref{fig:vac} visualizes the optimal spatiotemporal vaccination plan. As established in \Cref{prop:sensitive_priority}, service follows a clear priority order across all four representative sites. At the same time, the four spatial partitions are not identical: although assignments remain broadly local, the boundaries move across groups, especially in contested middle regions. For example, the downtown city-operated site serves a noticeably larger region for the highest-risk group because its larger capacity makes it particularly valuable for serving more urgent individuals earlier. The site-level schedules also differ in slackness: some sites are fully used within their effective capacity window, whereas others retain slack in the late period. Notably, the highest-risk group is served over nearly the same time window across the four sites, while the service windows for lower-risk groups vary much more across sites, suggesting that the optimal policy tightly coordinates early capacity for the most urgent populations, whereas lower-risk groups are scheduled more flexibly depending on each site’s location and remaining capacity.

\section{Conclusion}\label{sec:conclusion}
In this paper, we investigate the capacitated spatiotemporal matching problem, in which a planner must decide both ``where" and ``when" to match. We develop a theoretical framework based on OT theory to model the allocation of heterogeneous demand across a set of capacitated service stations, with the objective of maximizing social welfare by balancing service rewards against spatial and temporal costs. For the practical case of a finite number of demand types, we show that this complex, infinite-dimensional problem can be reduced to a tractable, finite-dimensional convex optimization problem. This allows the implementation of the optimal matching via an envy-free pricing mechanism, which can be implemented in a finite number of discrete time slots. 

Our analysis reveals several key structural properties of the optimal matching policy. We demonstrate that service regions are not static partitions but are defined by generalized Laguerre cells that adjust to demand characteristics such as time sensitivity. For scheduling, a subtle tension arises when demands have different sensitivities but the same preferred times: while more time-sensitive demands are prioritized with service times closer to their ideal, the system may prefer to admit less-sensitive demands to maximize overall efficiency. Furthermore, when demands have different preferred times but the same sensitivity, the optimal schedule preserves the chronological order of their preferences. Finally, by extending our analysis to a spatiotemporal Hotelling model, we uncover allocation switching: a phenomenon where a demand's assigned station can change as the service reward varies. This finding, along with the others, reveals structural behaviors that do not arise in classical optimal transport problems, underscoring the distinctive nature of spatiotemporal matching problems under capacity constraints. Our numerical study demonstrates the practical relevance of the framework. 

Despite these contributions, this study has limitations that suggest avenues for future research. To isolate the structural properties of optimal plans under capacity constraints, we consider an offline deterministic setting. These choices yield clear analytical insights, but they also point to several important directions for future research, including dynamic online settings with stochastic demand arrivals or random service times, as well as endogenous capacity under uncertainty (e.g., supply disruptions) and dynamic capacity redeployment across stations.

\bibliographystyle{informs2014}
\bibliography{reference}
\ECSwitch
\input{EC}
\end{document}

%% file: EC.tex
\ECHead{E-Companion for ``Capacitated Spatiotemporal Matching''}
\section{Davenport–Schinzel Sequences}\label{sec:D_S_sequence}
In this section, we introduce the definition of Davenport-Schinzel sequence and a few classical facts and bounds on the function $\lambda_s(n)$. 

\begin{definition}[Davenport-Schinzel sequence]\label{def:D_S_sequence}
Given the set $[n]$ and an integer $s\ge 1$, a finite sequence over $[n]$ is said to be a Davenport–Schinzel sequence of order $s$ if it satisfies:
\begin{itemize}
    \item No two consecutive numbers are the same;
    \item For any $i,j\in[n]$ be two distinct values occurring in the sequence, then the sequence does not contain a subsequence $... i, ... j, ..., i, ..., j, ...$ consisting of $s + 2$ values alternating between $i$ and $j$.
\end{itemize}
\end{definition}

The function $\lambda_s(n)$ denotes the length of the longest (n,s)-Davenport-Schinzel sequence. The first few values of $\lambda_s(n)$ are well known:
\[
\lambda_0(n)=1, \qquad 
\lambda_1(n)=n, \qquad 
\lambda_2(n)=2n-1.
\]
Thus, $\lambda_s(n)$ is linear for $s \le 2$.  
For $s=3$, the growth becomes slightly superlinear and involves the inverse Ackermann function $\texttt{a}(n)$, which is defined as follows.

\begin{definition}
The \emph{Ackermann function} is a classic example of a computable function that grows faster than any primitive recursive function. It is defined recursively as follows:
\[
Ac(0, n) = n+1, \qquad 
Ac(m+1, 0) = Ac(m,1), \qquad 
Ac(m+1, n+1) = Ac\bigl(m, Ac(m+1,n)\bigr).
\]
\end{definition}

Even for small arguments, $Ac(m,n)$ grows extremely rapidly. 
For example, $Ac(2,2)=7$, $Ac(3,3)=61$ and $Ac(4,4)=2^{2^{2^{65536}}}-3$. Because of this explosive growth, the Ackermann function serves as a useful benchmark for describing ``slowly growing'' functions defined via its inverse, which is defined as follows:  
\begin{definition}
The \emph{inverse Ackermann function}, denoted by $\texttt{a}(n)$, is defined as
\[
\texttt{a}(n) = \min\{\,k \ge 1 : Ac(k,k) \ge n\,\}.
\]

Intuitively, $\texttt{a}(n)$ measures how many times the Ackermann function must be iterated to exceed a given number~$n$. Due to the fast growth of $Ac(n,n)$, its inverse $\texttt{a}(n)$ grows extremely slowly so that for all practical input sizes, $\texttt{a}(n) \le 4$.     
\end{definition}

Based on $\texttt{a}(n)$, we have the following bounds for $\lambda$ (see, e.g. \citealtec{pettie2015sharp}).
\[
\lambda_3(n)=\Theta(n\,\texttt{a}(n)),\qquad\lambda_4(n)=\Theta\left(n\,2^{\texttt{a}(n)}\right).
\]
Hence, for small $s$, $\lambda_s(n)$ remains nearly linear in $n$, with extremely slow superlinear growth. 

Davenport--Schinzel sequences play a central role in computational geometry, where $\lambda_s(n)$ bounds the combinatorial complexity of lower envelopes of $n$ univariate functions whose pairwise intersections are bounded by~$s$. For example, when each pair of functions intersects at most $s$ times, the number of intervals in the lower envelope is at most $\lambda_s(n)$.
\begin{lemma}[\citealtec{agarwal2000davenport}]\label{lemma:lower_envelope_DS} 
Let $f_1, f_2, \dots, f_n : \mathbb{R} \to \mathbb{R}$ be continuous univariate functions such that any two distinct functions $f_i$ and $f_j$ intersect at most $s$ times. Then, the combinatorial complexity of the lower envelope
\[
E(x) = \min_{1 \le i \le n} f_i(x)
\]
is bounded above by $\lambda_s(n)$, where $\lambda_s(n)$. Equivalently, as $x$ ranges over $\mathbb{R}$, the number of distinct intervals on which a single function equivalent to the envelope is at most~$\lambda_s(n)$.
\end{lemma}

This lemma is instrumental in establishing the practical implementability of our optimal matching policy via finite time slots (\Cref{prop:finite_slots}). By characterizing the combinatorial complexity of the lower envelope, this lemma guarantees that the maximization problem defining our optimal pricing schedule yields a trajectory with a finite, bounded number of switching points rather than arbitrary swings. Therefore, it bridges the gap between theoretical optimality and managerial feasibility.
\section{An Illustration of Spatiotemporal Matching}\label{sec:Fig_STMatching}
In \Cref{fig:STMatching}, we present an example of spatiotemporal matching with synthetic data. In this example, each demand type has its own time preference and temporal cost, and the transportation cost is modeled as the Euclidean distance.
\begin{figure}[p]
    \centering
    \caption{An Illustration of Spatiotemporal Matching.}\label{fig:STMatching}
    \label{fig:space_type_1}
    \subfigure[Piecewise-linear temporal costs $\ell_j(t)$.]{
\begin{minipage}[b]{0.45\textwidth}
    \centering
    \includegraphics[width=1\textwidth]{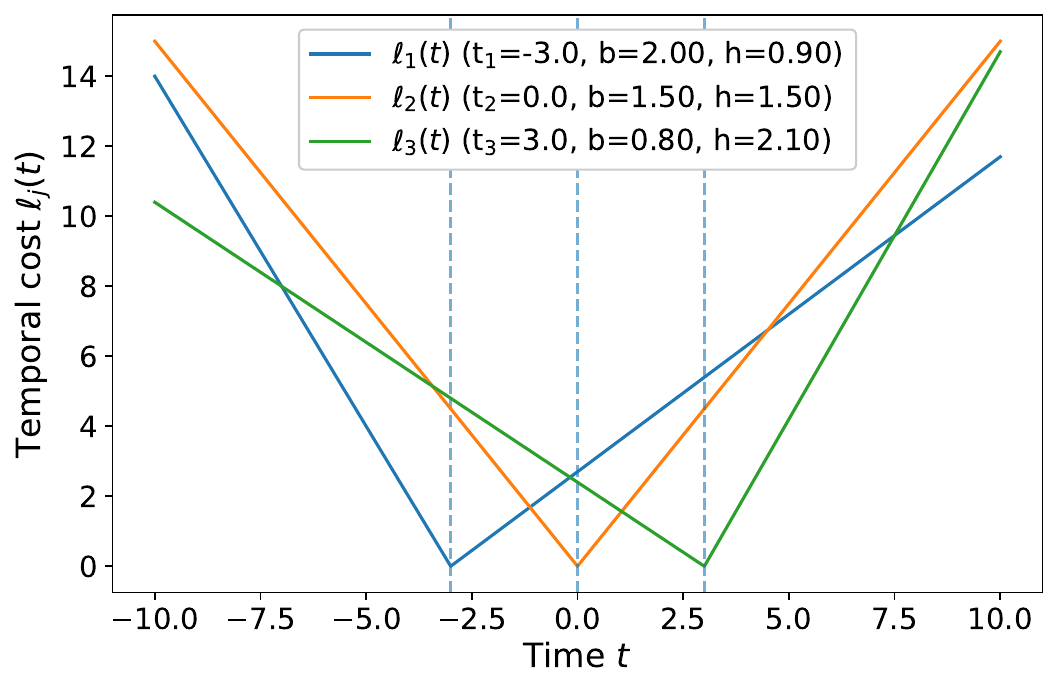}
    \end{minipage}
    }    \\
    \subfigure[Density of type-$1$ demand.]{
\begin{minipage}[b]{0.3\textwidth}
    \centering
    \includegraphics[width=1.05\textwidth]{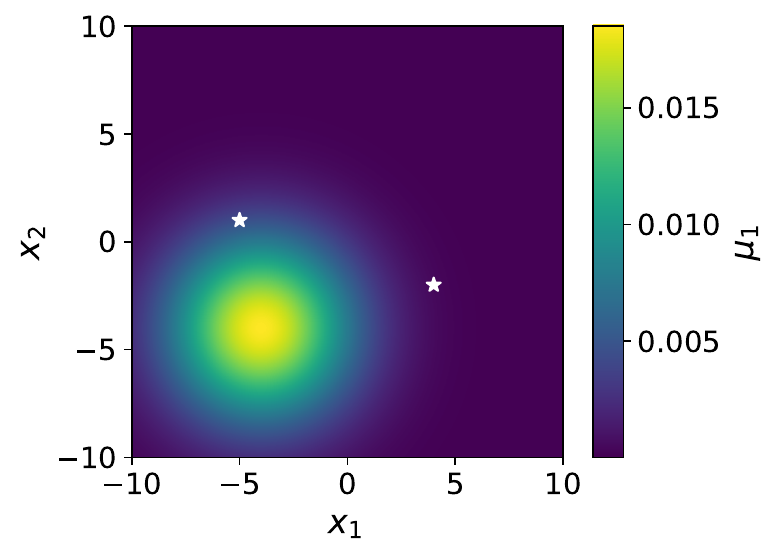}
    \end{minipage}
    }
    \subfigure[Density of type-$2$ demand.]{
\begin{minipage}[b]{0.3\textwidth}
    \centering
    \includegraphics[width=1.05\textwidth]{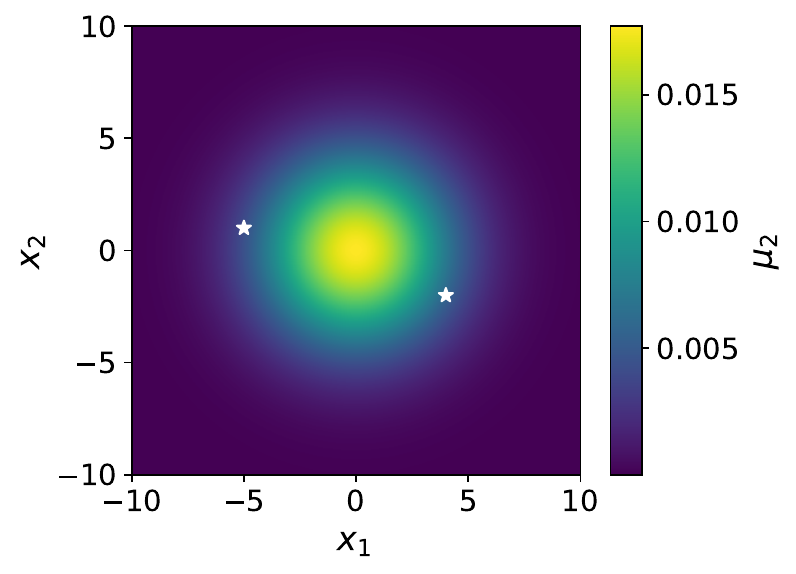}
    \end{minipage}
    }
    \subfigure[Density of type-$3$ demand.]{
\begin{minipage}[b]{0.3\textwidth}
    \centering
    \includegraphics[width=1.05\textwidth]{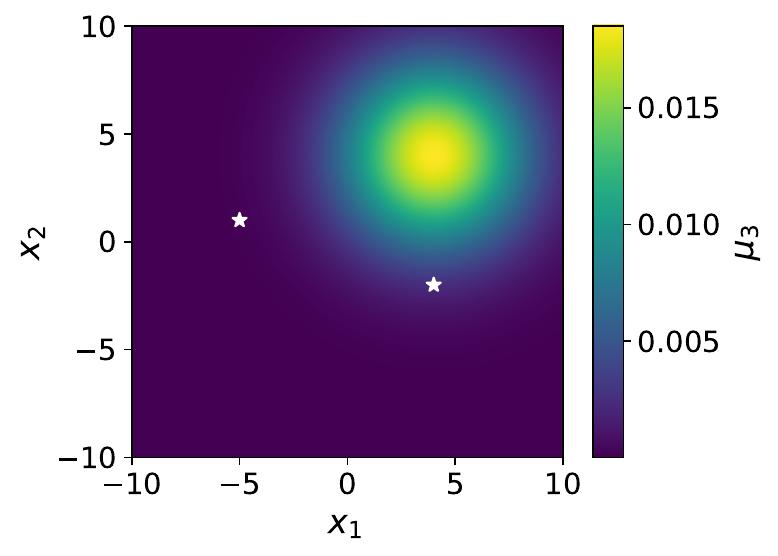}
    \end{minipage}
    }\\
    \subfigure[Space partition (type-$1$).]{
\begin{minipage}[b]{0.3\textwidth}
    \centering
    \includegraphics[width=1\textwidth]{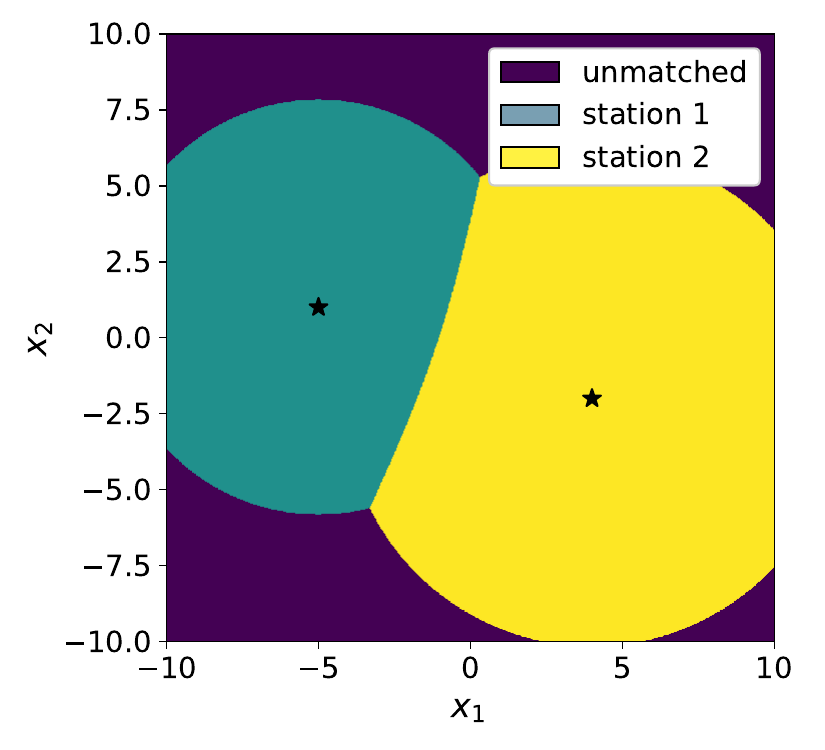}
    \end{minipage}
    }
    \subfigure[Space partition (type-$2$).]{
\begin{minipage}[b]{0.3\textwidth}
    \centering
    \includegraphics[width=1\textwidth]{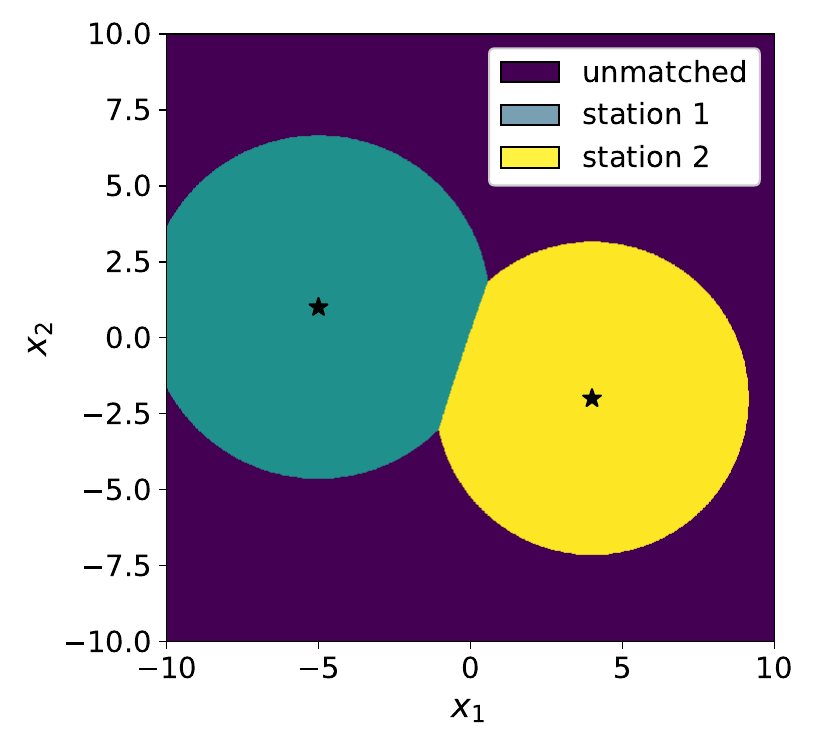}
    \end{minipage}
    }
    \subfigure[Space partition (type-$3$).]{
\begin{minipage}[b]{0.3\textwidth}
    \centering
    \includegraphics[width=1\textwidth]{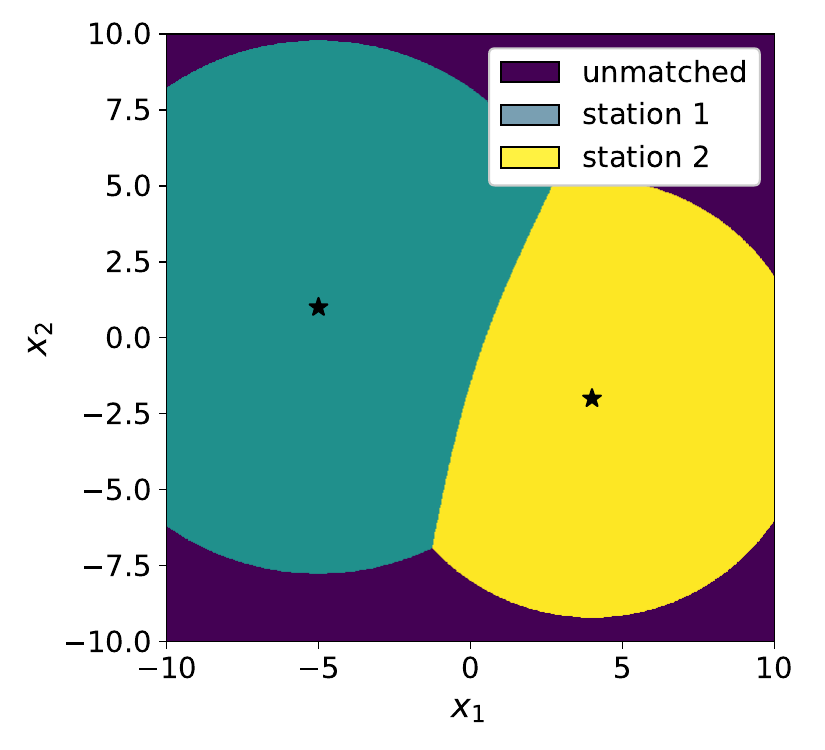}
    \end{minipage}
    }\\
    \subfigure[Time scheduling (station $1$).]{
\begin{minipage}[b]{0.45\textwidth}
    \centering
    \includegraphics[width=1\textwidth]{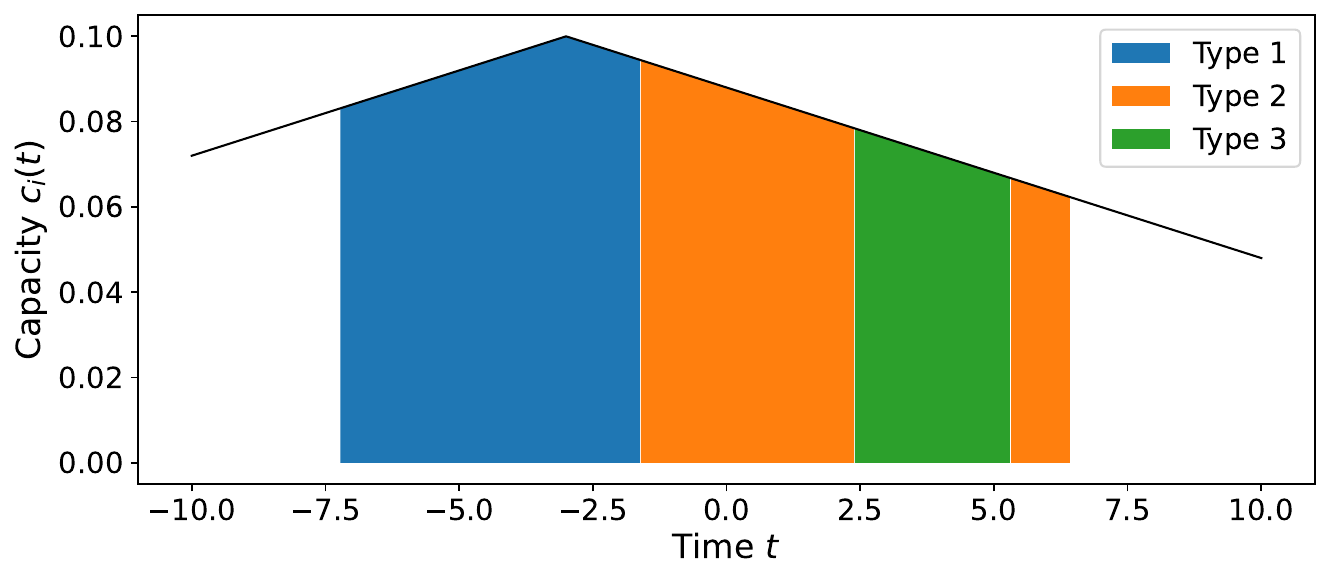}
    \end{minipage}
    }
    \subfigure[Time scheduling (station $2$).]{
\begin{minipage}[b]{0.45\textwidth}
    \centering
    \includegraphics[width=1\textwidth]{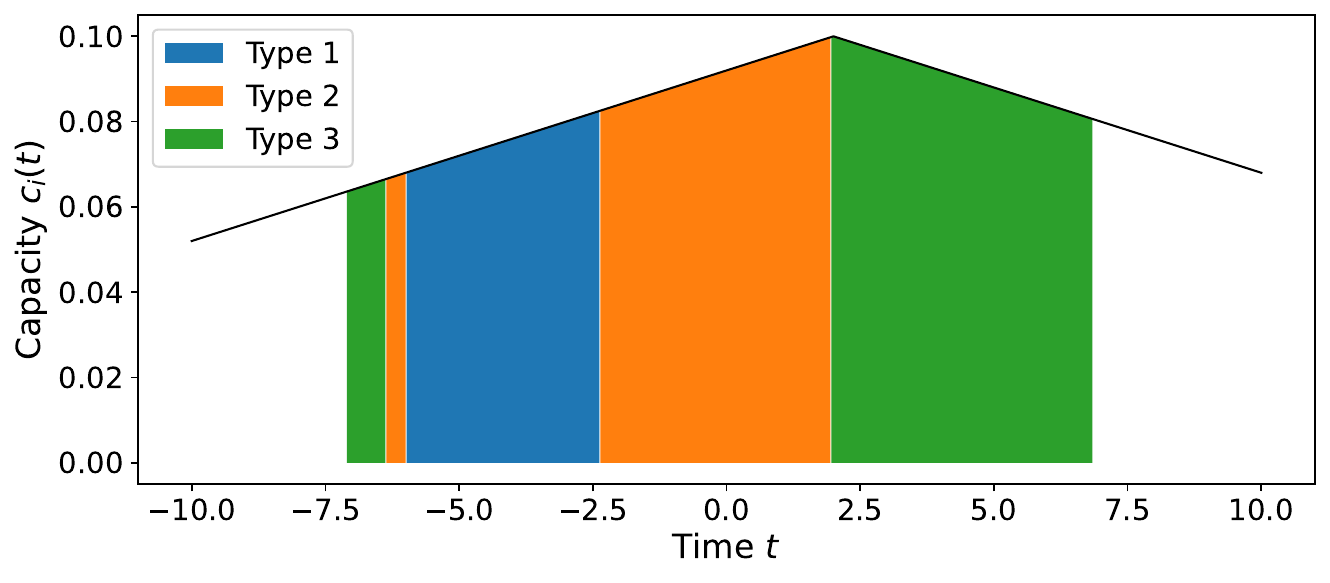}
    \end{minipage}
    }
\begin{minipage}{0.95\linewidth}
\footnotesize \emph{Note.} There are three demand types and two service stations, with Station 1 on the left and Station 2 on the right, both marked with ``$\bm\star$". Panel (a) shows the temporal cost associated with each demand type. Panels (b)–(d) illustrate the demand density functions $\mu_j(\bm{x})$ for $j = 1, 2, 3$. Panels (e)–(g) depict the spatial service regions $\mathcal{C}_{i,j}$. Finally, panels (h)–(i) present the capacity functions $c_i(t)$ and the corresponding temporal partitions $\mathcal{T}_{i,j}$ for Stations $i = 1, 2$.
\end{minipage}
\end{figure}
\section{The Inverse Matrix of $A$ in \Cref{prop:finite_linear}}\label{Sec:inverse_A}
A direct computation gives that:
$$A^{-1}_{ij} =
\begin{cases}
\displaystyle \frac{1}{s_1-s_2}, & i=j=1, \\
\displaystyle \frac{s_{i-1}-s_{i+1}}{(s_{i-1}-s_i)(s_i-s_{i+1})}, & i=j, i\ge 2, \\
\displaystyle -\frac{1}{s_i-s_{i+1}}, & j=i+1, \\
\displaystyle -\frac{1}{s_{i-1}-s_i}, & j=i-1, \\
0, & \text{otherwise}.
\end{cases} \quad i,j=1,\dots,m.
$$
\section{Capacity Configuration: Dispersion or Concentration}\label{sec:D_or_C}
In many service systems, planners choose between operating multiple stations or pooling capacity at a single hub. We study this choice when all demand is served, so the trade-off is clear: dispersion shortens travel times, whereas concentration reduces efficiency imbalances across locations. We compare the two designs and provide simple conditions under which each wins, along with a Hotelling-style threshold that serves as a practical rule of thumb for consolidation.

\begin{example}[Complete Service with Two Stations]\label{Example:Complete Service with Two Stations} 
Suppose \Cref{assm:long-horizon-linear-cost} holds and there is only one demand type (with sensitivity $s$). There are two stations ($n=2$), and $r$ is sufficiently large so that all demand is served. There are two possible system designs: \emph{dispersion} with capacity $c_1$ and $c_2$ at stations 1 and 2, respectively, and \emph{concentration}, with capacity $c_1+c_2$ and $0$ at stations 1 and 2, respectively.
\end{example}

\begin{proposition}\label{prop:dispersion_concentration}
Let $\eta^*_1$ and $\eta^*_2$ be the optimal solution in \cref{prop:finite_linear} for dispersion case and assume $|\delta(\bm{x},\bm{y}_1)-\delta(\bm{x},\bm{y}_2)|\leq D$. Then we have
\begin{itemize}
    \item If $\eta^*_2\geq \max\{\eta^*_1,\frac{(\eta^*_2-\eta^*_1)c_1}{2(c_1+c_2)}\}$, dispersion is better than concentration in terms of social welfare.
    \item If $\eta^*_2\leq \frac{(\eta^*_1-\eta^*_2)^2c_1}{2D(c_1+c_2)}$, concentration is better than dispersion in terms of social welfare.
\end{itemize}   
\end{proposition}
\begin{proof}{Proof.}
When there is only one demand type, the problem \eqref{eq:finite_linear_Dual} can be simplifies as
\begin{equation*}
\min_{\eta_{1},\eta_{2}\in\mathbb{R}}\quad
\frac{ s}{\beta} (c_1\eta_{1}^2+c_2\eta_{2}^2)+
\int_{\mathcal{X}}\left(\max_{i\in[2]}\left\{r-\delta(\bm{x},\bm{y}_i)-\eta_{i}\right\}\right)^+\mu(\bm{x})\mathrm{d}\bm{x}.
\end{equation*}
Since $r$ is large enough, we know that the term in the integral is always positive, and we can remove $r$ from the problem, that is
\begin{equation}\label{eq:min_eta_12}
\min_{\eta_{1},\eta_{2}\in\mathbb{R}}\quad
\frac{ s}{\beta} (c_1\eta_{1}^2+c_2\eta_{2}^2) -
\int_{\mathcal{X}}\max_{i\in[2]}\left\{\delta(\bm{x},\bm{y}_i)+\eta_{i}\right\}\mu(\bm{x})\mathrm{d}\bm{x}.
\end{equation}

According to the first order condition, the amount of demand that assigned to the station 1 and 2 are $c_1\eta^*_1$ and $c_2\eta^*_2$ respectively and the total amount of demand is then $c_1\eta^*_1 + c_2\eta^*_2$. By concentrating all processing capacity at station 1, the temporal cost decreases
\begin{equation*}
\frac{c_1^2\eta^{*2}_1}{2c_1}+ \frac{c_2^2\eta^{*2}_2}{2c_2} -\frac{(c_1\eta^*_1 + c_2\eta^*_2)^2}{2(c_1+c_2)} = \frac{c_1 c_2 (\eta_1^*-\eta_2^*)^2}{2(c_1+c_2)}.
\end{equation*}

Now switch to spatial cost. Let $\mathcal{X}_1$ be the area where $\delta(\bm{x},\bm{y}_1)+\eta^*_{1}\leq \delta(\bm{x},\bm{y}_2)+\eta^*_{2}$. similarly, define $\mathcal{X}_2$ analogously. $\mathcal{X}_1$ and $\mathcal{X}_2$ represent the region that is covered by station 1 or 2. All demand in $\mathcal{X}_2$ has now been reassigned to station 1 after concentrating. The increment of spatial cost is at least $\int_{\mathcal{X}_2}(\delta(\bm{x},\bm{y}_1)-\delta(\bm{x},\bm{y}_2))\mu(\bm x)\mathrm{d}\bm x \geq (\eta^*_2-\eta^*_1)c_2\eta^*_{2}$. Then, separation yields better social welfare if

\begin{equation*}
(\eta^*_2-\eta^*_1)c_2\eta^*_{2}\geq \frac{c_1 c_2 (\eta_2^*-\eta_1^*)^2}{2(c_1+c_2)} \Leftarrow \eta^*_2\geq \max\left\{\eta^*_1,\frac{(\eta^*_2-\eta^*_1)c_1}{2(c_1+c_2)}\right\}.
\end{equation*}

On the other hand, the increase of spatial cost is at most $D c_2\eta^*_2$, that is, the upper bound of the spatial cost increment of reassigning all demand to station 1. Then, concentration yields better social welfare if

\begin{equation*}
D c_2\eta^*_2 \leq \frac{c_1 c_2 (\eta_1^*-\eta_2^*)^2}{2(c_1+c_2)} \Leftarrow \eta^*_2\leq \frac{(\eta^*_1-\eta^*_2)^2c_1}{2D(c_1+c_2)}.
\end{equation*}
Therefore, we proved the two statements in the proposition.\hfill\qedsymbol
\end{proof}

The condition $|\delta(\bm{x},\bm{y}_1)-\delta(\bm{x},\bm{y}_1)|\leq D$ is often guaranteed by the triangle inequality in practical settings. For instance, if the transportation cost is Lipschitz continuous with respect to distance, then for any location $\bm{x}$, the absolute difference in transportation costs between stations 1 and 2 cannot exceed a fraction of the distance between these two stations. As previously discussed, $c_2 \eta^*_2$ is proportional to the volume of demand served at station 2, while $\eta^*_2$ is proportional to the operating duration of station 2. When $\eta^*_1 = \eta^*_2$, \Cref{Prop:Opt_c} shows that social welfare is maximized by evenly dispersing capacity, in agreement with the first statement of \Cref{Example:Complete Service with Two Stations}. Conversely, if $\eta^*_1 \neq \eta^*_2$, this example reveals that the classic ``dispersion versus concentration" decision is governed by the interplay between spatial and temporal costs. When station 2's demand volume is robust, its role in reducing transportation costs justifies its operation. However, when its volume falls below a threshold point, it imposes a temporal cost on the entire system that outweighs its spatial benefits. In such cases, consolidating capacity into a single, highly efficient hub emerges as the better strategy for maximizing social welfare.

\begin{example}[Dispersion or Concentration]\label{Example:D_or_C}
Consider the setting described in \Cref{Example:Hotelling_Homogeneous}, where the reward is sufficiently large to ensure that all demand must be served. Suppose the principal encounters a ``concentration'' opportunity, wherein station $2$ can be canceled, and its capacity reallocated to station $1$, along with an additional bonus capacity $c_0$. As a result, the efficiencies of stations $1$ and $2$ become $c_1 + c_2 + c_0$ and $0$, respectively. 
\end{example}

\begin{proposition}\label{prop:Hotel_D_or_C}
In \Cref{Example:D_or_C}, the preferable scheme that yields higher social welfare under optimal matching is determined as follows:
\begin{itemize}
    \item If $c_1 c_2\geq w^2$, dispersion always better than concentration irrelative of $c_0$.
    \item If $c_1 c_2< w^2$, concentration is better than dispersion if and only if 
    $$
    c_0 > \frac{c_1+3w}{w^2-c_1c_2}c_1c_2 .
    $$
\end{itemize}    
\end{proposition}
\begin{proof}{Proof.}
Since $r$ is large enough such that all demand can not be excluded, the social welfare maximization is equivalent to cost (spatial and temporal) minimization. By result in \Cref{Example:Hotelling_Homogeneous}, the total cost of dispersion is
\begin{equation*}
\frac{(c_2+w)(c_1+w)}{2(2 c_1 c_2 + (c_1 + c_2) w)},
\end{equation*}
and the total cost of concentration is
\begin{equation*}
\frac{1}{2} + \frac{w}{2(c_1+c_2+c_0)},
\end{equation*}
where the first term $1/2$ is the transportation cost and the second term is the waiting cost.

A concentration having a smaller cost than dispersion is equivalent to 
\begin{equation}\label{ineq:DorC}
\frac{1}{2} + \frac{w}{2(c_1+c_2+c_0)}-\frac{(c_2+w)(c_1+w)}{2(2 c_1 c_2 + (c_1 + c_2) w)} = \frac{
    c_1 c_2 (c_1 + c_2 + c_0) + 2 w c_1 c_2 - w^2 c_0}{2 (c_1 + c_2 + c_0) \left( 2 c_1 c_2 + (c_1 + c_2) w \right)}<0.
\end{equation}

If $c_1c_2\geq w^2 $, we have $ c_1 c_2 (c_1 + c_2 + c_0) + 2 w c_1 c_2 - w^2 c_0 \geq  4w^3 >0$, which means the inequality \eqref{ineq:DorC} can not hold, thus dispersion is better.

If $c_1c_2< w^2 $, since the denominator of the right-hand side in \eqref{ineq:DorC} is strictly positive. \eqref{ineq:DorC} is equivalent to  $ c_1 c_2 (c_1 + c_2 + c_0) + 2 w c_1 c_2 - w^2 c_0 <0$ which can be further simplified as 
\begin{equation*}
c_0 > \frac{c_1 c_2 (c_1 + c_2) + 2 w c_1 c_2}{w^2 - c_1 c_2}.
\end{equation*}
\hfill\qedsymbol
\end{proof}
This example highlights a trade-off in service system design: whether to maintain multiple dispersed stations or to consolidate capacity into a single location. In this example, maintaining both stations is always preferred when the product of the initial efficiencies is sufficiently high relative to the temporal cost. However, when efficiencies are more limited, consolidating capacity into a single station is worthwhile only if the additional capacity gained from the concentration effect exceeds a specific threshold. 
\section{Omitted Proofs}

\subsection{Proof of \Cref{prop:general_Dual}}
\begin{proof}{Proof.}
Denote $v(\bm{x},\bm{y}_i,t,\alpha)= r-\delta(\bm{x},\bm{y}_i) - \ell(t,\alpha)$, be the matching value. 
We extend the set $\mathcal{X}$, $\mathcal{A}$ and $\mathcal{T}$ by attaching an isolated dummy point $\hat{\infty}$ respectively. That is: $\hat{\mathcal{X}} = \mathcal{X}\cup\{\hat{\infty}\}$, $\hat{\mathcal{A}} = \mathcal{A}\cup\{\hat{\infty}\}$, and  $\hat{\mathcal{T}} = \mathcal{T}\cup\{\hat{\infty}\}$. For the demand side, assign additional point mass on dummy demand $(\hat{\infty},\hat{\infty})$ with density $\sum_{i=1}^{n}\int_{\mathcal{T}}c_i(t)\mathrm{d}t$. on supply side, assign additional processing capability $\int_{\mathcal{X}\times\mathcal{A}}\mu(\bm{x},\alpha)\mathrm{d}t/n$ on dummy time $(\hat{\infty})$ for each station $i$. We define 

\begin{equation*}
\begin{aligned}
&\hat{v}(\hat{\infty},\bm{y}_i,t,\hat{\infty})= 0\quad\forall i\in[n], t\in\hat{\mathcal{T}},\\
&\hat{v}(\bm{x},\bm{y}_i,\hat{\infty},\alpha) = 0 \quad\forall i\in[n], (\bm{x},\alpha)\in\hat{\mathcal{X}}\times\hat{\mathcal{A}}.
\end{aligned}
\end{equation*}
That is, any matching involving dummy demand or dummy time has an exact net value of $0$, reflecting the case of no matching. Therefore (by construction), the standard complete optimal transportation problem with extended space and parameters gives exactly the same objective of problem \eqref{eq:STM} and hence the strong duality holds due to Monge–Kantorovich duality theorem (see, e.g., \citealtec{villani2008optimal}). Therefore, it suffices to show that its optimal dual objective coincides with \eqref{eq:STM_D}.

The dual problem of the extended primal problem is presented below.
\begin{equation}\label{eq:P_General_Extend_Dual}
\begin{aligned}
\inf_{\hat{\phi},\hat{\varphi}_i}&\quad
\int_{\mathcal{X}\times\mathcal{A}} \hat{\phi}(\bm{x},\alpha)\mu(\bm{x},\alpha)\mathrm{d}(\bm{x},\alpha) + \hat{\phi}(\hat{\infty},\hat{\infty})\sum_{i=1}^{n}\int_{\mathcal{T}}c_i(t)\mathrm{d}t \\
&\quad + \sum_{i=1}^{n}\int_{\mathcal{T}} \hat{\varphi}_i(t)c_i(t) \mathrm{d}t + \frac{\int_{\mathcal{X}\times\mathcal{A}}\mu(\bm{x},\alpha)\mathrm{d}(\bm{x},\alpha)}{n}\sum_{i=1}^{n}\hat{\varphi}_i(\hat{\infty})\\
\text{s.t.}&\quad 
\hat{\phi}(\bm{x},\alpha) + \hat{\varphi}_i(t) \geq \hat{v}(\bm{x},\bm{y}_i,t,\alpha),\quad\forall (\bm{x},\alpha,t)\in\hat{\mathcal{X}}\times\hat{\mathcal{A}}\times\hat{\mathcal{T}}. \\
\end{aligned}
\end{equation}

Since any translation $\hat{\phi} + C$ and $\hat{\varphi}_i - C$ is equivalent to problem \eqref{eq:P_General_Extend_Dual}, it is safe to assume $\hat{\phi}(\hat{\infty},\hat{\infty})=0$ in the remainder of the proof. Then the constraint for $\hat{\varphi}_i(\hat{\infty})$ is equivalent to
$$
\hat{\varphi}_i(\hat{\infty}) \geq \sup_{(\bm{x},\alpha)\in\hat{\mathcal{X}}\times\hat{\mathcal{A}}} \hat{v}(\bm{x},\bm{y}_i,\hat{\infty},\alpha)-\hat{\phi}(\bm{x},\alpha) =-\inf_{(\bm{x},\alpha)\in\hat{\mathcal{X}}\times\hat{\mathcal{A}}}\hat{\phi}(\bm{x},\alpha)\coloneqq -\hat{\phi}_{\min}.
$$

So all $\hat{\varphi}_i(\infty) $ will take value $-\hat{\phi}_{\min}$ and the objective of \eqref{eq:P_General_Extend_Dual} becomes:
\begin{equation}\label{eq:Extend_Dual_obj}
\int_{\mathcal{X}\times\mathcal{A}}\left( \hat{\phi}(\bm{x},\alpha)-\hat{\phi}_{\min}\right)\mu(\bm{x},\alpha)\mathrm{d}(\bm{x},\alpha) + \sum_{i=1}^{n}\int_{\mathcal{T}} \hat{\varphi}_i(t)c_i(t) \mathrm{d}t.
\end{equation}

Moreover, since $\hat{\varphi}_i(t)\geq \hat{v}(\hat{\infty},\bm{y}_i,t,\hat{\infty})-\hat{\phi}(\hat{\infty},\hat{\infty})=0$, it follows that any feasible $\hat{\varphi}_i$ is nonnegative. Therefore, given any feasible solution to \eqref{eq:P_General_Extend_Dual}(again, with $\hat{\phi}(\hat{\infty},\hat{\infty})$ selected to be $0$), one can construct a feasible solution for \eqref{eq:STM_D} by letting $\phi(\bm{x},\alpha) = \hat{\phi}(\bm{x},\alpha)-\hat{\phi}_{\min}\geq 0$ for $(\bm{x},\alpha)\in\mathcal{X}\times\mathcal{A}$ and $\varphi_i(t)=\hat{\varphi}_i(t)$ for $t\in\mathcal{T}$. The objective value of \eqref{eq:STM_D} then matches \eqref{eq:Extend_Dual_obj}, thus the dual gap is $0$. 
\hfill\qedsymbol
\end{proof}

\subsection{Proof of \Cref{Thm:STB}}
The proof contains two parts. We first show the equivalence between \eqref{eq:STM} and \eqref{eq:STB}. Then we show that \eqref{eq:STB} is a convex optimization problem.
\paragraph{Part I: Equivalence.}

First, each feasible $\pi$ induces feasible $q$ by taking marginal densities:
\begin{equation*}
\begin{aligned}
q_i(\alpha)= \int_{\mathcal{X}\times\mathcal{T}}\pi_{i}(\bm{x},\alpha,t)\mathrm{d}(\bm{x},t).
\end{aligned}
\end{equation*}
Moreover, let $\nu$ and $\theta$ in the subproblem $OT_{\mathrm{time}}$ and $OT_{\mathrm{space}}$ be
\begin{equation*}\label{eq:construction_pi}
\nu_{i}(t,\alpha) = \int_{\mathcal{X}}\pi_{i}(\bm{x},\alpha,t)\mathrm{d}\bm{x},
\quad\theta_{i}(\bm{x},\alpha) = \int_{\mathcal{T}}\pi_{i}(\bm{x},\alpha,t)\mathrm{d}t.
\end{equation*}
It is not difficult to see that such $\nu$ and $\theta$ are feasible. Then we have 
\begin{equation*}
\begin{aligned}
OT_{\mathrm{time}}(\bm{c},\bm{q})\leq & \sum_{i=1}^{n}\int_{\mathcal{T}\times\mathcal{A}}\ell(t,\alpha)\nu_{i}(t,\alpha)\mathrm{d}(t,\alpha) =    \sum_{i=1}^{n}\int_{\mathcal{X}\times\mathcal{T}\times\mathcal{A}}\ell(t,\alpha)\pi_{i}(\bm{x},\alpha,t)\mathrm{d}(t,\alpha),\\
OT_{\mathrm{space}}(\bm{q},\mu)\leq & \sum_{i=1}^{n}\int_{\mathcal{X}\times\mathcal{A}} \delta(\bm{x},\bm{y}_i)\theta_{i}(\bm{x},\alpha)\mathrm{d}(\bm{x},\alpha) =    \sum_{i=1}^{n}\int_{\mathcal{X}\times\mathcal{T}\times\mathcal{A}}\delta(\bm{x},\bm{y}_i)\pi_{i}(\bm{x},\alpha,t)\mathrm{d}(t,\alpha).    
\end{aligned}
\end{equation*}
Therefore, we show that 
\begin{equation*}
\begin{aligned}
&\sum_{i=1}^{n}\int_{\mathcal{A}}rq_i(\alpha)\mathrm{d}\alpha-OT_{\mathrm{time}}(\bm{c},\bm{q})-OT_{\mathrm{space}}(\bm{q},\mu)\\ &\geq \sum_{i=1}^{n}\int_{\mathcal{X}\times\mathcal{A}\times\mathcal{T}}
\left(r- \delta(\bm{x },\bm{y}_i) - \ell(t,\alpha)\right)\pi_{i}(\bm{x},\alpha,t)\mathrm{d}(\bm{x},\alpha,t).
\end{aligned}
\end{equation*}
Therefore \eqref{eq:STB} serves as an upper bound of \eqref{eq:STM}.

Conversely, given feasible $\hat{\bm q}$ and corresponding $\hat{\bm \nu}$ and $\hat{\bm \theta}$ in subproblems, one can construct $\pi$ as:
\begin{equation}
\hat\pi_{i}(\bm{x},\alpha,t) = 
\begin{cases}
\frac{\hat\nu_{i}(t,\alpha) \hat\theta_{i}(\bm{x},\alpha)}{\hat q_i(\alpha)} & \hat q_i(\alpha)>0,\\
0& \hat q_i(\alpha)=0.
\end{cases}
\end{equation}
Whenever $\hat q_i(\alpha)>0$, such construction ensures that
\begin{equation*}
\begin{aligned}
\int_{\mathcal{T}}\hat\pi_{i}(\bm{x},\alpha,t) = &\frac{\int_{\mathcal{T}}\hat\nu_{i}(t,\alpha) \mathrm{d}t \hat\theta_{i}(\bm{x},\alpha)}{\hat q_i(\alpha)}  = \hat\theta_{i}(\bm{x},\alpha),  \\
\int_{\mathcal{X}}\hat\pi_{i}(\bm{x},\alpha,t)= &\frac{\hat\nu_{i}(t,\alpha)  \int_{\mathcal{X}}\hat\theta_{i}(\bm{x},\alpha)\mathrm{d}\bm{x}}{\hat q_i(\alpha)} =\hat\nu_{i}(t,\alpha).
\end{aligned}
\end{equation*}
where both the second equations are due to the equality constraints in $OT_{\mathrm{time}}$ and $OT_{\mathrm{space}}$. Therefore we have 
\begin{equation*}
\begin{aligned}
&\sum_{i=1}^{n}\int_{\mathcal{A}}r\hat q_i(\alpha)\mathrm{d}\alpha-Obj_{\mathrm{time}}(\hat{\nu})-Obj_{\mathrm{space}}(\hat{\theta}) \\
&= \sum_{i=1}^{n}\int_{\mathcal{X}\times\mathcal{A}\times\mathcal{T}}
\left(r- \delta(\bm{x },\bm{y}_i) - \ell(t,\alpha)\right)\hat\pi_{i}(\bm{x},\alpha,t)\mathrm{d}(\bm{x},\alpha,t).
\end{aligned}
\end{equation*}
where $Obj_{\mathrm{time}}(\hat{\nu})$ and $Obj_{\mathrm{space}}(\hat{\theta})$ denote the objective functions in $OT_{\mathrm{time}}$ and $OT_{\mathrm{space}}$ respectively, evaluated at $\nu=\hat{\nu}$ and $\theta=\hat{\theta}$. Therefore \eqref{eq:STB} also serves as a lower bound of \eqref{eq:STM}. This demonstrates that the optimal objective in \eqref{eq:STM} coincides with that in \eqref{eq:STB}, and that any optimal solution to \eqref{eq:STB} recovers the optimal solution to \eqref{eq:STM} by taking $\hat{\nu}$ and $\hat{\theta}$ attain their optimum. 

\paragraph{Part II: Concavity.}
Let $Obj(\bm{q})$ denote the objective function of \eqref{eq:STB} evaluated at $\bm q$. Given two feasible $\bm{q}^{(1)}$ and $\bm{q}^{(2)}$ and a linear combination $\bm{q}^{(\lambda)}=\lambda \bm{q}^{(1)}+ (1-\lambda)\bm{q}^{(2)}$. Let $\nu^{(k)}$ and $\theta^{(k)}$ be any feasible solutions in the subproblem $OT_{\mathrm{time}}$ and $OT_{\mathrm{space}}$ attaining (or approximating) $Obj(\bm{q}^{(k)})$ for $k=1,2$. Consider the convex combination $\nu = \lambda \nu^{(1)}+ (1-\lambda)\nu^{(2)}$ and $\theta = \lambda \theta^{(1)}+ (1-\lambda)\theta^{(2)}$. It is straightforward that this pair of  $\nu$ and $\theta$ is feasible for $\bm{q}^{(\lambda)}$ since all constraints are linear respect to $\nu$, $\theta$ and $\bm{q}$. Moreover, the linearity of $Obj_{\mathrm{time}}$ and $Obj_{\mathrm{space}}$ in terms of $\nu$ and $\theta$ implies that $Obj(\bm{q}^{(\lambda)})\geq \sum_{i=1}^{n}\int_{\mathcal{A}}r q^{(\lambda)}_i(\alpha)\mathrm{d}\alpha-Obj_{\mathrm{time}}(\nu)-Obj_{\mathrm{space}}(\theta)= \lambda Obj(\bm{q}^{(1)})+(1-\lambda)Obj(\bm{q}^{(2)})$. Therefore $Obj(\bm{q})$ is concave. \hfill\qedsymbol

\subsection{Proofs in \Cref{sec:opt_structure}}
\subsubsection{Proof of \Cref{lemma:null_space}.}
Define $\mathcal{B}_{i,j} = \left\{\bm x: \delta(\bm{x},\bm{y}_i) -\delta(\bm{x},\bm{y}_j) = \omega_j-\omega_i\right\}$ and $\mathcal{B}_{i,0} = \left\{\bm x: \delta(\bm{x},\bm{y}_i)+\omega_i = \omega_0 \right\}$. It is sufficient to show that all $\mathcal{B}_{i,j}$ with $i\in[n]$ and $j\in[n]\cup\{0\}$ have Lebesgue measure $0$.

The statement for $\mathcal{B}_{i,0}$ is straightforward. Since we assume $\delta(\bm{x},\bm{y}_i) = \hat{\delta}(\|\bm{x}-\bm{y}_i\|_2)$ and $ \hat{\delta}$ is strictly increasing, $\mathcal{B}_{i,0}$ is either empty or the sphere of a ball and hence has Lebesgue measure $0$. For  $\mathcal{B}_{i,j}$ with $i,j\in[n]$, without loss of generality, assume that $\bm y_i=\bm 0$ and $\|\bm y_i-\bm y_j\|_2=1$. Hence $\bm y_j$ is the unit vector from $\bm{y}_i$ to $\bm{y}_j$. Then we have $\bm{x} = \xi \bm{y}_j +\bm{z}$ with $\bm{z}\perp\bm{y}_j$ and $\xi\in\mathbb{R}$. For fixed $\bm{z}$, $\bm{x}\in\mathcal{B}_{i,j}$ if and only if 
\begin{equation}\label{eq:proof_lemma_1}
q_{\bm{z}}(\xi) = \hat\delta\left(\sqrt{\|\bm z\|_2^2 +\xi^2}\right) - \hat\delta\left(\sqrt{\|\bm z\|_2^2 +\left(\xi-1\right)^2}\right) = \omega_j-\omega_i.
\end{equation}

One can verify that $q_{\bm{z}}(\xi)$ is (strictly) monotone in $\xi$. To see this, Define $h_z(\xi):=\hat\delta\!\left(\sqrt{\|z\|_2^2+\xi^2}\right)$. Then
$q_z(\xi)=h_z(\xi)-h_z(\xi-1)$.
When $z\neq 0$, since $\hat\delta$ is convex and strictly increasing and $\sqrt{\|z\|_2^2+\xi^2}$ is convex, $h_z$ is strictly convex in $\xi$.
Hence $h_z'$ is strictly increasing, and therefore
$q_z'(\xi)=h_z'(\xi)-h_z'(\xi-1)>0$ for all $\xi$. Thus $q_z$ is strictly increasing in $\xi$. Hence, for each fixed $\bm{z}$ (a $d-1$ dimensional orthogonal slice), the equation has at most one solution $\xi$. Then by Fubini or Tonelli, $\mathcal{B}_{i,j}$ has Lebesgue measure $0$.\hfill\qedsymbol

\subsubsection{Proof of \Cref{prop:Laguerre_cell}.}
Consider the subproblem $OT_{\mathrm{space}}$ in \eqref{eq:STB} given $\bm q$. Note that this problem is separable in $\alpha$; one may solve it conditioning on each $\alpha$. To indicate that $\alpha$ is fixed, we use $\theta_i^\alpha$ to denote $\theta_i(\cdot,\alpha)$, that is 
\begin{equation}\label{eq:Thm2_A3}
\begin{aligned}
\inf_{\bm\theta^\alpha}\quad& \sum_{i=1}^{n}\int_{\mathcal{X}} \delta(\bm{x},\bm{y}_i)\theta_{i}^\alpha(\bm{x})\mathrm{d}\bm{x}\\
\text{s.t.}\quad& \sum_{i=1}^{n}\theta^\alpha_{i}(\bm{x})\leq \mu(\bm{x},\alpha),\quad\forall \bm{x},\\
&\int_{\mathcal{X}}\theta_{i}^\alpha(\bm{x})\mathrm{d}\bm{x} = q_i(\alpha), \quad\forall i.
\end{aligned}
\end{equation}

By introducing a dummy station (indexed as $0$) with capacity $\int_{\mathcal{X}}\mu(\bm{x}, \alpha) \mathrm{d}\bm{x} -\sum_{i=1}^{n}q_i(\alpha)$ and defining $\delta(\bm{x},\bm{y}_0)=0$ for all $\bm{x}\in\mathcal{X}$. The problem then can be further reformulated as a standard optimal transportation problem
\begin{equation}\label{eq:OT_space_dummy}
\begin{aligned}
\inf_{\theta^{\alpha}}&\quad  \sum_{i=1}^{n}
\int_{\mathcal{X}}\delta(\bm{x },\bm{y}_i)\theta^{\alpha}_{i}(\bm{x})\mathrm{d}\bm{x},\\
\text{s.t.} 
&\int_{\mathcal{X}}\theta^{\alpha}_i(\bm{x}) \mathrm{d}\bm{x} = q_i(\alpha),\quad \forall i\in[n],\\
&\int_{\mathcal{X}}\theta^{\alpha}_0(\bm{x}) \mathrm{d}\bm{x} = \int_{\mathcal{X}}\mu(\bm{x}, \alpha) \mathrm{d}\bm{x} -\sum_{i=1}^{n}q_i(\alpha),\\
&\sum_{i=0}^{n}\theta^{\alpha}_i(\bm{x}) = \mu(\bm{x}, \alpha),\\
\end{aligned}
\end{equation}
in which demand has marginal density $\mu(\cdot,\alpha)$ and each service station (including $0$) has Dirac mass $q_i(\alpha)$ for $i\in[n]$ and $\int_{\mathcal{X}}\mu(\bm{x}, \alpha) \mathrm{d}\bm{x} -\sum_{i=1}^{n}q_i(\alpha)$ for $i=0$.

The duality of this problem is 
\begin{equation*}
\begin{aligned}
\sup_{\phi,\varphi}&\quad 
\int_{\mathcal{X}}\phi(\bm{x })\mu(\bm{x},\alpha)\mathrm{d}\bm{x}+
\sum_{i=0}^{n}\varphi_i q_i(\alpha)\\
\text{s.t.} 
&\quad 
\phi(\bm{x})+\varphi_i\leq \delta(\bm{x},\bm{y}_i),\quad \forall i\in\{0\}\cup[n],\\
\end{aligned}
\end{equation*}
and strong duality holds by the Kantorovich duality theorem. In this dual formula, $\phi(\bm{x})$ will take value 
$\min\left\{\delta(\bm{x},\bm{y}_i)-\varphi_i, i =0,1,\hdots,n\right\}$ in order to maximize the objective. Two observations complete our argument. First, the optimality condition of optimal transport requires the equality $\phi(\bm{x})+\varphi_i=\delta(\bm{x},\bm{y}_i)$ to hold almost everywhere when optimal $\rho^{\alpha}_i$ is positive (c.f. Theorem 5.10 in \citealtec{villani2008optimal}). Second, the equality can only hold when $\bm{x}\in\mathcal{C}^{\varphi_0}_{i}(\bm{\varphi})$ and since the boundaries between different cells are negligible, we must have 
$\theta^{\alpha}_i(\bm{x}) = \mu(\bm{x},\alpha) \mathbbm{1}\left(\bm{x}\in\mathcal{C}^{\varphi_0}_{i}(\bm{\varphi})\right)$ in order to satisfying the marginal constraint. Integrating such result for all $\alpha$ proves \Cref{prop:Laguerre_cell}.\hfill\qedsymbol

\subsubsection{Proof of \Cref{prop:sensitive_priority}.}



Consider the formulation \eqref{eq:STB}. We first focus on the subproblem $OT_{\mathrm{time}}$. Note that, $OT_{\mathrm{time}}$ decouples across the index $i$, $\mathcal{T}^-$ and $\mathcal{T}^+$, that is:
\begin{equation*}
\begin{aligned}
OT_{\mathrm{time}}(\bm{c},\bm{q}) = \inf_{\bm\nu^+,\bm\nu^-,\bm{q}^+,\bm{q}^-}\quad& \sum_{i=1}^{n}\int_{\mathcal{T}^+\times\mathcal{A}}\ell(t,\alpha)\nu^+_{i}(t,\alpha)\mathrm{d}(t,\alpha)+\int_{\mathcal{T}^-\times\mathcal{A}}\ell(t,\alpha)\nu^-_{i}(t,\alpha)\mathrm{d}(t,\alpha),\\
\text{s.t.}\quad& \int_\mathcal{A}\nu^\pm_{i}(t,\alpha)\leq c_i(t),\quad\forall i\in[n],t\in\mathcal{T}^\pm,\\
&\int_{\mathcal{T}^\pm}\nu^\pm_{i}(t,\alpha)\mathrm{d}t = q_i^\pm(\alpha), \quad\forall i\in[n],\alpha\in\mathcal{A},\\
&q_i^+(\alpha)+q_i^-(\alpha) = q_i(\alpha),\quad\forall i\in[n],\alpha\in\mathcal{A}.
\end{aligned}    
\end{equation*}

We assume that $q_i^+(\alpha)$ and $q_i^-(\alpha)$ are given. Taking $i$ and $\mathcal{T}^+$ as an example, we may consider the following subproblem:
\begin{equation}\label{eq:Thm2_proof}
\begin{aligned}
\inf_{\bm\nu^+_i}&\quad  
\int_{\mathcal{A}\times\mathcal{T}^+}\ell(t,\alpha)\nu^+_{i}(\alpha,t)\mathrm{d}(\alpha,t) ,\\
\text{s.t.}&\quad 
\int_{\mathcal{T}^+}\nu^+_i(\alpha, t)\mathrm{d}t= q_i^+(\alpha), \quad \forall \alpha\in\mathcal{A},\\
&\quad 
\int_{\mathcal{A}}\nu^+_i(\alpha, t)\mathrm{d}\alpha\leq c_i(t), \quad \forall t\in\mathcal{T}^+.
\end{aligned}
\end{equation}

In problem \eqref{eq:Thm2_proof}, since $\ell$ is increasing with respect to $t$, which means that any capacity idle is not optimal. Let $t_{q_i^+}$ satisfying $\int_{[0,t_{q_i^+}]}c_i(t)\mathrm{d}t = \int_{\mathcal{A}}q_i^+(\alpha)\mathrm{d}\alpha$. We can then restrict time scheduling on the interval $[0,t_{q_i^+}]$. 

Moreover, the temporal cost function $\ell$ is exactly super-modular on $\mathcal{T}^+$. To see this, let $\alpha_1<\alpha_2$ and $t_1<t_2$. The relation between partial derivative of $\ell(t,\alpha_1)$ and $\ell(t,\alpha_2)$ requires that

\begin{equation*}
\ell(t_2,\alpha_2)-\ell(t_1,\alpha_2) =\int_{t_1}^{t_2} \frac{\partial \ell(t,\alpha_2)}{\partial t} \geq \int_{t_1}^{t_2} \frac{\partial \ell(t,\alpha_1)}{\partial t}  \geq \ell(t_2, \alpha_1)-\ell(t_1, \alpha_1).
\end{equation*}
which is exactly the definition of super-modularity. Given the super-modularity of $\ell$, it is known that a counter-monotonic copula will minimize \eqref{eq:Thm2_proof} (see, e.g., \citealtec{tchen1980inequalities}). Since $t$ is continuous, counter-monotonicity guarantees the sensitivity prioritized property. The same argument applies to any other $i$ and $\mathcal{T}^-$, hence we complete the proof.
\hfill\qedsymbol



\subsubsection{Proof of \Cref{prop:exclusion}.}
This result is based on the monotonicity of $\ell$ with respect to $\alpha$. Let $\pi^*$ be an optimal plan in \eqref{eq:STM} and $\phi^*$, $\varphi^*$ be the corresponding Kantorovich potentials in \eqref{eq:STM_D}. According to the marginal property of $\pi^*$ (c.f., \citealp[Theorem 5.10]{villani2008optimal}), we have
\begin{equation*}
\phi^*(\bm{x},\alpha) + \varphi^*_i(t) = r-\delta(\bm{x},\bm{y}_i) - \ell(t,\alpha)
\end{equation*}
$\pi^*$-almost surely. To seek contradiction, assume that $\mu_{\alpha_1}(\bm x),\mu_{\alpha_2}(\bm x)>0$,  $(\bm{x},\alpha_2,t_2,i_2)\in\operatorname{supp}(\pi^*)$ while $(\bm{x},\alpha_1,t,i)\notin\operatorname{supp}(\pi^*)$ for all $t$ and $i$ with $\alpha_2>\alpha_1$. We have
\begin{equation*}
\phi^*(\bm{x},\alpha_2) + \varphi^*_{i_2}(t_2) = r-\delta(\bm{x},\bm{y}_{i_2}) - \ell(t_2,\alpha_2) < r-\delta(\bm{x},\bm{y}_{i_2}) - \ell(t_2,\alpha_1)\leq \phi^*(\bm{x},\alpha_1) + \varphi^*_{i_2}(t_2),
\end{equation*}
which implies $\phi^*(\bm{x},\alpha_2)< \phi^*(\bm{x},\alpha_1)$. Recall the proof of \Cref{prop:general_Dual}, when $(\bm{x},\alpha_1)$ is not served (matched to a dummy time), we have $\phi^*(\bm{x},\alpha_1) = 0$, then we reached a contradiction of $\phi^*(\bm{x},\alpha_2)< 0$. \hfill\qedsymbol

\subsubsection{Proof of \Cref{prop:order_preserve}.}
This proposition is derived from the optimality of the co-monotone transport plan for a one-dimensional OT problem with convex cost (see \citealp[Definition 2.3 and Theorem 2.9]{santambrogio2015optimal}). Specifically, note that in $OT_{\mathrm{time}}$, the problem is separable in $i$; one may solve it conditioning on each station $i$. We introduce an auxiliary decision variable $\hat{c}_i(t)$ to represent the occupied processing capacity of station $i$. Under \Cref{assm:Hom_Sens_Het_Pref}, since the type captures only time preferences, we may set $\tau_\alpha = \alpha$ and $\mathcal{A} = \mathcal{T}$. The reformulated problem becomes:
\begin{equation*}
\begin{aligned}
OT_{\mathrm{time}}(\bm{c},\bm{q}) = \inf_{\hat{c}_i} \inf_{\nu_i}\quad& \int_{\mathcal{T}}\hat{\ell}(t-\alpha)\nu_{i}(t,\alpha)\mathrm{d}t,\\
\text{s.t.}\quad& \int_\mathcal{A}\nu_{i}(t,\alpha)= \hat{c}_i(t),\quad\forall t\in\mathcal{T},\\
&\int_{\mathcal{T}}\nu_{i}(t,\alpha)\mathrm{d}t = q_i(\alpha), \quad\forall \alpha\in\mathcal{T},\\
& \hat{c}_i(t)\leq c_i(t),\quad\forall t\in\mathcal{T}.\\
\end{aligned}    
\end{equation*}
The inner problem is a standard one-dimensional OT problem from $\mathcal{T}$ to $\mathcal{T}$. Therefore, Theorem 2.9 from \citeec{santambrogio2015optimal} applies, which completes the proof. \hfill\qedsymbol

\subsection{Proof of \Cref{Thm:finite_type}}\label{proof:finite_type}
The proof proceeds as follows. We begin by fixing $\bm{q}$ and formulating the duals of $OT_{\text{time}}$ and $OT_{\text{space}}$, both of which correspond to the dual of a semi-discrete partial OT problem. In the second step, we apply the minimax theorem to interchange the inner and outer optimization variables. This renders the inner problem with respect to $\bm{q}$ a linear program, for which we subsequently derive its dual. As a result, the overall problem becomes a minimization problem. Finally, we simplify the formulation by eliminating certain decision variables, ultimately arriving at \eqref{eq:STB_D}.
~\begin{paragraph}{Step 1: Dual reformulation of $OT_{\mathrm{time}}$ and $OT_{\mathrm{space}}$.}
We first reformulate the subproblems in their dual form, following an analogous approach to that used in the proof of \Cref{prop:general_Dual} and \Cref{prop:Laguerre_cell}. Since the subproblem $OT_{\mathrm{space}}$ is separable with respect to $j$, we fix a $j\in[m]$.  By introducing a dummy station (indexed as $0$) with capacity $\sum_{j=1}^{m}\int_{\mathcal{X}}\mu_j(\bm{x}) \mathrm{d}\bm{x} -\sum_{i=1}^{n}\sum_{j=1}^{m}q_{i,j}$ and defining $\delta(\bm{x},\bm{y}_0)=0$ for all $\bm{x}\in\mathcal{X}$, and can be further reformulated as a standard optimal transportation problem, which is a discrete variation of \eqref{eq:OT_space_dummy}:
\begin{equation}\label{eq:OT_space_dummy_discrete}
\begin{aligned}
\inf_{\theta_{i,j}}&\quad  \sum_{i=0}^{n}
\int_{\mathcal{X}}\delta(\bm{x },\bm{y}_i)\theta_{i,j}(\bm{x})\mathrm{d}\bm{x},\\
\text{s.t.} 
&\int_{\mathcal{X}}\theta_{i,j}(\bm{x}) \mathrm{d}\bm{x} = q_{i,j},\quad \forall i\in[n],\\
&\int_{\mathcal{X}}\theta_{0,j}(\bm{x}) \mathrm{d}\bm{x} = \int_{\mathcal{X}}\mu_j(\bm{x}) \mathrm{d}\bm{x} -\sum_{i=1}^{n}q_{i,j},\\
&\sum_{i=0}^{n}\theta_{i,j}(\bm{x}) = \mu_j(\bm{x}).\\
\end{aligned}
\end{equation}

The Kantorovich duality of \eqref{eq:OT_space_dummy_discrete} is 
\begin{equation}
\label{eq:Dj}
\begin{aligned}
\sup_{f_j,\{\eta_{i,j}\}_{i=1}^n}\;& 
\int_{\mathcal X} f_j(\bm x)\mu_j(\bm x)\mathrm d\bm x + \sum_{i=1}^n \eta_{i,j} q_{i,j} +\eta_{0,j}(\int_{\mathcal{X}}\mu_j(\bm{x}) \mathrm{d}\bm{x} -\sum_{i=1}^{n}q_{i,j}) \\
\text{s.t.}\;\;
& f_j(\bm x) + \eta_{i,j} \le \delta(\bm x,\bm y_i), 
   \qquad \forall i\in[n],\bm x\in\mathcal X,\\
& f_j(\bm x) + \eta_{0,j} \le 0, 
   \qquad \forall \bm x\in\mathcal X.
\end{aligned}
\end{equation}
Rearrange the objective, it becomes
\begin{equation*}
\int_{\mathcal X} \left(f_j(\bm x)+\eta_{0,j}\right)\mu_j(\bm x)\mathrm d\bm x + \sum_{i=1}^n \left(\eta_{i,j}-\eta_{0,j}\right) q_{i,j}. 
\end{equation*}
Adding a constant $C$ to all $\eta_{i,j}$ and subtracting $C$ from $f_j$ leaves the objective function unchanged. Therefore, without loss of generality, we may set $\eta_{0,j}=0$. Given the values $\{\eta_{i,j}\}_{i=1}^n$, the constraints on $f_j$ are summarized as $-f_{j}(\bm x)\geq \left(\max_{i\in[n]}[\eta_{i,j}-\delta(\bm x,\bm y_i)]\right)^+$, for all $\bm{x}$. Since the objective increases with respect to $f_{j}$, this constraint will be tight. Thus, we can directly set $f_{j}(\bm x) = -\left(\max_{i\in[n]}[\eta_{i,j}-\delta(\bm x,\bm y_i)]\right)^+$, which leads to the following formulation:
\begin{equation}\label{eq:OT_space_dual}
\begin{aligned}
OT_{\mathrm{space}}(\bm{q},\mu)=\sup_{\zeta_{i,j}}\quad &
\sum_{j=1}^m \left(
   \sum_{i=1}^{n}\zeta_{i,j}q_{i,j}
   - \int_{\mathcal X} \mu_j(\bm{x})
      \Big(\ \max_{i\in[n]}\big[\zeta_{i,j}-\delta(\bm{x},\bm{y}_i)\big]\ \Big)^+ \mathrm{d}\bm{x}
\right).
\end{aligned}
\end{equation}
By an analogous argument, we also have
\begin{equation}\label{eq:OT_time_dual}
\begin{aligned}
OT_{\mathrm{time}}(\bm{c},\bm{q})=\sup_{\eta_{i,j}}\quad &
\sum_{i=1}^n \left(
   \sum_{j=1}^{m}\eta_{i,j}q_{i,j}
   - \int_{\mathcal T} c_i(t)
      \Big(\ \max_{j\in[m]}\big[\eta_{i,j}-\ell(t,j)\big]\ \Big)^+ \mathrm{d}t
\right).
\end{aligned}
\end{equation}

Substituting \eqref{eq:OT_space_dual} and \eqref{eq:OT_time_dual} back into \eqref{eq:STB}, the entire problem becomes
\begin{equation}\label{eq:proof_STB_D_Dual_1}
\begin{aligned}
\sup_{\bm{q}\geq0}\inf_{\bm{\zeta},\bm{\eta}}\quad &
\sum_{i=1}^n 
   \sum_{j=1}^{m}\left(r-\zeta_{i,j}-\eta_{i,j}\right) q_{i,j}
   +\sum_{i=1}^{n}\int_{\mathcal T} 
   \Big( \max_{j\in[m]}\big[\eta_{i,j}-\ell_j(t)\big]\ \Big)^+c_i(t) \mathrm{d}t\\ 
   &+\sum_{j=1}^m \int_{\mathcal X} 
      \Big( \max_{i\in[n]}\big[\zeta_{i,j}-\delta(\bm{x},\bm{y}_i)\big]\ \Big)^+\mu_j(\bm{x}) \mathrm{d}\bm{x},\\
\text{s.t.} \quad & \sum_{i=1}^{n}q_{i,j}\leq \varrho(j) \quad\forall j,\\
& \sum_{j=1}^{m} q_{i,j} \leq \int_{\mathcal{T}}c_{i}(t)\mathrm{d}t\quad\forall i.
\end{aligned}
\end{equation}
\end{paragraph}

~\begin{paragraph}{Step 2: Swap $\inf$ and $\sup$.}
Note that objective consists a bilinear term of $\bm{q}$ and $(\bm{\zeta},\bm{\eta})$ while the remaining components are continuous and convex in $\bm{\zeta}$ and $\bm{\eta}$ given $\bm{q}$. Since the feasible region of $\bm{q}$ is compact, we can swap the $\inf$ and $\sup$ by Sion's minimax theorem. Then the inner problem is nothing but a linear programming of $\bm{q}$ :
\begin{equation}\label{eq:proof_STB_D_inner_LP}
\begin{aligned}
\sup_{\bm{q}\geq0}\quad &
\sum_{i=1}^n 
   \sum_{j=1}^{m}\left(r-\zeta_{i,j}-\eta_{i,j}\right) q_{i,j}\\
\text{s.t.} \quad & \sum_{i=1}^{n}q_{i,j}\leq \varrho(j) \quad\forall j,\\
& \sum_{j=1}^{m} q_{i,j} \leq \int_{\mathcal{T}}c_{i}(t)\mathrm{d}t\quad\forall i.
\end{aligned}
\end{equation}

Substituting the dual formulation of \eqref{eq:proof_STB_D_inner_LP} back to \eqref{eq:proof_STB_D_Dual_1}, we get the following reformulation.
\begin{equation}\label{eq:proof_STB_D_Dual_2}
\begin{aligned}
\inf_{\bm{\kappa}\ge 0,\bm{\lambda}\ge 0,\bm{\zeta},\bm{\eta}}\quad& 
     \sum_{i=1}^{n}\int_{\mathcal T} 
       \left[\Big( \max_{j\in[m]}\big[\eta_{i,j}-\ell_j(t)\big]\Big)^+ + \lambda_i\right]c_i(t)\mathrm{d}t \\ 
    &  + \sum_{j=1}^{m}\int_{\mathcal X} 
       \left[\Big( \max_{i\in[n]}\big[\zeta_{i,j}-\delta(\bm{x},\bm{y}_i)\big]\Big)^+ + \kappa_j\right]\mu_j(\bm{x})\mathrm{d}\bm{x} \\
\text{s.t.}\quad 
    & \kappa_j + \lambda_i+\zeta_{i,j}+ \eta_{i,j} \ge r, \quad \forall i\in[n], j\in[m],
\end{aligned}
\end{equation}
where $\kappa$ and $\lambda$ are the dual variables of \eqref{eq:proof_STB_D_inner_LP}. Recall that $\varrho(j) = \int_{\mathcal{X}} \mu_j(\bm{x})\mathrm{d}\bm{x}$, so $\kappa_j$ can be taken inside the integral.
\end{paragraph}
~\begin{paragraph}{Step 3: Eliminate $\bm\kappa$, $\bm\lambda$ and $\bm\zeta$.}
Note that given any feasible $(\bm\kappa, \bm\lambda,\bm\eta,\bm\zeta)$, we can always let $\hat{\kappa}_j=\hat{\lambda}_i=0$, $\hat{\eta}_{i,j}=\eta_{i,j}+\lambda_i$, and $\hat{\zeta}_{i,j}=r-\hat\eta_{i,j}=r-\eta_{i,j}-\lambda_i$. Then $(\hat{\bm\kappa}, \hat{\bm\lambda}, \hat{\bm\eta}, \hat{\bm\zeta})$ is still feasible and induces a smaller objective since
\begin{equation*}
\begin{aligned}
\Big( \max_{j\in[m]}\big[\hat{\eta}_{i,j}-\ell_j(t)\big]\Big)^+ + \hat{\lambda}_i =& \Big( \max_{j\in[m]}\big[\eta_{i,j}+ \lambda_i-\ell_j(t)\big]\Big)^+ \\
\leq& \Big( \max_{j\in[m]}\big[\eta_{i,j}-\ell_j(t)\big]\Big)^+ + \lambda_i,    \\
\Big( \max_{i\in[n]}\big[\hat{\zeta}_{i,j}-\delta(\bm{x},\bm{y}_i)\big]\Big)^+ + \hat{\kappa}_j =& \Big( \max_{i\in[n]}\big[r-\eta_{i,j}-\lambda_i-\delta(\bm{x},\bm{y}_i)\big]\Big)^+ \\
\leq& \Big( \max_{i\in[n]}\big[\zeta_{i,j} + \kappa_j-\delta(\bm{x},\bm{y}_i)\big]\Big)^+ \\
\leq& \Big( \max_{i\in[n]}\big[\zeta_{i,j}-\delta(\bm{x},\bm{y}_i)\big]\Big)^+ + \kappa_j .
\end{aligned}
\end{equation*}
\end{paragraph}
The next-to-last inequality is due to the constraint in \eqref{eq:proof_STB_D_Dual_2}. And the last steps in both inequalities are by the nonnegativity of $\bm\lambda$ and $\bm\kappa$. By this reduction, the only decision variable to be optimized is $\hat{\bm\eta}$ since $\hat{\bm\zeta}$ is given by $r-\hat{\bm\eta}$. Therefore, the dual problem of \eqref{eq:STB} reduces to \eqref{eq:STB_D} when the demand type is finite. \hfill\qedsymbol

\subsection{Proof of \Cref{lemma:pricing_inequality}}
The inequality is straightforward by definition of $p_i^*$ in \eqref{eq:incentive_price}. If further $t\in\mathcal T_{i,j}(\bm\eta^*)$, then $\eta_{i,j}^*-\ell_j(t)\ge 0$ and it maximizes $\eta_{i,j}^*-\ell_j(t)$ across all $j$, so $p_i^*(t)=\eta_{i,j}^*-\ell_j(t)$ and equality holds. Conversely, $p_i^*(t)= \eta_{i,j}^*-\ell_j(t)$ implies that $j$ maximizes $\eta_{i,j}^*-\ell_j(t)$ with nonnegative value given $t$, i.e., $t\in\mathcal T_{i,j}(\bm\eta^*)$ a.e. under \Cref{assm:null_boundary}. \hfill\qedsymbol

\subsection{Proof of \Cref{prop:EF_implementation}}
As a direct result of \Cref{lemma:pricing_inequality}, for any type $j$, station $i$, and time $t$,
$
\ell_j(t)+p_i^*(t)\ge\eta_{i,j}^*,
$
with equality if and only if $t\in\mathcal T_{i,j}(\bm\eta^*)$ (up to a null set). Therefore, for any $(\bm x,j)$, the maximal utility available at station $i$ equals
\begin{equation}\label{eq:max_utility}
\max_{t\in\mathcal T} U_j(\bm{x},i,t,\bm{p}^*) = \max_{t\in\mathcal T}\Big\{r-\delta(\bm x,\bm y_i)-\ell_j(t)-p_i^*(t)\Big\}
=r-\delta(\bm x,\bm y_i)-\eta_{i,j}^*,    
\end{equation}
attained by any $t\in\mathcal T_{i,j}(\bm\eta^*)$. 

Consider a type $j$ agent located in $\mathcal{C}_{i,j}(\bm\eta^*)$ for $i\neq 0$. As mentioned in \Cref{sec:SP_partition}, this agent will be matched to station $i$ at $t\in\mathcal T_{i,j}(\bm\eta^*)$. By \Cref{def:GLC}, we have
\begin{equation*}
\delta(\bm x,\bm y_i)+\eta_{i,j}^*\leq r \text{ and } \delta(\bm x,\bm y_i)+\eta_{i,j}^*\leq  \delta(\bm x,\bm y_{\tilde{i}})+\eta_{\tilde{i},j}^*, \quad \forall \tilde{i}\in[n], 
\end{equation*}
which implies $i$ maximizes the right hand of \eqref{eq:max_utility} and the maximal value is nonnegative. Therefore, we have established envy-freeness and individual rationality for this agent, as no alternative matching yields a higher utility and the agent's utility remains nonnegative.

Similarly, consider a type-$j$ agent located in $\mathcal{C}_{0,j}(\bm\eta^*)$. This agent is not matched to any station, satisfying individual rationality. By \Cref{def:GLC}, the right-hand side of \eqref{eq:max_utility} is nonpositive for all $i \in [n]$, implying that no matching yields strictly positive utility for the agent. Therefore, no matching is optimal, and the outcome is envy-free. 
\hfill\qedsymbol


\subsection{Proof of \Cref{prop:finite_slots}}
This proposition is a result of \Cref{lemma:lower_envelope_DS} and \Cref{prop:EF_implementation}. We first show that $\mathcal{T}_{i,j}(\bm\eta^*)$ consists of at most $m$ intervals. We consider a slightly different set $\widehat{\mathcal{T}}_{i,j}(\bm\eta^*) $ defined by removing the requirement $\ell_j(t) -\eta^{*}_{i,j} \leq 0$: 
\begin{equation*}
\widehat{\mathcal{T}}_{i,j}(\bm\eta^*) =  \left\{t\in\mathcal{T}:\ell_j(t) -\eta^{*}_{i,j}\leq \ell_{\tilde{j}}(t) -\eta^{*}_{i,\tilde{j}},\quad\forall\tilde{j}\in[m]\right\}.
\end{equation*}
Then it is clear that $ \mathcal{T}_{i,j}(\bm\eta^*) = \widehat{\mathcal{T}}_{i,j}(\bm\eta^*)\cap \mathcal{T}^{0}_{i,j}(\bm\eta^*)$ with $\mathcal{T}^{0}_{i,j}(\bm\eta^*)\coloneqq  \left\{t\in\mathcal{T}:\ell_j(t) -\eta^{*}_{i,j} \leq 0\right\}$. Since $\ell_j(t)$ is a quasi-convex function, $\mathcal{T}^{0}_{i,j}(\bm\eta^*)$ is then a closed interval (including singleton or empty set). On the other hand, by \Cref{lemma:lower_envelope_DS}, $\widehat{\mathcal{T}}_{i,j}(\bm\eta^*)$ can be expressed as the union of at most $\hat{K}_{i,j} \leq \lambda_s(n)$ disjoint intervals. Consequently, the intersection of these intervals with $\mathcal{T}^{0}_{i,j}(\bm\eta^*)$ also consists of at most $K_{i,j} \leq \hat{K}_{i,j} \leq \lambda_s(m)$ disjoint intervals, since each interval either remains unchanged, becomes smaller, or disappears.

Moreover, under this pricing scheme, the conclusion of \Cref{lemma:pricing_inequality} still holds. This is because the net utility of a type-$j$ demand selecting a slot in $\mathcal{T}_{i,j}$ remains unchanged, and the utilities of other slots are sub-optimal, as they form a convex combination of the net utilities described in \Cref{prop:EF_implementation}. Hence, the result follows directly from the same reasoning as in \Cref{prop:EF_implementation}. \hfill\qedsymbol

\subsection{Proof of \Cref{prop:finite_linear}}
We start by analyzing $OT_{\mathrm{time}}$ in \eqref{eq:STB}. Given $q_{i,j}$, the quantity of type-$j$ demand assigned to station $i$, since the capacity is constant, type $j$ demand occupies total time $\frac{q_{i,j} }{c_{i}}$. By \Cref{prop:sensitive_priority}, demands with higher time sensitivity are served closer to $t=0$. Hence, we introduce a sequence of time points
$
t^-_m \leq\cdots\leq t^-_1\leq 0\leq t^+_1\leq\cdots\leq t^+_m
$
to characterize the optimal schedule. Specifically, for each $j=1,\dots,m$, demands of sensitivity $s_j$ are served in the intervals
$[t^+_{j-1},t^+_j]$ and $[t^-_j,t^-_{j-1}]$ (ignoring the boundary),
with the convention $t^+_0 = t^-_0 = 0$. Therefore, the cumulative temporal cost is

\begin{equation*}
\begin{aligned}
&\sum_{j=1}^{m}\frac{s_jc_i h\left(t^+_j-t^+_{j-1}\right)\left(t^+_j+t^+_{j-1}\right)}{2}+\sum_{j=1}^{m}\frac{s_jc_i b\left(t^-_{j-1}-t^-_j\right)\left|t^-_{j-1}+t^-_j\right|}{2}\\ 
=&\sum_{j=1}^{m}\frac{s_j hc_i\left((t^+_j)^2-(t^+_{j-1})^2\right)}{2}+\sum_{j=1}^{m}\frac{s_j bc_i\left((t^-_j)^2-(t^-_{j-1})^2\right)}{2}\\
=&\frac{c_i}{2}
\sum_{j=1}^{m-1} \left((t^+_j)^2 h+(t^-_j)^2 b\right)(s_j - s_{j+1}) 
+ \frac{s_m c_i}{2} \left((t^+_m)^2 h+ (t^-_m)^2 b\right). 
\end{aligned}
\end{equation*}

On the other hand, by Jensen’s inequality,
\begin{equation*}
\left((t^+_j)^2 h+(t^-_j)^2 b\right) \geq \frac{bh}{b+h}\left(t_j^+ - t_j^- \right)^2 = \frac{bh}{b+h}\left(\sum_{k=1}^{j}\frac{q_{i,j}}{c_i} \right)^2.
\end{equation*}

Equality holds if and only if 
$t^+_j=\frac{b\sum_{k=1}^{j}q_{i,j}}{c_i(b+h)} $ and $t^-_j=-\frac{h\sum_{k=1}^{j}q_{i,j}}{c_i(b+h)}$ for all $j$. The corresponding minimal temporal cost at station $i$ is

\begin{equation}\label{eq:OT_time_closed}
\sum_{j=1}^{m}\frac{ bh \left(\sum_{k=1}^{j}q_{i,j} \right)^2}{2c_i(b+h)}(s_j - s_{j+1}) =  \sum_{j=1}^{m}\frac{ \beta \left(\sum_{k=1}^{j}q_{i,j} \right)^2}{4c_i}(s_j - s_{j+1}), 
\end{equation}
where we let $s_{m+1}=0$ for convenience. This result yields the optimal schedule once $\{q_{i,j}\}$ are fixed, and hence solves $OT_{\mathrm{time}}$ in closed form in terms of $\bm{q}$.

Substituting the right hand sides of \eqref{eq:OT_time_closed} and \eqref{eq:OT_space_dual} (in \Cref{proof:finite_type}) into \eqref{eq:STB}, we obtain
\begin{equation}\label{eq:proof_STB_D_Dual_finite}
\begin{aligned}
\sup_{\bm{q}\geq 0}\inf_{\bm\zeta}\quad& \sum_{i=1}^{n}\sum_{j=1}^{m}(r-\zeta_{i,j})q_{i,j}-\sum_{j=1}^{m}\frac{ \beta \left(\sum_{k=1}^{j}q_{i,j} \right)^2}{4c_i}(s_j - s_{j+1})\\
&+ \sum_{j=1}^m 
    \int_{\mathcal X} \mu_j(\bm{x})
      \Big(\ \max_{i\in[n]}\big[\zeta_{i,j}-\delta(\bm{x},\bm{y}_i)\big]\ \Big)^+ \mathrm{d}\bm{x},\\
\text{s.t.}\quad& \sum_{i=1}^{n}q_{i,j}\leq \varrho(j) \quad\forall j\in[m],\\
& \sum_{j=1}^{m}q_{i,j}\leq \int_{\mathcal{T}}c_{i}(t)\mathrm{d}t\quad\forall i\in[n].
\end{aligned}
\end{equation}
This reformulation is similar to \Cref{eq:proof_STB_D_Dual_1} and the remainder of the proof is similar. However, instead of first taking duality and then eliminating the dual variables. We first show here that the constraint in \eqref{eq:proof_STB_D_Dual_finite} is redundant. We first find that the second constraint is redundant due to the second condition in \Cref{assm:long-horizon-linear-cost}. After removing the second constraint, the feasible region of $\bm{q}$ is still closed, so we can apply Sion's minimax theorem and exchange the $\sup$ and $\inf$. We denote the resulting problem as 
\begin{equation}\label{eq:proof_STB_D_Dual_finite_inner_dual}
\begin{aligned}
\inf_{\bm\zeta}\sup_{\bm{q}\geq 0}\quad& \sum_{i=1}^{n}\sum_{j=1}^{m}(r-\zeta_{i,j})q_{i,j}-\sum_{j=1}^{m}\frac{\beta}{4 c_i }\bm{q}_{i}^\top \bm{A}\bm{q}_{i} + \sum_{j=1}^m 
    \int_{\mathcal X} \mu_j(\bm{x})
      \Big(\ \max_{j\in[m]}\big[\zeta_{i,j}-\delta(\bm{x},\bm{y}_i)\big]\ \Big)^+ \mathrm{d}\bm{x}\\
\text{s.t.}\quad& \sum_{i=1}^{n}q_{i,j}\leq \varrho(j) \quad\forall j\in[m],
\end{aligned}
\end{equation}
where we adopt the definition of matrix $\bm A$ and denote the vector $(q_{i,j})_{j=1}^{m}$ as $\bm{q}_{i}$. We denote \eqref{eq:proof_STB_D_Dual_finite_inner_dual} as $\inf_{\bm\zeta}F(\bm{\zeta}) $ for short, where $F(\bm{\zeta})$ is defined as the optimal objective of the inner problem. Let us consider the relaxed version of \eqref{eq:proof_STB_D_Dual_finite_inner_dual} by dropping all constraints on $\bm{q}$. Then the optimal unconstrained $\bm{q}$ is given by 
\begin{equation*}
\bm{q}_{i} = \frac{2c_i}{\beta}\bm{A}^{-1} (r-\bm{\zeta}_i).
\end{equation*}
Therefore we have
\begin{equation}\label{eq:F<G}
F(\bm{\zeta}) \leq   \sum_{i=1}^n\frac{c_i}{\beta}(r-\bm{\zeta}_i)^\top \bm{A}^{-1}(r-\bm{\zeta}_i)  + \sum_{j=1}^m 
    \int_{\mathcal X} \mu_j(\bm{x})
      \Big(\ \max_{i\in[n]}\big[\zeta_{i,j}-\delta(\bm{x},\bm{y}_i)\big]\ \Big)^+ \mathrm{d}\bm{x} \coloneqq H(r-\bm\zeta) . 
\end{equation}
Note that by letting $\bm\eta = r-\bm\zeta$, $H(\bm\eta)$ is exactly the objective of \eqref{eq:finite_linear_Dual}. 

Moreover, $H$ is strictly convex because the $\bm{A}^{-1}$ is an M-matrix and the integral term is convex. Therefore it attains its minimum with first order condition holds. That is, the $\bm\eta^*$ that minimize $H(\bm\eta)$ satisfies
\begin{equation}\label{eq:finite_firstorder}
\int_{\mathcal{C}_{i,j}(\bm\eta^*)}\mathrm{d}\mu_{j}(\bm{x}) = \left\{
\begin{aligned}
&\frac{2c_i}{\beta} \frac{\eta^{*}_{i,1}-\eta^{*}_{i,2}}{s_1-s_2} & j =1, \\
&\frac{2c_i}{\beta} \left(\frac{\eta^{*}_{i,j}-\eta^{*}_{i,j+1}}{s_{j}-s_{j+1}}-\frac{\eta^{*}_{i,j-1}-\eta^{*}_{i,j}}{s_{j-1}-s_{j}}\right) & j = 2, \hdots, m-1, \\
&\frac{2c_i}{\beta}\left(\frac{\eta^{*}_{i,m}}{s_{m}}-\frac{\eta^{*}_{i,m-1}-\eta^{*}_{i,m}}{s_{m-1}-s_{m}} \right)& j =m,
\end{aligned}\right.
\end{equation}
where $\mathcal{C}_{i,j}(\bm\eta^*)$ follows the same definition in \Cref{sec:SP_partition}. Choose $q^*_{i,j} = \int_{\mathcal{C}_{i,j}(\bm\eta^*)}\mathrm{d}\mu_{j}(\bm{x})$. This $\bm q^*$ satisfies the constraint in \eqref{eq:proof_STB_D_Dual_finite_inner_dual} and the objective of \eqref{eq:proof_STB_D_Dual_finite_inner_dual} equals $H(\bm\eta^*)$ when $\bm\zeta=\bm\zeta^* = r-\bm\eta^*$. Denote the objective of \eqref{eq:proof_STB_D_Dual_finite_inner_dual} as $G(\bm\zeta,\bm q)$ we have 
\begin{equation*}
H(\bm\eta^*) = G(\bm\zeta^*,\bm q^*)\leq F(\bm\zeta^*)\leq H(\bm\eta^*), 
\end{equation*}
where the first inequality is due to the definition of $F$ and the second inequality is by \eqref{eq:F<G}. Therefore, we have all equality holds in above inequalities chain. This means that given $\bm\zeta^*$, $\bm q^*$
 optimize $G(\bm\zeta^*,\bm q)$. On the other hand, given $\bm q^*$,
 \begin{equation*}
G(\bm\zeta,\bm q^*) = \sum_{i=1}^n\frac{c_i}{\beta}(r+\bm{\zeta}^*_i-2\bm{\zeta}_i)^\top \bm{A}^{-1}(r-\bm{\zeta}^*_i)  + \sum_{j=1}^m 
    \int_{\mathcal X} \mu_j(\bm{x})
      \Big(\ \max_{j\in[m]}\big[\zeta_{i,j}-\delta(\bm{x},\bm{y}_i)\big]\ \Big)^+ \mathrm{d}\bm{x}, 
 \end{equation*}
 which is convex with respect to $\bm\zeta$. When $\bm{\zeta} = \bm\zeta^*$, the first order condition holds by \eqref{eq:finite_firstorder}, hence $\bm\zeta^*$ optimize $G(\bm\zeta,\bm q^*) $. So $(\bm\zeta^*,\bm q^*)$ is the saddle point and therefore optimal. Therefore, we complete the proof, showing that by optimizing $H(\eta)$, or equivalently \eqref{eq:finite_linear_Dual}, we can recover \eqref{eq:proof_STB_D_Dual_finite} and further solve \eqref{eq:STB}. \hfill\qedsymbol

\subsection{Proof of \Cref{Prop:Opt_c}}
Using the same approach as in \Cref{sec:capacity_allocation}, one can write the capacity allocation problem based on \eqref{eq:finite_linear_Dual} with $m=1$ as follows.
\begin{equation*}
\begin{aligned}
\max_{\bm{c}\geq\bm0}\min_{\bm{\eta}}
&\sum_{i=1}^n\frac{c_i\eta_i^2}{s\beta}  -\int_{\mathcal{X}}\min_{i\in[n]}\left\{\delta(\bm{x},\bm{y}_i)+\eta_i\right\}\mu(\bm{x})\mathrm{d}\bm{x}\\
\text{s.t.} &\sum_{i=1}^{n} c_i \leq B.&
\end{aligned}
\end{equation*}
Note that, since we have the assumption that $r$ is sufficiently large, we have already removed the total service reward (which is a constant) in the above problem.

Next, we directly show that 
\begin{equation*}
c^*_i = \frac{\int_{\mathcal{C}^{+\infty}_i(\bm{0})}\mu(\bm{x})\mathrm{d}\bm{x}}{\int_{\mathcal{X}}\mu(\bm{x})\mathrm{d}\bm{x}}B ,\quad  \eta^*_i =   \frac{s\beta\int_{\mathcal{X}}\mu(\bm{x})\mathrm{d}\bm{x}}{2B}
\end{equation*}
is a saddle point of this problem. Indeed, given $\bm{c}^*$, it is not difficult to verify that $\bm{\eta}^*$ satisfies the first order condition and hence it is optimal. On the other hand, since all $\eta_i^*$ are the same, when $\bm{\eta}=\bm{\eta}^*$, any feasible $\bm{c}$ that exhausts the budget, including $\bm{c}^*$, is optimal. Therefore, the proof is completed. \hfill\qedsymbol





\subsection{Proof of \Cref{prop:Hotelling_Homogeneous}}
According to \Cref{prop:Laguerre_cell}, we only need to decide two values $x_1$ and $ x_2$ satisfying $x_1+x_2\leq 1$ such that all demand located at location $[0,x_1)$ will be transported to station $1$ and all demand on the interval $(1-x_2,1]$ will be transported to station $2$. Therefore, the problem can be simplified as the following problem with only $2$ decision variables
\begin{equation}
\begin{aligned}
 \max_{x_1,x_2>=0} &\quad (x_1+x_2)r -\frac{x_1^2}{2} -\frac{x_2^2}{2}-\frac{w x_1^2}{2 c_1}-\frac{w x_2^2}{2 c_2}   \\
  \text{s.t.} &\quad x_1+x_2\leq 1. 
\end{aligned}
\end{equation}

The optimal $x_1$ and $x_2$ can then be readily determined in two cases.

\textbf{Case 1:} If $r<\frac{(w+c_1)(w+c_2)}{2 c_1 c_2 + (c_1 + c_2) w}$, the constraint is inactive.  The optimal solution is then given by $x_1=\frac{rc_1}{c_1+w}$,  $x_2=\frac{rc_2}{c_2+w}$.

\textbf{Case 2:} If $r\geq\frac{(w+c_1)(w+c_2)}{2 c_1 c_2 + (c_1 + c_2) w}$, the constraint becomes active. By substituting $x_2=1-x_1$, the optimal solution is $x_1=\frac{c_1c_2+c_1w}{2 c_1 c_2 + (c_1 + c_2) w}$. \hfill\qedsymbol

\subsection{Proof of \Cref{prop:Hotelling_Uniform}}\label{sec:proof_Hotelling_Uniform}
The computation of parameters $\hat{\alpha}$ is in \eqref{eq:solve_hat_alpha}. 
The partitioning function in \textbf{Case 2} can be find in \eqref{eq:tilde_f_part1} for $\alpha\geq \hat\alpha$ and \eqref{eq:tilde_f_part2} for $\alpha\leq\hat\alpha$, where $\kappa$ is computed in \eqref{eq:kappa}.
\begin{proof}{Proof of \Cref{prop:Hotelling_Uniform}}
According to \Cref{prop:Laguerre_cell}, we only need to decide two function $f_1(\alpha)$ and $f_2(\alpha)$ such that all demand with sensitivity $\alpha$ and between location $0$ and $f_1(\alpha)$ will be assigned to station $1$ and all demand with sensitivity $\alpha$ and between location $1-f_2(\alpha)$ and $1$ will be transported to station $2$. Therefore, the cumulated reward minus transportation cost for station $1$ is given by 
\begin{equation*}
r\int_{0}^{1}f_{1}(\alpha)\mathrm{d}\alpha - \int_{0}^{1}\frac{f_{1}(\alpha)^2}{2}\mathrm{d}\alpha .
\end{equation*}

Moreover, by \Cref{prop:sensitive_priority}, we adopt sensitivity-prioritized schedules. Then the temporal (waiting) cost is expressed as
\begin{equation*}
\int_{0}^{1}\alpha w f_1(\alpha)\frac{\int_{\alpha}^{1}f_1(t)\mathrm{d}t}{c_1}\mathrm{d}\alpha = \frac{w}{2c_1}\int_{0}^{1}\left(\int_{\alpha}^{1}f_1(t)\mathrm{d}t\right)^2\mathrm{d}\alpha.
\end{equation*}

Therefore, the total welfare of station $i$ can be expressed as $J_i(f_i)$:
\begin{equation*}
J_i(f_i) = \int_{0}^{1} \left(rf_i(\alpha)\ - \frac{f_i(\alpha)^2}{2}   -\frac{w}{2c_i}\left(\int_{\alpha}^{1}f_i(t)\mathrm{d}t\right)^2\right)\mathrm{d}\alpha . 
\end{equation*}

So the matching problem can be formulated as 
\begin{equation}\label{eq:hotel_control}
\begin{aligned}
\sup_{f_1, f_2}  &\quad J_1(f_1) +  J_2(f_2)\\
\text{s.t.} &\quad f_1(\alpha)+f_2(\alpha)\leq 1 ,\quad\forall \alpha\in[0,1], \\
&\quad f_1(\alpha),f_2(\alpha)\geq 0,\quad \forall \alpha\in[0,1]. 
\end{aligned}
\end{equation}

 We first investigate the subproblem for station $1$ with out any constraint, that is 
\begin{equation}
\begin{aligned}
\sup_{f_1} \quad J_1(f_1) 
\quad=\quad &\sup_{F_1} \quad r F(0) 
- \frac{1}{2}\int_{0}^{1} \left(\frac{w}{c_1}F_{1}(\alpha)^2+F'_1(\alpha)^2\right)\mathrm{d}\alpha\\
\end{aligned}
\end{equation}

where $F_1(\alpha) = \int_{\alpha}^{1}f_1(t)\mathrm{d}t$.

The Euler–Lagrange equation gives that 
\begin{equation}\label{Hotel_ODE}
-\frac{w}{c_1}F_1(\alpha) + F''(\alpha)=0. 
\end{equation}

Solving the ODE \eqref{Hotel_ODE} with condition $F(1)=0$ and $F'(0)=-f_1(0)=-r$ (all demand on $[0,r]$ with $0$ temporal cost will be assigned to the station $1$ since there is no temporal cost) gives that 
\begin{equation*}
f^*_1(\alpha) = r \frac{\cosh\left((1-\alpha)\sqrt{\frac{w}{c_1}}\right)}{\cosh\left(\sqrt{\frac{w}{c_1}}\right)} .
\end{equation*}
Since it satisfies the Euler–Lagrange equation and the boundary conditions, it is therefore the unique maximizer of the original variational problem. Similarly, the optimal $f_2$ without constraint is
\begin{equation*}
f^*_2(\alpha) = r \frac{\cosh\left((1-\alpha)\sqrt{\frac{w}{c_2}}\right)}{\cosh\left(\sqrt{\frac{w}{c_2}}\right)} .
\end{equation*}

\textbf{Case 1}
When $r\leq 1/2$, both $f^*_1$ or $f^*_2$ are less or equal to $1/2$, hence the constraints automatically hold. Therefore, $f_1=f^*_1$ and $f_2=f^*_2$ are optimal.

\textbf{Case 2}
Define $\lambda = \sqrt{\frac{w(c_1+c_2)}{2c_1c_2}}$. Consider 

\begin{equation}\label{eq:r_case2}
1/2<r< 
\frac{c_1^2+c_2^2}{(c_1+c_2)^2} +\frac{w}{2(c_1+c_2)} -\frac{(c_1-c_2)^2}{2(c_1+c_2)^2\cosh(\lambda)}.
\end{equation}

We will define the partition boundary with the help of two parameters $\kappa$ and $\hat{\alpha}$. The computation of these two parameters will be presented later. When $\alpha\geq \hat{\alpha}$, let 

\begin{equation}\label{eq:tilde_f_part1}
\begin{aligned}
\tilde{f}_1(\alpha ) =&
\kappa
\frac{\cosh\left(\sqrt{\frac{w}{c_1}}(1-\alpha)\right)}{\cosh\left(\sqrt{\frac{w}{c_1}}(1-\hat{\alpha})\right)},\quad
\tilde{f}_2(\alpha ) =& 
(1-\kappa)
\frac{\cosh\left(\sqrt{\frac{w}{c_2}}(1-\alpha)\right)}{\cosh\left(\sqrt{\frac{w}{c_2}}(1-\hat{\alpha})\right)}  . 
\end{aligned}
\end{equation}
 
When $\alpha< \hat{\alpha}$, we construct 

\begin{equation}\label{eq:tilde_f_part2}
\begin{aligned}
 \tilde{f}_1(\alpha ) =&  \frac{c_2-c_1}{2(c_1 + c_2)}\frac{\sinh(\lambda(\hat{\alpha}-\alpha))}{\sinh(\lambda\hat{\alpha})} + \left(\kappa-\frac{c_1}{c_1+c_2}\right)\frac{\sinh(\lambda\alpha)}{\sinh(\lambda\hat{\alpha})}+ \frac{c_1}{c_1+c_2} ,   \\
\tilde{f}_2(\alpha ) = &\frac{c_1-c_2}{2(c_1 + c_2)}\frac{\sinh(\lambda(\hat{\alpha}-\alpha))}{\sinh(\lambda\hat{\alpha})} + \left(\frac{c_1}{c1+c_2}-\kappa\right)\frac{\sinh(\lambda\alpha)}{\sinh(\lambda\hat{\alpha})}+ \frac{c_2}{c_1+c_2}   . 
\end{aligned}
\end{equation}

It is clear that, $\tilde{f}_1(\alpha)+\tilde{f}_2(\alpha)=1$ when $\alpha< \hat{\alpha}$. 



Next, we show how to determine $(\hat{\alpha},\kappa)$. 

$\hat{\alpha}$ is defined as solution of following equation with respect to $\alpha$:
%
\begin{equation}\label{eq:solve_hat_alpha}
\begin{aligned}
r =&
\frac{c_2}{c_1+c_2}+\frac{w\alpha^{2}}{2(c_1+c_2)}- 
\frac{w \alpha A_1(\alpha)}{(c_1+c_2)\lambda}+\frac{(c_1-c_2)w\alpha}{2c_1(c_1+c_2)\lambda\sinh(\lambda\alpha)}\\
&+\frac{\bigl(B_2(\alpha)(c_1 + c_2) + 2A_1(\alpha)c_1\lambda - \frac{c_1 - c_2}{\sinh(\lambda\alpha)}\lambda\bigr)\bigl(c_1 - c_2 + B_1(\alpha)(c_1 + c_2)\alpha + 2A_1(\alpha)c_2\alpha\lambda\bigr)}{(c_1 + c_2)^2(B_1(\alpha) + B_2(\alpha) + 2A_1(\alpha)\lambda)}, 
\end{aligned}
\end{equation}
where for saving the space, we use $A_1$, $B_1$ and $B_2$ to represent the following expressions.

\begin{equation*}
\begin{aligned}
&A_1(\alpha) = \frac{\cosh(\lambda\alpha)}{\sinh(\lambda\alpha)} ,\quad B_1(\alpha) =   \sqrt{\frac{w}{c_{1}}}
\tanh\Bigl(\sqrt{\tfrac{w}{c_{1}}}(1-\alpha)\Bigr)
,\quad B_2(\alpha) =   \sqrt{\frac{w}{c_{2}}}
\tanh\Bigl(\sqrt{\tfrac{w}{c_{2}}}(1-\alpha)\Bigr) .
\end{aligned}
\end{equation*}
When  $\alpha\rightarrow 0$, the right hand side of \eqref{eq:solve_hat_alpha} goes to  $1/2$ ; when $\alpha=1$ the right hand side of \eqref{eq:solve_hat_alpha} becomes $\frac{c_1^2+c_2^2}{(c_1+c_2)^2} +\frac{w}{2(c_1+c_2)} -\frac{(c_1-c_2)^2}{2(c_1+c_2)^2\cosh(\lambda)}$. So, due to our assumption of the range of $r$ in \eqref{eq:r_case2}, there always exists a solution $\hat{\alpha}$. 
After obtaining the $\hat{\alpha}$, $\kappa$ is computed by 
\begin{equation}\label{eq:kappa}
\kappa =  \frac{
2\dfrac{\lambda}{\sinh(\lambda\hat{\alpha})}
\Bigl(
\frac{c_{1}}{c_{1}+c_{2}}\cosh(\lambda\hat{\alpha})
-\frac{c_{1}-c_{2}}{2(c_{1}+c_{2})}
\Bigr)
+
B_2({\hat\alpha})
}
{
2\lambda A_1({\hat\alpha})
+
B_2({\hat\alpha})
+
B_1({\hat\alpha})
} . 
\end{equation}

Finally, we would like to prove the optimality. We first analyze the optimality conditions. Note that the problem \eqref{eq:hotel_control} is an optimal control problem, hence we will handle it with Pontryagin’s Maximum Principle (PMP) \citepec{sethi2021optimal}. 

Recall that $F_i(\alpha)=\int_\alpha^1 f_i(t)\mathrm{d}t$ is the \emph{tail state}.
With multiplier $\nu(\alpha)\ge 0$ for the mixed constraint $f_1+f_2\le 1$, define the Hamiltonian
$$
H(\alpha,S,\lambda;f,\nu)
=\sum_{i=1}^2\left(r f_i-\tfrac12 f_i^2-\tfrac{w}{2c_i}F_i(\alpha)^2+\lambda_i(-f_i)\right)
+\nu(\alpha)\big(1-f_1-f_2\big),
$$
where $\lambda_i$ are the costates. PMP requires the following conditions.

\paragraph{Maximization condition.}
For a.e.\ $\alpha\in[0,1]$, the optimal control $f(\alpha)=(f_1(\alpha),f_2(\alpha))$ maximizes
the Hamiltonian $H$.
Equivalently, the KKT conditions of this static maximization give multipliers
$\nu(\alpha)\ge 0$ and $\eta_i(\alpha)\ge 0$ such that
$$
\begin{aligned}
&
\frac{\partial H}{\partial f_i}
= r - f_i(\alpha) - \lambda_i(\alpha) - \nu(\alpha) - \eta_i(\alpha) =0,\qquad i=1,2,&\text{(stationarity)}\\
&f_i(\alpha)\ge 0,\ \ f_1(\alpha)+f_2(\alpha)\le 1, &\text{(primal feasibility)}\\
&\nu(\alpha)\ge 0,\ \ \eta_i(\alpha)\ge 0, &\text{(dual feasibility)}\\
&\nu(\alpha)\big(1-f_1(\alpha)-f_2(\alpha)\big)=0,\quad \eta_i(\alpha)f_i(\alpha)=0. &\text{(complementarity)}
\end{aligned}
$$
Eliminating $\eta_i(\alpha)$ yields the convenient inequality form
\begin{equation}\label{eq:ineq_KKT}
r - f_i(\alpha) - \lambda_i(\alpha) \le \nu(\alpha),
\qquad\text{with equality whenever } f_i(\alpha)>0,\quad i=1,2,
\end{equation}
together with $f_i\ge 0$, $f_1+f_2\le 1$, $\nu\ge 0$, and $\nu(1-f_1-f_2)=0$.

\paragraph{Adjoint equations.}
Since $F_i'=-f_i$ and $\partial H/\partial F_i = -\frac{w}{c_i}F_i$, the costates satisfy
$$
\lambda_i'(\alpha)=-\frac{\partial H}{\partial F_i}(\alpha)=\frac{w}{c_i}F_i(\alpha),\qquad i=1,2.
$$

\paragraph{Transversality.}
The terminal state is fixed ($F_i(1)=0$), hence $\lambda_i(1)$ is free.
The initial state $F_i(0)$ is free and there is no endpoint cost at $\alpha=0$, hence
$
\lambda_i(0)=0$, $i=1,2$.

From $\lambda_i'=\frac{w}{c_i}F_i$ and $\lambda_i(0)=0$, we can write $\lambda_i(\alpha)$ as
$
\lambda_i(\alpha)=\frac{w}{c_i}\int_{0}^{\alpha} F_i(s)\mathrm{d}s
$, 
so the inequality form \eqref{eq:ineq_KKT} reads
$$
\nu(\alpha) \geq r -f_i(\alpha) -\frac{w}{c_i}\int_{0}^{\alpha}F_i(t)\mathrm{d}t,
\quad \text{with equality when } f_i(\alpha)>0, i=1,2.
$$

Therefore, the optimality conditions are expressed as

\begin{equation}\label{eq:PMP_condition}
\left\{
\begin{aligned}
&\nu(\alpha) \geq r -f_i(\alpha) -\frac{w}{c_i}\int_{0}^{\alpha}F_i(t)\mathrm{d}t,
\quad \text{with equality when } f_i(\alpha)>0, i=1,2,\\
&\nu(\alpha)\left[1-f_1(\alpha)-f_2(\alpha)\right]=0 . \\ &f_1(\alpha) \geq 0, f_2(\alpha) \geq 0,\nu(\alpha) \geq 0,\\  
\end{aligned}
\right.
\end{equation}

Note that the feasible set is convex and the objective functional is strictly concave in $(f_1,f_2)$. To see this, take a small permutation $h\in L^2([0,1])$ and set $F_h(\alpha)=\int_\alpha^1 h(t)dt$. 
Note $F_h\in L^2([0,1])$ and $\|F_h\|_2^2\le \tfrac12\|h\|_2^2$ by Cauchy--Schwarz and Fubini. Consider the perturbation $f+\varepsilon h$ with a small $\varepsilon$ and compute the objective $J_i$:
$$
\begin{aligned}
J_i(f_i+\varepsilon h)
&=\int_0^1\Big(r(f_i+\varepsilon h)-\tfrac12(f_i+\varepsilon h)^2-\tfrac{w}{2c_i}(F_i+\varepsilon F_h)^2\Big)d\alpha\\
&=J_i(f_i)+\varepsilon\int_0^1\Big(rh-fh-\tfrac{w}{c_i}F_i F_h\Big)d\alpha
-\frac{\varepsilon^2}{2}\int_0^1\Big(h^2+\tfrac{w}{c_i}F_h^2\Big)d\alpha .
\end{aligned}
$$
Therefore the second variation at $f$ in the direction $h$ is 
$$
D^{2}J_i(f)[h,h]=-\int_0^1\Big(h^2+\tfrac{w}{c_i}F_h^2\Big)d\alpha <0\quad\text{for all }h\neq 0.
$$

Hence, $J_i$ is strictly concave.
As a result, a feasible pair satisfying all conditions in \eqref{eq:PMP_condition} (a.e.\ $\alpha$) is the global maximizer.

Next, we check the optimality conditions for our constructed solution $\tilde{f}_i$. Let $G_1(\alpha) = r -\tilde{f}_1(\alpha) -\frac{w}{c_1}\int_{0}^{\alpha}\tilde{F}_1(t)\mathrm{d}t$, $G_2(\alpha) =r -\tilde{f}_2(\alpha) -\frac{w}{c_2}\int_{0}^{\alpha}\tilde{F}_2(t)\mathrm{d}t $. 
By definition, one can verify that for all $\alpha\neq\hat{\alpha}$,
\begin{equation}
G_1'(\alpha)  =  - \tilde{f}''_1(\alpha) + \frac{w}{c_1}\tilde{f}_1(\alpha)=- \tilde{f}''_1(\alpha)+\frac{w}{c_2}\tilde{f}_1(\alpha)= G_2'(\alpha) .
\end{equation}

Next, we would like to check the semi-derivative of $\tilde{f}_1$ and $\tilde{f}_2$ at point $\hat{\alpha}$. 
\begin{subequations}
\begin{align}
\partial_{+}G_1(\hat{\alpha}) =& - \partial_{+}\tilde{f}_1(\hat{\alpha}) - \frac{w}{c_1}\tilde{F}_1(\hat{\alpha}) = 0 , \label{subeq1}\\
\partial_{+}G_2(\hat{\alpha}) =&- \partial_{+}\tilde{f}_2(\hat{\alpha}) - \frac{w}{c_1}\tilde{F}_2(\hat{\alpha}) = 0 , \label{subeq2}\\
\partial_{-}G_1(\hat{\alpha}) =&- \partial_{-}\tilde{f}_1(\hat{\alpha}) - \frac{w}{c_1}\tilde{F}_1(\hat{\alpha})\\
=& -\frac{\lambda}{\sinh(\lambda\hat\alpha)}
\Biggl[
\Bigl(\kappa - \frac{c_1}{c_1 + c_2}\Bigr)\cosh(\lambda\hat\alpha)
- \frac{c_2 - c_1}{2(c_1 + c_2)}
\Biggr]
-\kappa\sqrt{\frac{w}{c_1}}
\frac{\sinh\left(\sqrt{\frac{w}{c_1}}(1-\hat{\alpha})\right)}{\cosh\left(\sqrt{\frac{w}{c_1}}(1-\hat{\alpha})\right)} , \label{subeq3}\\
\partial_{-}G_2(\hat{\alpha}) =& - \partial_{-}\tilde{f}_2(\hat{\alpha}) - \frac{w}{c_1}\tilde{F}_2(\hat{\alpha}) \\
=&  -\frac{\lambda}{\sinh(\lambda\hat\alpha)}
\Biggl[
\Bigl(\frac{c_1}{c_1 + c_2} - \kappa\Bigr)\cosh(\lambda\hat\alpha)
- \frac{c_1 - c_2}{2(c_1 + c_2)}
\Biggr]-(1-\kappa)\sqrt{\frac{w}{c_2}}\frac{\sinh\left(\sqrt{\frac{w}{c_2}}(1-\hat{\alpha})\right)}{\cosh\left(\sqrt{\frac{w}{c_2}}(1-\hat{\alpha})\right)}. \label{subeq4}
\end{align}
\end{subequations}

\eqref{subeq3} and \eqref{subeq4} takes same value due to our construction of $\kappa$. In addition
\begin{equation}\label{eq:G_1}
G_1(\hat{\alpha}) = r - \kappa 
-\frac{w}{c_1}\int_{0}^{\hat{\alpha}}\int_{\alpha}^{\hat{\alpha}}\tilde{f}_1(t)\mathrm{d}t\mathrm{d}\alpha
\end{equation}
substituting \eqref{eq:kappa} into \eqref{eq:G_1} and rearranging terms yields  $G_1(\hat{\alpha}) = r-R(\hat{\alpha}) =0$. By an analogous calculation, $G_2(\hat\alpha) = 0$. Hence, both functions are valued $0$ at $\hat\alpha$, share the same one‐sided derivative there, and possess identical second‐order derivatives on the interval $[0,1]$. Furthermore, it is clear that $G_1$ and $G_2$ is decreasing functions and $G_1(0)=G_2(0)=r-1/2>0$, so they are nonnegative. Therefore, when $\alpha \geq \hat{\alpha}$, $\nu(\alpha)=0$, when $\alpha \leq \hat{\alpha}$, $\nu(\alpha)=G_1(\alpha)=G_1(\alpha)\geq0$ and $\tilde{f}_1(\alpha)+\tilde{f}_2(\alpha)=1$. Therefore, all conditions in \eqref{eq:PMP_condition} are satisfied, establishing the optimality of $\tilde f_1(\alpha)$ and $\tilde f_2(\alpha)$.

\textbf{Case:3}
Let $\lambda = \sqrt{\frac{w(c_1+c_2)}{2c_1c_2}}$. When 

\begin{equation}\label{eq:r_case3}
r\geq 
\frac{c_1^2+c_2^2}{(c_1+c_2)^2} +\frac{w}{2(c_1+c_2)} -\frac{(c_1-c_2)^2}{2(c_1+c_2)^2\cosh(\lambda)}.
\end{equation}
Define
\begin{equation*}
\begin{aligned}
\tilde{f}_1(\alpha ) =& \frac{c_2-c_1}{2(c_1 + c_2)}\frac{\cosh(\lambda(1-\alpha))}{\cosh(\lambda)} + \frac{c_1}{c_1 + c_2}, \\
\tilde{f}_2(\alpha ) = & \frac{c_1-c_2}{2(c_1 + c_2)}\frac{\cosh(\lambda(1-\alpha))}{\cosh(\lambda)} + \frac{c_2}{c_1 + c_2}. 
\end{aligned}
\end{equation*}
It is clear that $\tilde{f}_1(\alpha )+\tilde{f}_2(\alpha ) = 1$ for all $\alpha$. On the other hand, due to the assumption of $r$ in \eqref{eq:r_case3}, we have
\begin{equation*}
\begin{aligned}
&r -f_1(\alpha) -\frac{w}{c_1}\int_{0}^{\alpha}F_1(t)\mathrm{d}t =r -f_2(\alpha) -\frac{w}{c_2}\int_{0}^{\alpha}F_2(t)\mathrm{d}t\\
=& r + \frac{(c_1-c_2)^2 \cosh\left(\lambda(1-\alpha)\right)}{2(c_1+c_2)^2\cosh(\lambda)} -\frac{c_1^2+c_2^2}{(c_1+c_2)^2}- \frac{w}{c_1 + c_2}\left(\alpha - \frac{\alpha^2}{2}\right)\\
\geq& r + \frac{(c_1-c_2)^2 }{2(c_1+c_2)^2\cosh(\lambda)} -\frac{c_1^2+c_2^2}{(c_1+c_2)^2}- \frac{w}{c_1 + c_2}\frac{1}{2}\\
\geq&  0.
\end{aligned}
\end{equation*}

Therefore, the conditions \eqref{eq:PMP_condition} hold hence $f_1=\tilde{f}_1$ and $f_2=\tilde{f}_2$ are optimal.\hfill\qedsymbol
\end{proof}
\bibliographystyleec{informs2014}
\bibliographyec{EC_reference}